\documentclass{article}
\usepackage[T1]{fontenc}
\usepackage[margin=1in]{geometry}

\usepackage{amsmath,amssymb,amsthm,mathtools}
\usepackage[cal=euler]{mathalpha}
\usepackage{cochineal}
\usepackage{booktabs,microtype,enumitem}
\usepackage{graphicx}
\usepackage{hyperref}
\hypersetup{colorlinks=true,linkcolor=blue,citecolor=blue,urlcolor=blue}
\usepackage{microtype}
\usepackage[numbers,sort&compress]{natbib}

\newtheorem{theorem}{Theorem}[section]
\newtheorem{proposition}[theorem]{Proposition}
\newtheorem{lemma}[theorem]{Lemma}
\newtheorem{corollary}[theorem]{Corollary}
\newtheorem{conjecture}[theorem]{Conjecture}
\theoremstyle{definition}

\theoremstyle{remark}

\DeclareMathOperator{\Var}{Var}
\DeclareMathOperator{\Cov}{Cov}
\DeclareMathOperator{\supp}{supp}

\newcommand{\R}{\mathbb{R}}
\newcommand{\Q}{\mathbb{Q}}
\newcommand{\N}{\mathbb{N}}
\newcommand{\E}{\mathbb{E}}

\newcommand{\Cdiv}{\mathsf{C}}     
\newcommand{\Renent}{H^{\mathrm{R}}}
\newcommand{\Dren}{D}              
\newcommand{\Tdiv}{D^{T}}          
\newcommand{\KLdiv}{D_{1}}         
\newcommand{\simplex}{\Delta_{W}}
\newcommand{\Aplus}{\mathcal{A}_{+}}
\newcommand{\Aminus}{\mathcal{A}_{-}}
\newcommand{\Bminus}{\mathcal{B}_{-}}
\newcommand{\paramset}{\widehat{\mathcal{A}}}
\newcommand{\Sym}{\mathfrak{S}_{W}}
\newcommand{\dnu}{d\nu}

\newcommand{\one}{\mathbf{1}}

\usepackage{xcolor}

\providecommand{\akshay}[1]{\authcmt{red}{[AB: #1]}}
\providecommand{\vac}[1]{\authcmt{blue}{[VAC: #1]}}
\providecommand{\authcmt}[2]{\textcolor{#1}{#2}}

\makeatletter
\@ifundefined{argmax}{}{}
\@ifundefined{argmin}{}{}
\@ifundefined{diam}{}{}
\@ifundefined{problem}{}{}
\makeatother

\providecommand{\bbN}{\mathbb{N}}

\title{All you need is log}
\author{Akshay Balsubramani \\ {\small \texttt{akshay@vac.bio}}}
\date{}

\makeatletter
\def\bvapx@title{}
\AtBeginDocument{\providecommand{\phantomsection}{}}
\let\bvapx@maketitle\maketitle
\def\maketitle{\global\let\bvapx@title\@title\bvapx@maketitle}
\newcommand{\appendixtitle}{%
  \clearpage
  \suppressfloats[t]%
  \phantomsection
  \begin{center}
    {\LARGE\bfseries Appendices\ifx\bvapx@title\@empty\else\space to ``\bvapx@title''\fi\par}%
  \end{center}
  \vspace{1.5\baselineskip}%
}
\makeatother

\providecommand{\akshay}[1]{}\renewcommand{\akshay}[1]{}%
\providecommand{\vac}[1]{}\renewcommand{\vac}[1]{}%
\begin{document}
\maketitle

\begin{abstract}
How different are several probability distributions from one another?
For two distributions the standard answer is the family of R\'enyi divergences, singled out by two natural requirements: processing the data never makes distributions easier to tell apart, and independent repetitions add.
Many problems in learning and statistics compare more than two distributions at once, such as testing among several hypotheses or bounding generalization against several priors.
The same two requirements leave one kind of building block, built on a coincidence probability: how unlikely it is that independent samples, one from each distribution, all show the same empirical distribution.
The logarithm is forced because repetitions add, which is already visible for a single experiment repeated.
This characterization is known in greater generality, and this paper is about the meaning of its building blocks.
On a finite alphabet, each building block indexed by a rational point of the simplex is the exponential rate of that coincidence as the samples grow in fixed proportions.
Each is also the limiting free energy of Bayesian inference over distributions.
At any amount of data, the free energy of the posterior is the coincidence measure plus two costs: the expected distance from a posterior draw to the most likely distribution, and the information gained per unit of data.
Both costs vanish as data accumulate.
When the comparison is conditioned on side information, every kind of building block has a conditional counterpart, and the coincidence ones alone do not suffice.
\end{abstract}


\section{Introduction}\label{sec:intro}

Comparing probability distributions is a basic operation in statistics and machine learning.
It underlies hypothesis testing, generalization bounds, differential privacy, model selection, and robustness, among many other tasks.
For two distributions there is a standard family of comparisons, the R\'enyi divergences, and two natural requirements single it out.
The first is \emph{monotonicity under data processing}: transforming the data can only forget information, so it never makes two distributions easier to tell apart.
The second is \emph{additivity on independent products}: comparing two independent experiments together gives the sum of comparing them separately, as repeated sampling should.
The relative entropy is the most familiar member of the family.

Many problems compare more than two distributions at once.
A fairness audit compares a model's behavior across several demographic groups.
A PAC--Bayes bound or a Bayesian model selection weighs the data against several priors.
A multi-hypothesis test asks how quickly several candidate sources can be told apart.
A differential-privacy analysis may need guarantees across several neighboring datasets, where the additivity of the R\'enyi divergence makes privacy accounting work~\cite{mironov2017rdp}.
Earlier work extends the $f$-divergences to several distributions and ties them to the error of classifying among several classes~\cite{garcia2012divergences}, and that framework respects data processing but does not ask for additivity over repetitions.
The question here asks for both: which ways of comparing several distributions respect data processing, add over independent repetitions, and vanish when all the distributions coincide?

The simplest case already shows why the answer is a logarithm.
If a comparison adds over independent repetitions, then repeating an experiment $n$ times multiplies the comparison by $n$.
That scaling is Cauchy's functional equation, and once monotonicity rules out its pathological solutions, the only one left is a logarithm.
The logarithm is therefore forced, with any number of distributions as with two, and Section~\ref{sec:functionaleq} makes the argument precise.

The logarithm is applied to a Hellinger transform.
Every such comparison $D$ is a positive integral of rescaled members of the $(W{-}1)$-parameter family of \emph{multi-way coincidence divergences}
\[
\Cdiv_{\alpha}(\pi_{1},\dots,\pi_{W}) := -\log\E_{x\sim\nu}\!\Big[\textstyle\prod_{k=1}^{W}\pi_{k}^{\alpha_{k}}(x)\Big],
\qquad \sum_{k}\alpha_{k}=1
\]
together with the boundary atoms that arise as limits of these basic terms.
The exponent vector $\alpha=(\alpha_{1},\dots,\alpha_{W})$ lives on the affine plane $\{\alpha\in\R^{W}:\sum_{k}\alpha_{k}=1\}$; on the simplex of geometric mixtures (all $\alpha_{k}\in[0,1]$),
$\Cdiv_{\alpha}$ is the negative log of the multi-distribution Bhattacharyya--Chernoff coefficient.
For $\alpha$ in the simplex with rational coordinates and distributions with full support on a finite alphabet, $\Cdiv_{\alpha}$ is the exponential rate at which independent samples, one from each distribution with sizes in the proportions $\alpha$, all have the same empirical distribution (Proposition~\ref{prop:coincidence-rate}).
The argument inside the log,
$H_{\alpha}=\E_{\nu}[\prod_{k}\pi_{k}^{\alpha_{k}}]$, is the multi-distribution generalization of the bivariate Bhattacharyya--Chernoff coefficient, and on the simplex it is Le~Cam's Hellinger transform of the experiment.

The index set is rigid in both directions.
It has four strata: the simplex $\Aplus$ without its vertices; the mixed-sign exponent cones $\Aminus$, where one $\alpha_{l}>1$ and the rest are $\le 0$; the tropical boundary at infinity; and the pairwise relative entropies at the simplex vertices.
The cones are the multi-distribution analogue of R\'enyi orders $>1$, and the tropical boundary contains the worst-case max-divergences.
None of the strata can be dropped: for each of them, some symmetric divergence satisfying the axioms has no representation that leaves that stratum without mass (Corollary~\ref{cor:necessity}).
Nor can the index set be larger, since every divergence satisfying the axioms is a positive integral over these strata, with a representing measure that is unique once the tropical directions are normalized (Proposition~\ref{prop:uniqueness}).
The tropical and relative-entropy strata are limits of the coincidence atoms $\Cdiv_{\alpha}$ (Section~\ref{ssec:limits}): pairwise relative entropies as derivative-style limits at the simplex vertices, max-divergences as zero-temperature scaling limits along the mixed-sign cones.
An explicit compactification of the index set turns these limits into continuity at its boundary (Proposition~\ref{prop:compactification}).
Theorem~\ref{thm:correct} states the representation and Corollary~\ref{cor:converse} its converse.

This representation is established for finite alphabets by~\cite{haapasalo2025barycentric}, in a framework that covers classical and quantum multivariate divergences by way of preordered-semiring machinery, with priority for the classical multivariate case theirs.
A subsequent companion~\cite{haapasalo2026continuous} extends the same classical characterization to continuous observation spaces.
An existence theorem of that kind fixes the form of the atoms without assigning them probabilistic content.
That probabilistic content is this paper's subject: the atoms as events, the family as a transform, and the behavior of the calculus when a learning problem supplies side information.
Section~\ref{sec:catmarkov} sets out the relationship to that line of work in full, alongside the rest of the surrounding literature.

On a finite alphabet, each atom indexed by a rational point of the simplex is the exponential rate of the coincidence described above (Proposition~\ref{prop:coincidence-rate}).
The family is the multivariate Laplace transform of the log-loss vector, and its maximum over the simplex is a worst-case Kullback projection radius (Section~\ref{sec:readings}).
The sharpest reading is Bayesian: each atom is the exact free energy of a posterior over distributions, less two information costs that vanish as data accumulate, which are the expected relative entropy from a posterior draw to the typical distribution and the information gained per unit of resolution.
The identity holds exactly at every resolution (Theorem~\ref{thm:laplace-mixed-concentration}).
As resolution grows, the posterior satisfies a level-2 large-deviation principle whose rate function is the relative entropy to the typical distribution (Corollary~\ref{cor:level2-ldp}).
When the reference measure on distributions has a density that is continuous and positive at the typical distribution, the expected relative entropy falls as $(|\mathcal{X}|-1)/(2t)$ to first order (Proposition~\ref{prop:laplace-rates}).
Section~\ref{sec:conditional} extends the calculus to conditioned experiments, where every stratum reappears with an aggregator over the side information whose admissible range depends on the stratum.

The proof of the representation given here was obtained independently, and its monotonicity lemma uses the tensoring argument of~\cite{haapasalo2025barycentric}.
The proof builds an explicit compactification of the four-stratum index set, uses the large-sample theorem of~\cite{farooq2024matrix} to make every such divergence monotone in the values of the atoms, and obtains the representing measure from the Riesz--Markov theorem (Section~\ref{sec:proof}).
The proof also shows that the representing measure is unique and that every stratum is needed, and with a trivial conditioning variable it shows that the conditional divergences need every stratum as well (Corollary~\ref{cor:cond-forced}).
The representation needs its boundedness hypothesis (Proposition~\ref{prop:unbounded}), and for three states the conditions that~\cite[Appendix~K]{mu2021blackwell} conjectured to be sufficient for large-sample dominance are not (Proposition~\ref{prop:mpst-counterexample}).
For a single distribution the same logarithm is reached without data-processing monotonicity, since the Kolmogorov--Nagumo classification of quasi-arithmetic means together with R\'enyi's mean-style axiomatics yields the R\'enyi entropies.
Multi-lottery betting gives an operational reading of the rescaled atoms, and the multi-hypothesis Bayes error exponent is an envelope of the atoms (Section~\ref{sec:w3}).

\section{Setup and the simplex-only special case}\label{sec:setup}

\subsection{Notation}\label{ssec:notation}
Throughout, $W\ge 2$ is a fixed number of priors and $(\mathcal{X},\nu)$ is a Polish observation space with a fixed dominating reference measure $\nu$.
A \emph{$W$-tuple of distributions} is an ordered list $\boldsymbol{\pi}=(\pi_{1},\dots,\pi_{W})$ of probability measures on $\mathcal{X}$,
all absolutely continuous with respect to $\nu$.
The same symbol $\pi_{k}$ refers interchangeably to a measure and its $\nu$-density (the meaning is unambiguous in context).

The componentwise product $\boldsymbol{\pi}\otimes\boldsymbol{\pi}'$ acts on $\mathcal{X}\times\mathcal{X}'$ with reference $\nu\otimes\nu'$.
The componentwise pushforward under a Markov kernel $K:\mathcal{X}\to\mathcal{Y}$ is written $K\boldsymbol{\pi}=(K\pi_{1},\dots,K\pi_{W})$.

The exponent vector $\alpha=(\alpha_{1},\dots,\alpha_{W})\in\R^{W}$ that indexes the simplex/cone atoms lies in the \emph{affine slice}
$\mathcal{A} := \{\alpha\in\R^{W}:\sum_{k}\alpha_{k}=1\}$;
we write $\alpha_{\star}:=\max_{k}\alpha_{k}$ for its largest component (so $\alpha_{\star}\in[1/W,1]$ on the simplex $\Aplus$ and $\alpha_{\star}\ge 1$ on $\Aminus$).
The direction vector $\beta=(\beta_{1},\dots,\beta_{W})\in\R^{W}$ that indexes the tropical atoms instead lies on the \emph{tangent plane} to the slice, $\sum_{k}\beta_{k}=0$; we write $\beta_{\star}:=\max_{k}\beta_{k}$ (with $\beta_{\star}>0$ for $\beta\ne 0$, since the components sum to zero and not all are zero).
Geometrically, $\mathcal{A}$ is the affine slice and the $\beta$-vectors parametrize rays $\alpha+t\beta$ inside it.

For divergence functionals we use two related notations.
\begin{itemize}[leftmargin=*]
\item The \emph{coincidence divergence}
\[
\Cdiv_{\alpha}(\boldsymbol{\pi})\;:=\;-\log\E_{\nu}\!\Big[\textstyle\prod_{k}\pi_{k}^{\alpha_{k}}\Big]
\;=\;-\log H_{\alpha}(\boldsymbol{\pi})
\]
is the non-normalized form.
The argument $H_{\alpha}(\boldsymbol{\pi})=\E_{\nu}[\prod_{k}\pi_{k}^{\alpha_{k}}]$ of the logarithm is the mixed partition function $Z(\boldsymbol{\alpha})$ studied in~\cite{balsubramani2026information},
which develops it as the Boltzmann coincidence weight of a multi-way independent-draws experiment, the normalizer of the geometric mixture $p_{\boldsymbol{\alpha}}^{\star}\propto\prod_{k}\pi_{k}^{\alpha_{k}}$,
and, through its logarithm, the value of an unconstrained max-entropy Lagrangian $\max_{p}[\text{H}(p)-\sum_{k}\alpha_{k}\text{H}(p,\pi_{k})]$.
The coincidence divergence is non-negative on the simplex $\alpha\in\Aplus$ (where Jensen's inequality gives $H_{\alpha}\le 1$) and non-positive on $\Aminus$ (where Jensen's inequality with mixed-sign exponents gives $H_{\alpha}\ge 1$).
\item The normalized \emph{matrix R\'enyi atom}~\cite{farooq2024matrix}
\[
\Dren_{\alpha}(\boldsymbol{\pi})\;:=\;\frac{1}{\alpha_{\star}-1}\log H_{\alpha}(\boldsymbol{\pi})
\;=\;\frac{-\Cdiv_{\alpha}(\boldsymbol{\pi})}{\alpha_{\star}-1}
\]
rescales $\Cdiv_{\alpha}$ by the signed scalar $-1/(\alpha_{\star}-1)=1/(1-\alpha_{\star})$, which is positive on $\Aplus$ (where $\alpha_{\star}<1$) and negative on $\Aminus$ (where $\alpha_{\star}>1$).
The two sign-flips combine: $\Dren_{\alpha}\ge 0$ on the entire signed-exponent set $(\Aplus\cup\Aminus)\setminus E$, where $E:=\{e_{1},\dots,e_{W}\}$ is the vertex set.
\end{itemize}
The atom is $\Cdiv_{\alpha}=-\log H_{\alpha}$: a single logarithm of the mixed partition function, the multi-distribution Bhattacharyya--Chernoff coefficient,
the form used throughout the paper.
The matrix R\'enyi atom $\Dren_{\alpha}$ is the rescaled form in which the borrowed spectral input is stated.
The matrix-Blackwell spectrum theorems~\cite{farooq2024matrix} prove the spectral characterization for $\Dren_{\alpha}$, which is non-negative on the entire signed-exponent set, whereas $\Cdiv_{\alpha}$ undergoes the simplex/cone sign change.
The two differ only by the signed rescaling above, agree in sign on the simplex interior, and the Choquet representation (Theorem~\ref{thm:correct}) is therefore stated in $\Dren_{\alpha}$, where the integral against a positive measure is over a non-negative integrand.

\paragraph{The role of the three axioms.}
\emph{Joint data-processing monotonicity} captures the operational primitive that no single transformation of the $W$ distributions can increase the resolution between them. It is the multi-distribution generalization of Blackwell's classical statement for two distributions,
applied uniformly across all coordinates.
\emph{Additivity on independent products} encodes the law-of-large-numbers content: the $n$-fold tensor power of an experiment scales the divergence linearly,
which is the property a useful ``information measure'' must have if it is to track repeated independent trials.
The \emph{coincidence ground state} is the minimal requirement that the divergence vanishes when there is nothing to discriminate.
Together the axioms force the calculus to factor through $-\log H_{\alpha}$ (Theorem~\ref{thm:correct}).
Dropping additivity admits the multi-prior $f$-divergences, while the ground state already follows from the other two axioms and finiteness on bounded tuples (Section~\ref{ssec:relaxing}).
The $W=2$ specialization of these axioms is precisely the hypothesis of~\cite{mu2021blackwell}, and Theorem~\ref{thm:mpst} is the two-distribution case of Theorem~\ref{thm:correct}.

A map $D:\boldsymbol{\pi}\mapsto D(\boldsymbol{\pi})\in[0,\infty]$ defined on a class of $W$-tuples (closed under pushforward and product) is a \emph{$W$-way DPI--additive divergence} (\cite{gour2021entropy}) when it satisfies three structural properties\label{def:axioms}. (a) \emph{joint DPI}: for every Markov kernel $K$ acting identically on each component, $D(K\boldsymbol{\pi}) \le
D(\boldsymbol{\pi})$, so simultaneous data processing on every coordinate can
only erase distinguishability. (b) \emph{additivity on products}:
$D(\boldsymbol{\pi}\otimes\boldsymbol{\pi}') = D(\boldsymbol{\pi}) +
D(\boldsymbol{\pi}')$, the law-of-large-numbers content of paired experiments. And (c)
\emph{coincidence ground state}: $D(\pi,\pi,\dots,\pi)=0$ for all $\pi$, with $D$ finite on bounded tuples (those whose log-likelihood ratios $\log(\pi_{k}/\pi_{\ell})$ are uniformly bounded on $\supp\nu$ for every $k\ne\ell$).
The divergence is \emph{symmetric} if it is moreover invariant under joint permutation,
$D(\pi_{\sigma(1)},\dots,\pi_{\sigma(W)})=D(\pi_{1},\dots,\pi_{W})$ for every $\sigma\in\Sym$.

The four-stratum index space of Theorem~\ref{thm:correct} comprises the simplex without its vertices, the mixed-sign cones, the tropical boundary at infinity, and the pairwise relative-entropy vertex edges.
The simplex-restricted family $\{\Cdiv_{\alpha}\}_{\alpha\in\Aplus}$ is the obvious first guess.
It is contained in the case where the three boundary families receive no mass:

\begin{conjecture}[$W$-distribution simplex-only form]\label{conj:naive}
Every symmetric, $W$-way DPI--additive divergence $D$ on bounded tuples admits a representation
\[
D(\boldsymbol{\pi}) = \int_{\simplex} \Cdiv_{\alpha}(\boldsymbol{\pi})\, dm(\alpha)
\]
for a finite, $\Sym$-invariant Borel measure $m$ on the standard simplex $\simplex=\{\alpha\in[0,1]^{W}:\sum_{k}\alpha_{k}=1\}$.
\end{conjecture}

The simplex-only form is false already for two distributions: for each of the three boundary families (mixed-sign cones, tropical directions and relative-entropy vertex edges), some symmetric divergence has no representation that gives the family zero mass (Section~\ref{sec:exotic}).
Theorem~\ref{thm:correct} keeps the integral shape over the full index space, where the tropical and relative-entropy atoms arise as limits of rescaled $\Cdiv_{\alpha}$.

\section{The forcing of the logarithm}\label{sec:functionaleq}

The logarithm in $\Cdiv_{\alpha}=-\log H_{\alpha}$ is not a convention: it is the only wrapping the axioms permit, and the reason is visible without the representation theorem.
For a $W$-way DPI--additive divergence $D$, additivity on tensor products gives
\begin{equation}\label{eq:cauchy-like}
D(\boldsymbol{\pi}^{\otimes n}) = n\,D(\boldsymbol{\pi})
\end{equation}
along the ray $n\mapsto\boldsymbol{\pi}^{\otimes n}$, the additive form of Cauchy's functional equation in $n$, while the Hellinger transform is multiplicative, $H_{\alpha}(\boldsymbol{\pi}^{\otimes n})=H_{\alpha}(\boldsymbol{\pi})^{n}$.
A divergence that depends on the tuple only through one Hellinger transform must therefore turn a product into a sum, and monotonicity under data processing leaves the logarithm as the only way to do so.

\begin{proposition}[The logarithm is forced]\label{prop:log-forced}
Let $\alpha\in\Aplus\setminus E$, and let $D$ be a $W$-way DPI--additive divergence with $D(\boldsymbol{\pi})=\varphi(H_{\alpha}(\boldsymbol{\pi}))$ for every bounded $W$-tuple $\boldsymbol{\pi}$, for some function $\varphi:(0,1]\to[0,\infty)$.
Then $\varphi(h)=-c\log h$ for a constant $c\ge 0$, so $D=c\,\Cdiv_{\alpha}$ on bounded tuples.
\end{proposition}

The proof is in Appendix~\ref{app:atom-proofs}.
Theorem~\ref{thm:correct} removes the assumption that $D$ depends on a single Hellinger transform: every $W$-way DPI--additive divergence is a positive integral of rescaled logarithms of Hellinger transforms and of their limits.
The bivariate argument of~\cite{mu2021blackwell} reaches its representation through Blackwell dominance (Section~\ref{sec:bg}), and the proof of Theorem~\ref{thm:correct} extends that route to $W$ distributions.
Routes to the same logarithm that do not use data processing, and operational readings of the atoms, are collected in Appendix~\ref{app:lognatural}.

\section{Background: the bivariate R\'enyi-family representation}\label{sec:bg}

For $\mu,\nu$ probability measures on $\mathcal{X}$ and $t\in(0,1)\cup(1,\infty)$,
the R\'enyi divergence of order $t$ is
\begin{equation}\label{eq:renyi}
R_{t}(\mu\|\nu) := \frac{1}{t-1}\log\E_{x\sim\nu}\!\big[(d\mu/d\nu)^{t}(x)\big]
= \frac{1}{t-1}\log\E_{\nu}\!\big[\mu^{t}\nu^{1-t}\big]
\end{equation}
extended by limits to $R_{1}(\mu\|\nu)=\KLdiv(\mu\|\nu)$ and $R_{\infty}(\mu\|\nu)=\log\sup_{x}(d\mu/d\nu)$.
Different positive integrals against this family recover relative entropy and the Bhattacharyya distance, with the Hellinger and $\chi^{2}$ divergences entering through the monotone reparametrizations $R_{1/2}=-2\log(1-\tfrac12 H^{2})$ and $R_{2}=\log(1+\chi^{2})$.

\begin{theorem}[binary MPST, Thm.~2 of~\cite{mu2021blackwell}]\label{thm:mpst}
Let $D(\mu,\nu)$ be a divergence between pairs of distributions on a common Polish space, defined for all bounded pairs (those with $d\mu/d\nu$ bounded above and away from $0$), satisfying joint DPI and additivity on products.
Then there exist finite Borel measures $m_{0},m_{1}$ on $[1/2,\infty]$ such that
\begin{equation}\label{eq:mpst}
D(\mu,\nu) = \int_{[1/2,\infty]} R_{t}(\mu\|\nu)\, dm_{0}(t) + \int_{[1/2,\infty]} R_{t}(\nu\|\mu)\, dm_{1}(t)
\end{equation}
\end{theorem}

Several points recur in the multi-distribution setting.
First, the parameter space is $[1/2,\infty]$ and not $[0,\infty]$ because the reflection identity $R_{t}(\mu\|\nu) = \frac{t}{1-t}R_{1-t}(\nu\|\mu)$ for $t\in(0,1)$ makes $[1/2,\infty]$ a fundamental domain for the involution $t\mapsto 1-t$.
The multi-distribution analogue is the $\Sym$-orbit reduction of Section~\ref{sec:perm}.
Second, the endpoint $t=\infty$ is a genuine boundary point:
masses there contribute the max-divergence $R_{\infty}(\mu\|\nu) = \log\sup_{x}(d\mu/d\nu)$,
a supremum-style object that is structurally distinct from the finite-$t$ integral-style R\'enyi divergences and cannot be recovered from them as a finite-$t$ linear combination.
The multi-distribution analogue is the tropical boundary at infinity.
Third, the two measures $m_{0},m_{1}$ play asymmetric roles; imposing the symmetry $D(\mu,\nu)=D(\nu,\mu)$ collapses them to a single integral against the symmetrized divergence $\frac{1}{2}(R_{t}(\mu\|\nu)+R_{t}(\nu\|\mu))$, a reduction by the orbit of the $\mu\leftrightarrow\nu$ swap.

The proof of Theorem~\ref{thm:mpst} has two halves that together set the template for the multi-distribution generalization.

\emph{Half (A): data-processing monotonicity plus additivity yield monotonicity in the R\'enyi order.} Data-processing monotonicity alone makes $D$ monotone under Blackwell garblings.
Additivity on independent products extends this to monotonicity in the \emph{large-sample} Blackwell order; the identification of that order with $R_{t}(\mu\|\nu)\ge R_{t}(\mu'\|\nu')$ for all $t>0$ (together with the reversed-orientation statement) is the main theorem of~\cite{mu2021blackwell}.

\emph{Half (B): integral representation by functional analysis.}
$D$ is then a monotone, additive functional on a positive cone whose extreme rays are the R\'enyi divergences.
The Riesz--Markov representation theorem yields the integral form~\eqref{eq:mpst}.

The $W$-distribution generalization follows the same template: replace the large-sample theorem of~\cite{mu2021blackwell} by its multi-state analogue for strict inequalities~\cite[Theorem~19]{farooq2024matrix}, and run a Riesz--Markov argument over a compactification of the larger parameter space.

\section{The atom space: multi-way coincidence divergences and their boundary limits}\label{sec:atoms}

The atom is fixed by the two requirements; the index set over which the representation integrates is not the simplex alone.
$\Cdiv_{\alpha}$ is indexed by a parameter space $\paramset$ obtained by enlarging the simplex $\Delta_{W}$ in three directions: signed exponents (the $\Aminus$ region), tropical limits at infinity (the $\Bminus$ region),
and vertex derivations (giving pairwise relative entropies).
The enlargement is forced, since every stratum is needed (Corollary~\ref{cor:necessity}).
It matches the classification of~\cite[Propositions~13--14]{farooq2024matrix}: every \emph{nondegenerate} monotone homomorphism of the matrix-majorization preordered semiring is a Hellinger transform $f_{\alpha}$ with $\alpha\in(\Aplus\cup\Aminus)\setminus E$ or a tropical $f^{T}_{\beta}$ with $\beta\in\Bminus\setminus\{0\}$, and every monotone derivation at a degenerate vertex is a nonnegative combination of pairwise relative entropies.

\subsection{The Hellinger transform and its structural properties}\label{ssec:hellinger}

Le~Cam's \emph{Hellinger transform} of a $W$-tuple $\boldsymbol{\pi}$ is the function of $\alpha\in\R^{W}$ given by
\begin{equation}\label{eq:hellinger}
H_{\alpha}(\boldsymbol{\pi}) := \E_{\nu}\!\Big[\prod_{k=1}^{W}\pi_{k}^{\alpha_{k}}\Big]
= \int\prod_{k=1}^{W}\pi_{k}^{\alpha_{k}}\,d\nu
\end{equation}
defined on the affine slice $\{\sum_{k}\alpha_{k}=1\}$ (and extending naturally to a homogeneous function of $\alpha\in\R^{W}$ by absorbing $\nu$).
It has three structural properties that drive all that follows:
\begin{enumerate}[label=(H\arabic*),leftmargin=*]
\item \textbf{Multiplicativity under products.} $H_{\alpha}(\boldsymbol{\pi}\otimes\boldsymbol{\pi}') = H_{\alpha}(\boldsymbol{\pi})\,H_{\alpha}(\boldsymbol{\pi}')$.
\item \textbf{Monotonicity under garbling.} $H_{\alpha}(K\boldsymbol{\pi}) \ge H_{\alpha}(\boldsymbol{\pi})$ for $\alpha\in\Aplus$ and any Markov kernel $K$ (by the conditional Jensen inequality, Appendix~\ref{app:atom-proofs});
the inequality is reversed on $\Aminus$.
\item \textbf{Permutation equivariance.} $H_{\sigma\cdot\alpha}(\sigma\cdot\boldsymbol{\pi}) = H_{\alpha}(\boldsymbol{\pi})$ for $\sigma\in\Sym$.
\end{enumerate}
Property (H1) is the source of additivity; (H2) is the source of DPI; (H3) is the source of permutation symmetry.
The functional $\Cdiv_{\alpha}=-\log H_{\alpha}$ inherits all three, with the inequalities flipping appropriately.

By~\cite{cam1986} (Theorem~9.4 in Chapter~9; see also the textbook treatment of comparison of experiments in~\cite{torgersen1991}),
$H_{\alpha}|_{\alpha\in\Aplus}$ is a \emph{full invariant} of the experiment: two $W$-tuples are Blackwell-equivalent iff their simplex-restricted Hellinger transforms agree.
The matrix-Blackwell spectrum theorems~\cite{farooq2024matrix} sharpen this to the large-sample setting and extend the parameter range from $\Aplus$ to $\Aplus\cup\Aminus\cup\Bminus$.
The multi-way coincidence calculus is canonical because it is the negative log of the Le~Cam invariant for additive comparison of experiments.

\subsection{The signed-exponent affine slice}

$\Cdiv_{\alpha}$ is initially defined on the simplex $\Delta_{W}$, but the binary case already shows that signed exponents are needed: $R_{t}(\mu\|\nu)$ for $t>1$ corresponds to $\alpha=(t,1-t)$ with one negative entry, the regime that controls error exponents in the easy-error regime of hypothesis testing.
The natural enlarged parameter set is the affine slice
\[
\mathcal{A} := \big\{\alpha\in\R^{W}:\textstyle\sum_{k=1}^{W}\alpha_{k}=1\big\}
\]
with two distinguished sub-regions:
\begin{align}
\Aplus  &:= \{\alpha\in\mathcal{A}:\alpha_{k}\ge 0\ \forall k\} = \simplex\\
\Aminus &:= \bigcup_{k=1}^{W}\big\{\alpha\in\mathcal{A}:\alpha_{k}\ge 1\ \text{and}\ \alpha_{\ell}\le 0\ \forall\ell\ne k\big\}
\end{align}
The set $\Aplus$ is the closed simplex; $\Aminus$ is a union of $W$ closed cones, one emerging from each vertex $e_{k}$.
The vertex points $\{e_{1},\dots,e_{W}\}=\Aplus\cap\Aminus$ are the \emph{degenerate} parameter values where $\Cdiv_{\alpha}$ becomes trivial (it equals $-\log\int\pi_{k}=0$ at $\alpha=e_{k}$).

For $\alpha\in(\Aplus\cup\Aminus)\setminus\{e_{1},\dots,e_{W}\}$, define
\begin{equation}\label{eq:Cdiv-extended}
\Cdiv_{\alpha}(\boldsymbol{\pi}) := -\log\E_{x\sim\nu}\!\Big[\prod_{k=1}^{W}\pi_{k}^{\alpha_{k}}(x)\Big],
\qquad
\Dren_{\alpha}(\boldsymbol{\pi}) := \frac{1}{\alpha_{\star}-1}\log\E_{x\sim\nu}\!\Big[\prod_{k=1}^{W}\pi_{k}^{\alpha_{k}}(x)\Big]
\end{equation}
where $\alpha_{\star} := \max_{k}\alpha_{k}$.
The first formula is the negative log of the ``mixed coincidence partition function''; the second is the \emph{matrix R\'enyi divergence}~\cite{farooq2024matrix} (also the multivariate R\'enyi divergence appearing in the multi-lottery betting framework of~\cite{ducuara2026}).
They differ only by the scalar $1-\alpha_{\star}$, which is positive on $\Aplus$, where $\alpha_{\star}<1$, and negative on $\Aminus$, where $\alpha_{\star}>1$:
\begin{equation}\label{eq:Cdiv-Dren}
\Cdiv_{\alpha} = (1-\alpha_{\star})\,\Dren_{\alpha}\quad\text{on }\Aplus\setminus\{e_{k}\},\qquad
\Cdiv_{\alpha} = -(\alpha_{\star}-1)\,\Dren_{\alpha}\quad\text{on }\Aminus\setminus\{e_{k}\}
\end{equation}
$\Dren_{\alpha}$ is nonnegative on the entire signed-exponent set $(\Aplus\cup\Aminus)\setminus\{e_{k}\}$.
The coincidence \emph{atom} is $\Cdiv_{\alpha}=-\log H_{\alpha}$; the integral representation (Theorem~\ref{thm:correct}) is stated in its non-negative rescaling $\Dren_{\alpha}$, which on $\Aplus$ agrees with $\Cdiv_{\alpha}$ up to the smooth positive scalar above.

\subsection{The two extra families of atoms: tropical and relative-entropy}

The simplex/cone family $\{\Dren_{\alpha}\}_{\alpha\in(\Aplus\cup\Aminus)\setminus E}$
(here $E:=\{e_{1},\dots,e_{W}\}$) does not exhaust the cone of DPI--additive divergences.
There are two further families that arise as natural \emph{boundary/limit} atoms.

\paragraph{Tropical (Chernoff-$\infty$) atoms.}
The binary R\'enyi family has the endpoint $R_{\infty}(\mu\|\nu)=\log\sup_{x}\mu/\nu$,
a max functional living at the boundary $t\to\infty$.
The multi-way analogue replaces a single ratio by a product of ratios with weighted exponents; the $\sup$ structure is preserved.
For each $\beta\in\R^{W}$ with $\sum_{k}\beta_{k}=0$ and $\beta\ne 0$, define
\begin{equation}\label{eq:tropical}
\Tdiv_{\beta}(\boldsymbol{\pi}) := \frac{1}{\beta_{\star}}\log\sup_{x\in\supp\nu}\prod_{k=1}^{W}\pi_{k}^{\beta_{k}}(x),
\qquad \beta_{\star} := \max_{k}\beta_{k}
\end{equation}
This is the multi-way $R_{\infty}$.
The supremum is the $\nu$-essential supremum over the common support of the components, and its value does not depend on the dominating measure $\nu$, because $\sum_{k}\beta_{k}=0$.
Restrict to the cones $\Bminus := \bigcup_{k}\{\beta\in\R^{W}:\sum\beta=0,\ \beta_{k}\ge 0,\ \beta_{\ell}\le 0\ \forall\ell\ne k\}$.
For $W=2$, $\beta=(t,-t)$ gives $\Tdiv_{\beta}(\mu,\nu)=\log\sup_{x}\mu/\nu = R_{\infty}(\mu\|\nu)$.
This is exactly the $t=\infty$ endpoint atom in the bivariate case~\cite{mu2021blackwell}, extended to $W>2$.
Nonnegativity of $\Tdiv_{\beta}$ follows by the same H\"older argument that gives $R_{\infty}\ge 0$.
Normalizing $\beta_{\star}=1$, which changes nothing since $\Tdiv_{c\beta}=\Tdiv_{\beta}$ for $c>0$, one coordinate is $\beta_{k}=1$ and the others are nonpositive with sum $-1$, so the weighted geometric mean $g:=\prod_{\ell\ne k}\pi_{\ell}^{-\beta_{\ell}}$ satisfies $\int g\,\dnu\le 1$ by H\"older's inequality.
Since $\int\pi_{k}\,\dnu=1$, the ratio $\pi_{k}/g=\prod_{j}\pi_{j}^{\beta_{j}}$ cannot be below $1$ everywhere, so its supremum is at least $1$.

\paragraph{Pairwise relative-entropy ``edge'' atoms.}
For each ordered pair $(k,\ell)$ with $k\ne\ell$, the relative entropy $\KLdiv(\pi_{k}\|\pi_{\ell})$ is a $W$-way DPI--additive divergence (it depends only on two of the priors, but is well-defined as a $W$-way functional).
It is the ``edge'' atom indexed by the directed pair $(k,\ell)$.

\paragraph{The extra families as limits of $\Dren_{\alpha}$.}\label{ssec:limits}
Both extra families arise as boundary/scaling limits of $\Dren_{\alpha}$ (or equivalently of $\Cdiv_{\alpha}$ after the rescaling by $\frac{1}{1-\alpha_{\star}}$, which diverges at the vertices and is negative on $\Aminus$), and Proposition~\ref{prop:compactification} makes the extended index space an explicit compactification.
Using the $\Dren$ normalization (which is nonnegative on the entire signed-exponent set), the limits are:
\begin{align}
\KLdiv(\pi_{k}\|\pi_{\ell})
&= \lim_{\epsilon\downarrow 0}\frac{1}{\epsilon}\,\Cdiv_{(1-\epsilon)e_{k}+\epsilon e_{\ell}}(\boldsymbol{\pi})
\quad\text{(boundary derivation at the vertex }e_{k}\text{ in direction }e_{\ell}\text{).}\label{eq:KL-as-limit}\\
\Tdiv_{\beta}(\boldsymbol{\pi})
&= \lim_{t\to\infty} \Dren_{e_{k}+t\beta}(\boldsymbol{\pi})
\;=\; -\frac{1}{\beta_{\star}}\lim_{t\to\infty}\tfrac{1}{t}\Cdiv_{e_{k}+t\beta}(\boldsymbol{\pi})
\quad\text{(scaling limit toward infinity in the }\beta\text{-direction, with }\beta_{k}=\beta_{\star}\text{).}\label{eq:trop-as-limit}
\end{align}
Lemmas~\ref{lem:vertex-derivation} and~\ref{lem:tropical-limit} in Appendix~\ref{app:limits} prove both limits, the first with an explicit error term.
The first gives $\Cdiv_{(1-\epsilon)e_{k}+\epsilon e_{\ell}}(\boldsymbol{\pi})=\epsilon\,\KLdiv(\pi_{k}\|\pi_{\ell})+O(\epsilon^{2})$, so the directional derivative of $\Cdiv_{\alpha}$ at the vertex $e_{k}$ in the direction $e_{\ell}-e_{k}$ is the relative entropy.
The second shows that $\Dren_{e_{k}+t\beta}$ is nondecreasing in $t$ and converges to $\Tdiv_{\beta}$.
The sign flip in~\eqref{eq:trop-as-limit} reflects the structural fact that $\Cdiv_{\alpha}\le 0$ on $\Aminus$, where $H_{\alpha}\ge 1$, so the matrix R\'enyi atom $\Dren_{\alpha}=\Cdiv_{\alpha}/(1-\alpha_{\star})$ becomes positive again because $1-\alpha_{\star}<0$ on $\Aminus$.

\subsection{The full atom space}\label{ssec:full-atoms}

Putting these together, define the \emph{full atom set}
\begin{equation}\label{eq:paramset}
\paramset := \big[(\Aplus\cup\Aminus)\setminus E\big]\;\sqcup\;\Bminus\setminus\{0\}\;\sqcup\;\{(k,\ell):k\ne\ell\}
\end{equation}
together with the atom map
\[
\paramset\ni\xi\;\longmapsto\;\Phi_{\xi}(\boldsymbol{\pi})\in[0,\infty],\qquad
\Phi_{\xi} := \begin{cases}
\Dren_{\alpha} & \text{if }\xi=\alpha\in(\Aplus\cup\Aminus)\setminus E,\\
\Tdiv_{\beta}  & \text{if }\xi=\beta\in\Bminus\setminus\{0\},\\
\KLdiv(\pi_{k}\|\pi_{\ell}) & \text{if }\xi=(k,\ell).
\end{cases}
\]
\begin{lemma}[Each atom is DPI--additive]\label{lem:atoms-are-divergences}
For every $\xi\in\paramset$, the functional $\Phi_{\xi}:\boldsymbol{\pi}\mapsto\Phi_{\xi}(\boldsymbol{\pi})\in[0,\infty]$ is a $W$-way DPI--additive divergence on the class of $W$-tuples of mutually absolutely continuous distributions.
\end{lemma}
See Appendix~\ref{app:atom-proofs} for the proof, which verifies the three axioms in turn for each of the three atom families (signed-exponent atoms, tropical atoms, and pairwise relative-entropy atoms).

\subsection{A concrete instance: $W=3$}\label{ssec:w3-body}

The smallest case beyond binary shows the atom geometry directly.
The affine slice $\mathcal{A}=\{\alpha\in\R^{3}:\alpha_{1}+\alpha_{2}+\alpha_{3}=1\}$ is a 2-dimensional plane, on which the simplex $\Aplus=\Delta_{3}$ is the central triangle with vertices $e_{1},e_{2},e_{3}$, the signed-exponent region $\Aminus$ is three closed cones emerging from those vertices, and the tropical region $\Bminus$ is three cones extending to infinity.
The vertex derivations contribute the six pairwise relative-entropy atoms $\KLdiv(\pi_{k}\|\pi_{\ell})$.
Section~\ref{sec:w3} works the case out: the atoms concretely, the symmetric form under $S_{3}$, the fundamental domain $\Aplus/S_{3}$, and the multi-hypothesis Chernoff connection picked up in Section~\ref{sec:readings}.

\section{The representation theorem}
\label{sec:proof}

\subsection{Statement}

Theorem~\ref{thm:correct} is the forcing statement of the introduction made precise: a functional that is monotone under joint data processing and additive on independent products is a positive integral of rescaled single-logarithm coincidence atoms and their limits.
Every $W$-prior DPI--additive divergence on bounded tuples splits into three measure-components covering the four geometric strata.
The components are a positive integral over the simplex and signed-exponent cones (against the multi-way $\Dren_{\alpha}$ atoms), a positive integral over the tropical-at-infinity boundary (against the $\Tdiv_{\beta}$ atoms), and a finite weighted sum of the $W(W{-}1)$ pairwise relative-entropy vertex edges.
Once the tropical directions are normalized the three components are unique (Proposition~\ref{prop:uniqueness}), and symmetry under joint permutations makes them permutation-invariant with a single relative-entropy coefficient.

\begin{theorem}[$W$-way representation]\label{thm:correct}
Let $D$ be a $W$-way DPI--additive divergence on the class of bounded $W$-tuples.
Then there exist finite Borel measures $m^{\Dren}$ on $(\Aplus\cup\Aminus)\setminus E$,
$m^{\Tdiv}$ on $\Bminus\setminus\{0\}$, and nonnegative coefficients $\{c_{k\ell}:k\ne\ell\}\subset\R_{\ge 0}$ such that, for every bounded tuple $\boldsymbol{\pi}$,
\begin{equation}\label{eq:correct}
D(\boldsymbol{\pi}) \;=\; \int_{(\Aplus\cup\Aminus)\setminus E}\!\Dren_{\alpha}(\boldsymbol{\pi})\,dm^{\Dren}(\alpha)
\;+\; \int_{\Bminus\setminus\{0\}}\!\Tdiv_{\beta}(\boldsymbol{\pi})\,dm^{\Tdiv}(\beta)
\;+\; \sum_{k\ne\ell} c_{k\ell}\,\KLdiv(\pi_{k}\|\pi_{\ell})
\end{equation}
If $D$ is moreover symmetric (Section~\ref{def:axioms}) then the measures $m^{\Dren}$ and $m^{\Tdiv}$ may be chosen $\Sym$-invariant under the diagonal action on $(\Aplus\cup\Aminus)$ and $\Bminus$ respectively, and the matrix $(c_{k\ell})$ is constant off-diagonal:
$c_{k\ell}=c$ for all $k\ne\ell$.

\end{theorem}

\begin{corollary}[Converse: the integral representation is sufficient]\label{cor:converse}
Conversely, for any choice of finite Borel measures $m^{\Dren}$ on $(\Aplus\cup\Aminus)\setminus E$, $m^{\Tdiv}$ on $\Bminus\setminus\{0\}$,
and nonnegative coefficients $\{c_{k\ell}:k\ne\ell\}$ such that the right-hand side of \eqref{eq:correct} is finite on every bounded $W$-tuple, the resulting functional
\[
D(\boldsymbol{\pi}) := \int\Dren_{\alpha}\,dm^{\Dren} + \int\Tdiv_{\beta}\,dm^{\Tdiv} + \sum_{k\ne\ell}c_{k\ell}\KLdiv(\pi_{k}\|\pi_{\ell})
\]
is a $W$-way DPI--additive divergence (Section~\ref{def:axioms}).
\end{corollary}

Corollary~\ref{cor:converse} is proved in Appendix~\ref{app:atom-proofs} and Theorem~\ref{thm:correct} in Appendix~\ref{app:recipe-proofs}.
Together they give the full characterization: a functional on bounded $W$-tuples is a DPI--additive divergence if and only if it admits the integral representation~\eqref{eq:correct}.

The finite-alphabet statement of Theorem~\ref{thm:correct} is due to~\cite[Example~9]{haapasalo2025barycentric}, which proves it from the matrix-Blackwell spectrum theorems~\cite{farooq2024matrix,verhagen2025matrix} by a preordered-semiring argument, and the statement on standard Borel spaces is~\cite[Theorem~12]{haapasalo2026continuous}.
The proof given here works with an explicit compactification of the atom set.
Its one input from the matrix-majorization literature is the large-sample theorem~\cite[Theorem~19]{farooq2024matrix}, through which data processing and additivity make $D$ monotone in the pointwise order of the atom values (Lemma~\ref{lem:profile-monotone}).
The Riesz--Markov theorem on the compactification then supplies the representing measure, and Proposition~\ref{prop:measurable} extends the representation from finite alphabets to arbitrary sample spaces.

\subsection{The compactified atom set and uniqueness}\label{ssec:compactification}

The limit identities~\eqref{eq:KL-as-limit} and~\eqref{eq:trop-as-limit} place the relative-entropy and tropical atoms at the edge of the parameter range of the simplex and cone atoms, and the following coordinates make that edge a topological boundary.
Write $U:=(\Aplus\cup\Aminus)\setminus E$, let $S_{k}$ be the set of probability vectors on the index set $\{1,\dots,W\}\setminus\{k\}$, and for $\alpha\in\mathcal{A}$ let $s_{k}(\alpha):=\sum_{\ell\ne k}|\alpha_{\ell}|$.
For $\alpha\ne e_{k}$, $s_{k}(\alpha)>0$, and $w^{(k)}(\alpha)\in S_{k}$ denotes the vector with entries $|\alpha_{\ell}|/s_{k}(\alpha)$.
In the simplex $\alpha=e_{k}+s_{k}(\alpha)(w^{(k)}(\alpha)-e_{k})$, and in the cone of $\Aminus$ at $e_{k}$, $\alpha=e_{k}+s_{k}(\alpha)(e_{k}-w^{(k)}(\alpha))$, so $s_{k}$ measures the distance from the vertex and $w^{(k)}$ the direction of approach.
With $c:=(1/W,\dots,1/W)$, the map $\phi(\alpha):=(\alpha-c)/(1+|\alpha-c|_{2})$ sends $\mathcal{A}$ homeomorphically onto the open unit ball $B$ of the plane $\{v:\sum_{k}v_{k}=0\}$.
Set
\[
\iota(\alpha):=\big(\phi(\alpha),w^{(1)}(\alpha),\dots,w^{(W)}(\alpha)\big)\in\bar{B}\times S_{1}\times\dots\times S_{W},\qquad\alpha\in U
\]
and let $\overline{U}$ be the closure of $\iota(U)$.
For $w\in S_{k}$ write $\beta^{k,w}:=e_{k}-w$, a direction of $\Bminus$ with $\beta^{k,w}_{\star}=1$; every $\beta\in\Bminus\setminus\{0\}$ is a positive multiple of exactly one such direction.

\begin{proposition}[An explicit compactification]\label{prop:compactification}
The space $\overline{U}$ is compact and metrizable, and $\iota$ maps $U$ homeomorphically onto an open subset of $\overline{U}$.
The complement $\overline{U}\setminus\iota(U)$ is the disjoint union of the images $F_{k}$ of maps $f_{k}$ and $T_{k}$ of maps $\tau_{k}$, $k=1,\dots,W$, each a homeomorphism of $S_{k}$ onto its image, such that for $\alpha_{n}\in U$:
\begin{enumerate}[label=(\roman*)]
\item $\iota(\alpha_{n})\to f_{k}(w)$ if and only if $\alpha_{n}\to e_{k}$ and $w^{(k)}(\alpha_{n})\to w$;
\item $\iota(\alpha_{n})\to\tau_{k}(w)$ if and only if $\alpha_{n}$ lies in the cone of $\Aminus$ at $e_{k}$ for all large $n$, $s_{k}(\alpha_{n})\to\infty$ and $w^{(k)}(\alpha_{n})\to w$.
\end{enumerate}
For every bounded tuple $\boldsymbol{\pi}$ on a finite alphabet, the profile $\Psi_{\boldsymbol{\pi}}$ defined by $\Psi_{\boldsymbol{\pi}}(\iota(\alpha)):=\Dren_{\alpha}(\boldsymbol{\pi})$, $\Psi_{\boldsymbol{\pi}}(f_{k}(w)):=\sum_{\ell\ne k}w_{\ell}\KLdiv(\pi_{k}\|\pi_{\ell})$ and $\Psi_{\boldsymbol{\pi}}(\tau_{k}(w)):=\Tdiv_{\beta^{k,w}}(\boldsymbol{\pi})$ is continuous on $\overline{U}$.
\end{proposition}

The fiber $F_{k}$ over a vertex is the simplex of directions of approach to $e_{k}$, and at its extreme points $f_{k}(e_{\ell})$ the profile is the relative-entropy atom $\KLdiv(\pi_{k}\|\pi_{\ell})$; the fiber $T_{k}$ at infinity parametrizes the normalized tropical directions of the cone at $e_{k}$.
Continuity of $\Psi_{\boldsymbol{\pi}}$ at these points contains the limit identities~\eqref{eq:KL-as-limit} and~\eqref{eq:trop-as-limit}, uniformly in the direction of approach, and the proof in Appendix~\ref{app:recipe-proofs} bounds the error: when $s_{k}(\alpha)\le\tfrac12$, $|\Dren_{\alpha}(\boldsymbol{\pi})-\Psi_{\boldsymbol{\pi}}(f_{k}(w^{(k)}(\alpha)))|\le\Lambda^{2}s_{k}(\alpha)/2$, where $\Lambda$ bounds the log-likelihood ratios of $\boldsymbol{\pi}$, and $0\le\Tdiv_{\beta^{k,w}}(\boldsymbol{\pi})-\Dren_{e_{k}+t\beta^{k,w}}(\boldsymbol{\pi})\le t^{-1}\log(1/\min_{x}\pi_{k}(x))$.

\begin{proposition}[Uniqueness]\label{prop:uniqueness}
If two triples $(m^{\Dren},m^{\Tdiv},(c_{k\ell}))$ and $(m'^{\Dren},m'^{\Tdiv},(c'_{k\ell}))$ of finite measures and nonnegative coefficients, as in Theorem~\ref{thm:correct}, give the same value in~\eqref{eq:correct} for every bounded tuple on a finite alphabet, then $m^{\Dren}=m'^{\Dren}$, $c_{k\ell}=c'_{k\ell}$ for all $k\ne\ell$, and $m^{\Tdiv}$ and $m'^{\Tdiv}$ have the same image under $\beta\mapsto\beta/\beta_{\star}$.
\end{proposition}

The normalization cannot be dropped.
The numerator and the normalizer $\beta_{\star}$ in~\eqref{eq:tropical} are both positively homogeneous of degree one in $\beta$, so $\Tdiv_{c\beta}=\Tdiv_{\beta}$ for $c>0$, and moving mass along a tropical ray changes no divergence.
The measure of Theorem~\ref{thm:correct} is therefore determined by $D$ once its tropical part is placed on normalized directions, and the relative-entropy coefficients are determined outright.

\subsection{Necessity of the boundary strata}\label{sec:exotic}

The simplex-only form (Conjecture~\ref{conj:naive}) restricts attention to $\alpha\in\simplex=\Aplus$, whereas Theorem~\ref{thm:correct} uses $\Aminus$, $\Bminus$ and the relative-entropy edges.
Uniqueness shows that each of these strata is needed.

\begin{corollary}[Every stratum is needed]\label{cor:necessity}
For every $\xi\in\paramset$ the atom $\Phi_{\xi}$ is a $W$-way DPI--additive divergence whose only representation~\eqref{eq:correct} is the unit mass at $\xi$, up to moving tropical mass along rays.
Consequently, for each of the four strata of $\paramset$ some symmetric $W$-way DPI--additive divergence has no representation~\eqref{eq:correct} that gives that stratum zero mass, and the atoms span exactly the extreme rays of the cone of $W$-way DPI--additive divergences on bounded tuples.
\end{corollary}

The proof is in Appendix~\ref{app:recipe-proofs}.
A representation of the form in Conjecture~\ref{conj:naive}, with $m$ a finite measure on $\simplex$, is a representation~\eqref{eq:correct} with $m^{\Dren}=(1-\alpha_{\star})\,m$ on $\Aplus\setminus E$ and no other mass, since $\Cdiv_{\alpha}=(1-\alpha_{\star})\Dren_{\alpha}$ there and $\Cdiv_{e_{k}}=0$.
So Conjecture~\ref{conj:naive} fails already for $W=2$: the symmetrized R\'enyi divergence $R_{1+t}(\mu\|\nu)+R_{1+t}(\nu\|\mu)$ of an order above one, the symmetrized max-divergence $R_{\infty}(\mu\|\nu)+R_{\infty}(\nu\|\mu)$, and the symmetrized relative entropy each need their own stratum.
The relative-entropy atoms are nonetheless pointwise limits of the rescaled simplex atoms $\Dren_{\alpha}$, and the tropical atoms are pointwise limits of the rescaled cone atoms $\Dren_{\alpha}$, by~\eqref{eq:KL-as-limit} and~\eqref{eq:trop-as-limit}.
The necessity therefore concerns finite representing measures: in the compactification of Proposition~\ref{prop:compactification} the limiting unit masses sit on the fibers $F_{k}$ and $T_{k}$, which a finite measure on $U$ does not reach.

\paragraph{Mixed-coincidence reading of the non-simplex regions.}
The mixed-coincidence calculus of~\cite{balsubramani2026information} is defined for general real exponent vectors $\boldsymbol{\alpha}\in\R^{W}$ with the partition function $Z(\boldsymbol{\alpha})=H_{\alpha}$ at the center, so the whole parameter space $\paramset$ of Theorem~\ref{thm:correct},
including its non-simplex regions, has an immediate coincidence-style reading.
The four-perspective identity of that paper specializes to each region as follows.
\begin{itemize}[leftmargin=*]
\item \emph{Simplex $\Aplus$ ($\sum_{k}\alpha_{k}=1$, $\alpha_{k}\in[0,1]$).}
$\log Z(\boldsymbol{\alpha})$ is minus the relative-entropy barycenter value $\min_{p}\sum_{k}\alpha_{k}\,\KLdiv(p\|\pi_{k})$, which is attained at the logarithmic opinion pool~\cite{genest1986combining} $p_{\boldsymbol{\alpha}}^{\star}\propto\prod_{k}\pi_{k}^{\alpha_{k}}$;
$\Cdiv_{\alpha}$ is the negative log of the multi-distribution Bhattacharyya--Chernoff coefficient.
\item \emph{Mixed-sign cones $\Aminus$ ($\sum_{k}\alpha_{k}=1$, one $\alpha_{l}>1$ and the rest $\le 0$).}
The same partition function $Z(\boldsymbol{\alpha})$ is finite on bounded $W$-tuples, and the mixed coincidence identity of~\cite{balsubramani2026information} continues to hold with the $(\sum_{k}\alpha_{k}-1)\,\text{H}(p)$ counting term (which vanishes on the affine slice $\mathcal{A}$ where our atoms live) and yields the same Lagrangian reading $\log Z(\boldsymbol{\alpha})=\max_{p}[\text{H}(p)-\sum_{k}\alpha_{k}\text{H}(p,\pi_{k})]$.
The negative components of $\boldsymbol{\alpha}$ act as reversed-direction (repulsive) Lagrange multipliers (\cite{balsubramani2026information}):
priors $\pi_{k}$ with $\alpha_{k}\le 0$ enter the constrained max-entropy problem as priors the geometric mixture $p_{\boldsymbol{\alpha}}^{\star}$ is being pushed away from, with the active prior $\pi_{l}$ at $\alpha_{l}>1$ providing the attractive constraint.
For $W=2$ the exponent vector $\alpha=(1+t,-t)$ with $t>0$ makes the first distribution attractive and the second repulsive, and $\Dren_{\alpha}(\mu,\nu)=R_{1+t}(\mu\|\nu)$.
\item \emph{Tropical boundary $\Bminus\setminus\{0\}$ ($\sum_{k}\beta_{k}=0$,
$\beta\ne 0$).}
Along an $\Aminus$-ray $\alpha=e_{k}+t\beta$ with $t\to\infty$, the geometric mixture $p_{\alpha}^{\star}\propto\prod_{j}\pi_{j}^{\alpha_{j}}$ concentrates on the support points that maximize $\prod_{j}\pi_{j}^{\beta_{j}}$ (the zero-temperature limit of the Gibbs-conditioning interpretation of~\cite{balsubramani2026information}); $t^{-1}\Cdiv_{e_{k}+t\beta}$ converges to $-\log\sup_{x}\prod_{j}\pi_{j}^{\beta_{j}}=-\beta_{\star}\Tdiv_{\beta}$, the negative of the scaled tropical max-functional $\Tdiv_{\beta}$ of~\eqref{eq:tropical}.
The tropical strata are thus the zero-temperature limits of the mixed-coincidence calculus along the mixed-sign rays.
\item \emph{Vertex relative-entropy atoms.}
At a simplex vertex $\alpha=e_{k}$, $p_{\alpha}^{\star}=\pi_{k}$ itself and $\Cdiv_{e_{k}}=0$.
The mixed-coincidence identity expands $\Cdiv_{e_{k}+\epsilon(e_{\ell}-e_{k})}=-\log Z(e_{k}+\epsilon(e_{\ell}-e_{k}))$ to first order in $\epsilon$, and the coefficient is exactly the vertex relative-entropy atom $\KLdiv(\pi_{k}\|\pi_{\ell})$ recovered as a derivative-style limit in~\eqref{eq:KL-as-limit}; the entropy-counting term $(\sum_{j}\alpha_{j}-1)\text{H}(p)$ of~\cite{balsubramani2026information} vanishes identically at the vertex itself, so the leading-order behavior is the standard relative-entropy expansion.
\end{itemize}
The mixed-coincidence calculus therefore gives a coincidence-counting reading of each stratum: simplex atoms are relative-entropy barycenter values, mixed-sign atoms are attract-repel Lagrangian values, tropical atoms are zero-temperature limits, and relative-entropy vertex atoms are derivative-style expansions of the same partition function.

\section{The cumulant-generating-function reading}\label{sec:cgf}

The forcing of Section~\ref{sec:functionaleq} has an equivalent reading in cumulant-generating-function (CGF) terms: for $j\ne k$ set $X^{j,k}:=\log(\pi_{j}/\pi_{k})$; then
\begin{equation}\label{eq:cgf-renyi}
K^{j,k}_{\boldsymbol{\pi}}(t) := \log\E_{\pi_{k}}\!\big[e^{tX^{j,k}}\big]
= \log\!\int\pi_{j}^{t}\pi_{k}^{1-t}\,\dnu
= (t-1)\,R_{t}(\pi_{j}\|\pi_{k})
\end{equation}
so the CGF of a log-likelihood ratio is exactly $(t-1)$ times the binary R\'enyi divergence.
The CGF family has its own axiomatic characterization~\cite{mu2024monotone}.
Consider a real-valued statistic of bounded random variables that assigns each constant to itself, is monotone under first-order stochastic dominance, and is additive over independent sums.
Every such statistic is a mixture over $a$ of the normalized CGFs $X\mapsto a^{-1}\log\E\!\big[e^{aX}\big]$, with $a$ running over the compactified real line $[-\infty,\infty]$, where $a=0$ gives the expectation and $a=\pm\infty$ the essential maximum and minimum~\cite[Theorem~1]{mu2024monotone}.
This one-parameter family is the set of extreme points.
The convolution structure admits a companion characterization on the map side~\cite{sandomirskiy2025boltzmann}: over finitely supported laws on the integers or the rationals, every support-preserving endomorphism of the convolution semigroup is an exponential tilt $\mu(s)\mapsto\mu(s)e^{-\beta s}/\sum_{t}\mu(t)e^{-\beta t}$.
Mixtures behave differently on the two sides.
A mixture of monotone additive statistics is again monotone and additive, so the statistic side is a convex set.
A mixture $\Psi:=\theta\Psi_{1}+(1-\theta)\Psi_{2}$ of two endomorphisms, with $\theta\in(0,1)$, satisfies
\[
\Psi(\mu*\nu)-\Psi(\mu)*\Psi(\nu)=\theta(1-\theta)\,\big(\Psi_{1}(\mu)-\Psi_{2}(\mu)\big)*\big(\Psi_{1}(\nu)-\Psi_{2}(\nu)\big)
\]
and a convolution of two nonzero finitely supported signed measures on the rationals is nonzero, since after a common rescaling of the support it is a product of two nonzero Laurent polynomials.
Taking $\nu=\mu$, $\Psi$ is an endomorphism only if $\Psi_{1}(\mu)=\Psi_{2}(\mu)$ for every $\mu$, so the map side contains no proper mixtures.
The endpoint atoms appear on the map side as the limits $\beta\to\pm\infty$, which place all mass on the extremes of the support, and are conjectured in~\cite{sandomirskiy2025boltzmann} to exhaust the endomorphisms once support preservation is weakened to absolute continuity.
Both characterizations sit inside the algebraic theory of the convolution semigroup~\cite{ruzsa1988algebraic}, whose universal continuous homomorphism into Hausdorff groups is the cumulant sequence itself~\cite{mattner2004cumulants}.
Theorems~\ref{thm:mpst} and~\ref{thm:correct} have the same Choquet shape one level up, with experiment tuples in place of a single random variable and data-processing monotonicity in place of stochastic dominance.
For more than two states,~\cite[Appendix~K]{mu2021blackwell} gives necessary conditions for large-sample dominance in terms of the multivariate moment generating function of the log-likelihood-ratio vector, which is the Hellinger transform $H_{\alpha}$ on $\Aplus\cup\Aminus$, and conjectures that they are sufficient.
The conjecture is false for three states, and~\cite[Theorem~19]{farooq2024matrix} proves the sufficiency once strict tropical inequalities are added (Appendix~\ref{app:section-k}).
The CGF identity~\eqref{eq:cgf-renyi} restricts to two coordinates a multivariate statement that reappears in Section~\ref{sec:readings}: the Hellinger transform $H_{\alpha}$ is a multivariate Laplace transform of the joint pushforward of $\nu$ under the log-loss map $x\mapsto(-\log\pi_{1}(x),\dots,-\log\pi_{W}(x))$.
Appendix~\ref{app:laplace} develops this Laplace-transform normal form (binary Theorem~\ref{thm:binary-laplace} and multi-way Theorem~\ref{thm:multiway-laplace-normal-form}) and proves an exact free-energy identity for the posterior over distributions that this transform induces (Theorem~\ref{thm:laplace-mixed-concentration}).

\section{The role of permutation invariance}\label{sec:perm}

Permutation invariance is a substantive axiom in the multi-prior setting: it is \emph{not} automatic from data-processing monotonicity plus additivity.
The asymmetric examples furnished by~\cite{mu2021blackwell}
(masses on $m_{0}$ but not $m_{1}$) survive in the multi-way setting: any one-sided Kullback projection of the form $D(\boldsymbol{\pi})=\KLdiv(\pi_{1}\|\pi_{2})$ is a $W$-way DPI--additive divergence, but it is manifestly not symmetric.

Imposing $\Sym$-invariance has the effect of averaging over the orbit.
The same passage-to-the-quotient is the structural content of the Deep Sets representation theorem of~\cite{zaheer2017deepsets},
which characterizes every continuous permutation-invariant function $f$ on a finite multiset $\{x_{1},\dots,x_{n}\}$ as a composition $f(\{x_{i}\})=\rho(\sum_{i}\phi(x_{i}))$ for continuous $\phi,\rho$; permutation-invariance implies a sum-pooling factorization through the quotient by $\Sym$.
The Riesz--Markov measure $m^{\Dren}$ in our setting plays the role of $\phi$ (a per-orbit-class weight) integrated against the equivariant atom map $\Dren_{[\alpha]}$ in place of a finite sum over a multiset.
A $\Sym$-invariant divergence is a positive integral against an $\Sym$-equivariant atom family, with the measure necessarily descending to the orbit space. Concretely,
the $\Sym$-orbit decomposition of $(\Aplus\cup\Aminus)\setminus E$ has finitely many strata (one per multiplicity pattern of the components of $\alpha$), and a $\Sym$-invariant Borel measure on the parameter set is a finite Borel measure on the quotient $((\Aplus\cup\Aminus)\setminus E)/\Sym$.
The atom $\Cdiv_{\alpha}$ itself is automatically $\Sym$-equivariant under the joint action $(\sigma\cdot\alpha,\sigma\cdot\boldsymbol{\pi})$
(this follows from the basic properties of the coincidence divergence), so passing to the quotient is a purely group-theoretic reduction; nothing about the analytic structure of the integral representation changes.

For the tropical and relative-entropy atoms the same reduction applies.
The $W$ cones of $\Bminus\setminus\{0\}$ form a single $\Sym$-orbit; symmetry collapses $m^{\Tdiv}$ to a measure on one fundamental cone, on which the stabilizer $\mathfrak{S}_{W-1}$ of that cone still acts.
The $W(W-1)$ ordered pairs of relative-entropy atoms collapse to a single coefficient $c\,\sum_{k\ne\ell}\KLdiv(\pi_{k}\|\pi_{\ell})$.

\paragraph{The symmetric form.}
Combining (\ref{eq:correct}) with $\Sym$-symmetry yields the symmetric form
\begin{equation}\label{eq:correct-sym}
D_{\mathrm{sym}}(\boldsymbol{\pi}) = \int_{[\Aplus\cup\Aminus]/\Sym}\!\Dren_{[\alpha]}\,d\bar m^{\Dren}([\alpha]) + \int_{\Bminus/\Sym}\!\Tdiv_{[\beta]}\,d\bar m^{\Tdiv}([\beta]) + c\!\sum_{k\ne\ell}\!\KLdiv(\pi_{k}\|\pi_{\ell})
\end{equation}
where $\Dren_{[\alpha]}$ and $\Tdiv_{[\beta]}$ denote the orbit-symmetrized atoms.

\section{A worked example: $W=3$}\label{sec:w3}

$W=3$ is the smallest non-binary case and exhibits the new structural features in concrete form.
The parameter set $\mathcal{A}$ is 2-dimensional, so the simplex, signed cones, and tropical cones appear as distinct geometric regions.
The permutation orbit structure of $S_{3}$ on $\Delta_{3}$ has the central Bhattacharyya--Matusita atom plus edges.
The connection to multi-hypothesis testing of~\cite{salikhov1973asymptotic,leang1997asymptotics} also becomes explicit.

\subsection{The atom set for $W=3$}

The affine slice $\mathcal{A}=\{\alpha\in\R^{3}:\alpha_{1}+\alpha_{2}+\alpha_{3}=1\}$ is a 2-dimensional affine plane.
Its distinguished sub-regions are:
\begin{itemize}[leftmargin=*]
\item $\Aplus$: the closed standard simplex $\Delta_{3}$, a triangle with vertices $e_{1},e_{2},e_{3}$.
\item $\Aminus$: three closed cones $\Aminus^{(k)}$, $k=1,2,3$, each emerging from vertex $e_{k}$ in the directions where $\alpha_{k}\ge 1$ and $\alpha_{\ell}\le 0$ for $\ell\ne k$.
For example $\Aminus^{(1)} = \{\alpha:\alpha_{1}\ge 1,\alpha_{2}\le 0,\alpha_{3}\le 0,\sum=1\}$.
\item $\Bminus$: three cones $\Bminus^{(k)}$, $k=1,2,3$, of \emph{tropical} parameters $\beta\in\R^{3}$ with $\sum\beta=0$, $\beta_{k}\ge 0$, $\beta_{\ell}\le 0$ for $\ell\ne k$.
\item $E=\{e_{1},e_{2},e_{3}\}$: the three vertices, where $\Aplus$ and $\Aminus$ meet, excluded from the index set $(\Aplus\cup\Aminus)\setminus E$.
\end{itemize}
An affine plane in $\R^{3}$ has the simplex as a triangle in the center,
three R\'enyi-cone wedges $\Aminus^{(k)}$ emerging from the three vertices, and three tropical cones $\Bminus^{(k)}$ extending to infinity along the lines through each vertex in the direction $2e_{k}-\sum_{\ell\ne k}e_{\ell}$.
The relative-entropy ``edges'' are six ordered pairs.

\subsection{Concrete atoms}

For three distributions $\pi_{1},\pi_{2},\pi_{3}$:
\begin{itemize}[leftmargin=*]
\item Simplex atoms: $\Cdiv_{(\alpha_{1},\alpha_{2},\alpha_{3})}=-\log\int\pi_{1}^{\alpha_{1}}\pi_{2}^{\alpha_{2}}\pi_{3}^{\alpha_{3}}\dnu$ for $(\alpha_{1},\alpha_{2},\alpha_{3})\in\Delta_{3}$.
Special case $\alpha=(1/3,1/3,1/3)$: the negative log of the \emph{Bhattacharyya--Matusita 3-way affinity}~\cite{matusita1967affinity,toussaint1974}.
\item R\'enyi-above-1 atoms: $\Cdiv_{(\alpha_{1},\alpha_{2},\alpha_{3})}$ for e.g.\ $\alpha=(2,-1/2,-1/2)$, ``twice $\pi_{1}$ minus the geometric mean of $\pi_{2},\pi_{3}$''.
\item Tropical atoms: $\Tdiv_{(1,-1/2,-1/2)}=\log\sup_{x}\pi_{1}(x)/\sqrt{\pi_{2}(x)\pi_{3}(x)}$,
the ``maximum log-likelihood ratio of $\pi_{1}$ over the geometric mean of $\pi_{2},\pi_{3}$.''
\item Relative-entropy edges: $\KLdiv(\pi_{1}\|\pi_{2}),\dots,\KLdiv(\pi_{3}\|\pi_{2})$ etc., six in total.
\end{itemize}

\subsection{The symmetric form for $W=3$}

If $D$ is symmetric under $\Sym=S_{3}$, then by Theorem~\ref{thm:correct} together with the symmetry reduction:
\begin{align*}
D(\pi_{1},\pi_{2},\pi_{3})
&= \int_{[\Aplus\cup\Aminus]/S_{3}}\!\Dren_{[\alpha]}\,d\bar m^{\Dren}([\alpha]) \\
&\qquad{} + \int_{\Bminus/S_{3}}\!\Tdiv_{[\beta]}\,d\bar m^{\Tdiv}([\beta]) \\
&\qquad{} + c\big[\KLdiv(\pi_{1}\|\pi_{2})+\KLdiv(\pi_{2}\|\pi_{1})+\KLdiv(\pi_{1}\|\pi_{3})+\KLdiv(\pi_{3}\|\pi_{1})+\KLdiv(\pi_{2}\|\pi_{3})+\KLdiv(\pi_{3}\|\pi_{2})\big]
\end{align*}
The fundamental domain $\Aplus/S_{3}$ is the closed sub-triangle $\{\alpha:\alpha_{1}\ge\alpha_{2}\ge\alpha_{3}\ge 0,\sum=1\}$ (a sextant of the simplex), and $\Bminus/S_{3}$ is half of one of the three cones (say $\{\beta\in\Bminus^{(1)}:\beta_{2}\le\beta_{3}\}$): the other two cones are $S_{3}$-conjugate to $\Bminus^{(1)}$, and the transposition of the last two coordinates, which fixes $\Bminus^{(1)}$, folds it in half.

\subsection{Two multi-hypothesis Chernoff quantities}

Two distinct support-style functionals, easily conflated, appear in the multi-hypothesis literature.

\paragraph{The simplex-support Chernoff.}
The support function of the simplex-indexed atom family is the $\Cdiv$-supremum over the simplex,
\[
\mathsf{C}^{\sup}_{(W)}(\boldsymbol{\pi})\;:=\;\max_{\alpha\in\simplex}\Cdiv_{\alpha}(\boldsymbol{\pi})
\]
the pointwise upper envelope of the atom family $\{\Cdiv_{\alpha}\}_{\alpha\in\simplex}$.
It is an operational quantity built from the atoms of Theorem~\ref{thm:correct} but does not belong to their cone: every divergence in the cone is additive on products, while $\max_{\alpha}(\Cdiv_{\alpha}(\boldsymbol{\pi})+\Cdiv_{\alpha}(\boldsymbol{\pi}'))\le\max_{\alpha}\Cdiv_{\alpha}(\boldsymbol{\pi})+\max_{\alpha}\Cdiv_{\alpha}(\boldsymbol{\pi}')$ makes $\mathsf{C}^{\sup}_{(W)}$ sub-additive, strictly so already for two distributions (Appendix~\ref{app:radius}).
Each fixed $\Cdiv_{\alpha}$ with $\alpha\in\Aplus\setminus E$, by contrast, is a $W$-way DPI--additive divergence, represented by the point mass $(1-\alpha_{\star})\delta_{\alpha}$ in~\eqref{eq:correct}.

\paragraph{The Salikhov--Leang--Johnson Bayes-error rate.}
The Salikhov--Leang--Johnson rate for the Bayes-error exponent in $W$-hypothesis testing~\cite{salikhov1973asymptotic,leang1997asymptotics} is the \emph{minimum pairwise} Chernoff,
\[
\mathsf{C}^{\mathrm{SLJ}}_{(W)}(\boldsymbol{\pi})\;:=\;\min_{j\ne k}\mathsf{C}_{\mathrm{Ch}}^{(2)}(\pi_{j},\pi_{k})
\;=\;\min_{j\ne k}\max_{t\in[0,1]}\Cdiv_{(t,1-t)}(\pi_{j},\pi_{k})
\]
the rate at which the average Bayes error decays in the $W$-state hypothesis-testing setup; the quantum analogue is~\cite{nussbaum2009chernoff}.

\paragraph{Relation between the two.}
The two quantities are distinct in general.
Edge-restricted maxima coincide with pairwise binary Chernoffs,
\[
\max_{\alpha\in\Delta_{3},\alpha_{3}=0}\Cdiv_{\alpha}(\pi_{1},\pi_{2},\pi_{3})
\;=\;\mathsf{C}_{\mathrm{Ch}}^{(2)}(\pi_{1},\pi_{2}),
\]
so, since the supremum over the full simplex dominates the supremum over any face, the simplex-support is bounded below by the \emph{maximum} pairwise Chernoff,
$\mathsf{C}^{\sup}_{(W)}\ge\max_{j\ne k}\mathsf{C}_{\mathrm{Ch}}^{(2)}(\pi_{j},\pi_{k})$,
which generically dominates the SLJ rate $\min_{j\ne k}\mathsf{C}_{\mathrm{Ch}}^{(2)}$.
The maximizer need not lie on a low-dimensional face, however.
For the symmetric triple $\pi_{1}=(0.9,0.05,0.05)$, $\pi_{2}=(0.05,0.9,0.05)$, $\pi_{3}=(0.05,0.05,0.9)$, the function $\alpha\mapsto\Cdiv_{\alpha}(\boldsymbol{\pi})$ is concave on the simplex, since $\log H_{\alpha}$ is convex, and invariant under permuting the coordinates of $\alpha$, so averaging a maximizer over its permutations shows that the center $(\tfrac13,\tfrac13,\tfrac13)$ is a maximizer.
There $\mathsf{C}^{\sup}_{(W)}=-\log\big(3\,(0.9\cdot 0.05^{2})^{1/3}\big)\approx 0.934$, while on each edge the same argument places the maximum at the midpoint, with value $-\log\big(2\sqrt{0.9\cdot 0.05}+0.05\big)\approx 0.746$, so the lower bound above can be loose.
Since $\mathsf{C}^{\sup}_{(W)}\ge\max_{j\ne k}\mathsf{C}_{\mathrm{Ch}}^{(2)}(\pi_{j},\pi_{k})$ while $\mathsf{C}^{\mathrm{SLJ}}_{(W)}=\min_{j\ne k}\mathsf{C}_{\mathrm{Ch}}^{(2)}(\pi_{j},\pi_{k})$,
the strict inequality $\mathsf{C}^{\sup}_{(W)}>\mathsf{C}^{\mathrm{SLJ}}_{(W)}$ holds whenever the pairwise Chernoffs are not all equal.

\paragraph{Position of each in the structural picture.}
Neither quantity belongs to the additive cone of Theorem~\ref{thm:correct}, because neither is additive on products.
$\mathsf{C}^{\sup}_{(W)}$ is sub-additive, as shown above.
$\mathsf{C}^{\mathrm{SLJ}}_{(W)}$ can be strictly super-additive.
If $\pi_{1}=\pi_{2}\ne\pi_{3}$ and $\pi'_{1}=\pi'_{3}\ne\pi'_{2}$, then $\mathsf{C}^{\mathrm{SLJ}}_{(3)}(\boldsymbol{\pi})=\mathsf{C}^{\mathrm{SLJ}}_{(3)}(\boldsymbol{\pi}')=0$.
Every pair of components of $\boldsymbol{\pi}\otimes\boldsymbol{\pi}'$ differs in one of the two factors, and by data processing under marginalization its Chernoff information is at least that of the corresponding pair in that factor, so it is positive.
Both are \emph{envelopes} of the atoms, selected pointwise by the data: $\mathsf{C}^{\sup}_{(W)}$ the upper envelope $\max_{\alpha}\Cdiv_{\alpha}$, and $\mathsf{C}^{\mathrm{SLJ}}_{(W)}$ a lower envelope, a minimum over pairs $(j,k)$ of the edge-restricted binary Chernoff informations.
Both are functions of the same simplex-restricted family of Hellinger transforms, and they differ in which extremizer of that family the operational context selects.

\section{Probabilistic readings of the coincidence divergence}\label{sec:readings}

The forcing argument fixes the form of the atoms; the identities below say what the atoms on the simplex measure.
Each is a property of the family $\{\Cdiv_{\alpha}\}_{\alpha\in\Aplus}$ of simplex atoms, and together they take the coincidence divergence from a geometric quantity to an analytic one and then to the free energy of a posterior over distributions.

\paragraph{Information radius / minimax.}
The mixed coincidence identity of~\cite{balsubramani2026information} gives, for $\alpha\in\simplex$ and priors with full support on a finite alphabet,
$\Cdiv_{\alpha}(\boldsymbol{\pi}) = \min_{r\in\Delta(\mathcal{X})}\sum_{k=1}^{W}\alpha_{k}\KLdiv(r\|\pi_{k})$,
with optimum $r=p^{\star}_{\alpha}\propto\prod_{k}\pi_{k}^{\alpha_{k}}$, the geometric mixture.
This identity is a one-line Gibbs-variational computation, given in Appendix~\ref{app:radius} together with the application of Sion's minimax theorem that yields
\begin{equation}\label{eq:info-radius}
\max_{\alpha\in\simplex}\Cdiv_{\alpha}(\boldsymbol{\pi}) = \min_{r\in\Delta(\mathcal{X})}\max_{k}\KLdiv(r\|\pi_{k})
\end{equation}
the worst-case Kullback projection radius. It is the first-argument counterpart of the information radius of~\cite{sibson1969information}, which places the center in the second argument.
This information radius is an operational summary built from the simplex family $\{\Cdiv_{\alpha}\}$.
Unlike each fixed-$\alpha$ atom, it is not an element of the additive cone of Theorem~\ref{thm:correct}: its data-dependent maximizer leaves the support functional $\max_{\alpha}\Cdiv_{\alpha}$ only sub-additive under products, strictly so in the $W=2$ example of Appendix~\ref{app:radius}.
It is one of two pointwise envelopes of the simplex atoms, distinct from the Salikhov--Leang--Johnson minimum-pairwise rate $\mathsf{C}^{\mathrm{SLJ}}_{(W)}$, and Section~\ref{sec:w3} compares the two.

\paragraph{Coincidence of empirical distributions.}
On the simplex the atom is the exponential rate of a coincidence event among independent samples, one sample from each distribution, with sample sizes in the proportions $\alpha$.

\begin{proposition}[Coincidence of empirical distributions]\label{prop:coincidence-rate}
Let $\mathcal{X}$ be finite, let $\pi_{1},\dots,\pi_{W}$ have full support on $\mathcal{X}$, and let $\alpha\in\Aplus$ have rational coordinates with common denominator $q$.
For $N\in q\mathbb{N}$ let $E_{N}$ be the event that independent samples, the $k$-th consisting of $\alpha_{k}N$ independent draws from $\pi_{k}$, all have the same empirical distribution, samples of size zero being ignored.
Then $\frac{1}{N}\log\Pr(E_{N})\to-\Cdiv_{\alpha}(\boldsymbol{\pi})$ as $N\to\infty$ through $q\mathbb{N}$.
\end{proposition}

The proof, in Appendix~\ref{app:radius}, rests on the method of types; the common empirical distribution that dominates the event is the geometric mixture $p^{\star}_{\alpha}$.
On a finite alphabet with counting measure, a single draw from each distribution lands on a common point with probability $\sum_{x}\prod_{k}\pi_{k}(x)$, the Hellinger transform at the exponent vector $(1,\dots,1)$, which lies off the affine slice.

\paragraph{Laplace-transform view.}
The partition function $Z(\alpha):=H_{\alpha}(\boldsymbol{\pi})$ is a multivariate Laplace transform of the joint law of the per-prior log-losses $\ell_{k}(x):=-\log\pi_{k}(x)$ under $\nu$.
With $\tilde q := \ell_{*}\nu$ the pushforward of $\nu$ on $\R^{W}$,
\begin{equation}\label{eq:laplace-view}
Z(\alpha)=\E_{X\sim\nu}\!\Big[\textstyle\prod_{k}\pi_{k}(X)^{\alpha_{k}}\Big]
= \int_{\R^{W}} e^{-\langle\alpha,t\rangle}\,d\tilde q(t)
\end{equation}
so $\Phi(\alpha):=\log Z(\alpha)$ is the log-Laplace transform of $\tilde q$, which is the cumulant generating function of the log-likelihood vector $(\log\pi_{k}(X))_{k}$ when $\nu$ is a probability measure, and the geometric mixture $p^{\star}_{\alpha}$ is the Gibbs measure, or exponential tilt, of $\tilde q$.
This is the Donsker--Varadhan and maximum-entropy reading developed at full generality, for any real $\boldsymbol{\alpha}$ and possibly unnormalized factors, in~\cite{balsubramani2026information}, where $\log Z(\boldsymbol{\alpha})$ is identified with the value of the unconstrained Lagrangian $\max_{p}[\text{H}(p)-\sum_{k}\alpha_{k}\text{H}(p,\pi_{k})]$.
This makes the classical analytic toolkit available for $-\log H_{\alpha}$:
complete monotonicity (when $\ell_{k}\ge 0$), real-analyticity of $\Phi$ on its domain of finiteness, cumulant expansions and Hessian-as-covariance identities, Chernoff and Markov tail bounds along rays, large-deviation Legendre duality, and Laplace inversion.
Inversion needs care about where $Z$ is known.
On the hyperplane $\sum_{k}\alpha_{k}=1$, where every atom lives, the values of $Z$ on any open piece, such as the interior of the simplex, determine every atom of a bounded tuple, the signed-exponent and tropical ones included.
They do not determine $\tilde q$; the values on a full-dimensional open set in $\R^{W}$ where $Z$ is finite do (Proposition~\ref{prop:hyperplane} in Appendix~\ref{app:laplace}).
A binary specialization of this view (\cite{galke2024sufficiencyrenyi}) already underlies the bivariate argument of~\cite{mu2021blackwell}.

\paragraph{Free energy of a posterior over distributions.}
The exponential tilt has a Bayesian reading in which the coincidence divergence enters an exact free-energy identity.
Place a reference probability measure $\sigma$ with full support on the simplex $\Delta(\mathcal{X})$, and let $\tilde{\mu}^{(t)}$ be the posterior over distributions whose density against $\sigma$ is proportional to $\exp(-t[\sum_{k}\alpha_{k}H(p,\pi_{k})-H(p)])$, with $t$ the observation resolution and $\mathcal{Z}_{t}$ its normalizer.
Then, at every resolution $t>0$ (Theorem~\ref{thm:laplace-mixed-concentration}),
\begin{equation}\label{eq:free-energy-body}
-\frac{1}{t}\log\mathcal{Z}_{t} \;=\; \Cdiv_{\alpha}(\boldsymbol{\pi}) \;+\; \E_{\tilde{\mu}^{(t)}}\!\big[\KLdiv(p\,\|\,p_{\alpha}^{\star})\big] \;+\; \frac{1}{t}\,\KLdiv\big(\tilde{\mu}^{(t)}\,\big\|\,\sigma\big)
\end{equation}
The free energy of the posterior is the coincidence divergence plus two nonnegative information costs: the expected relative entropy from a posterior draw to the typical distribution $p_{\alpha}^{\star}$, and the information the posterior gains over its reference per unit of resolution.
Both costs vanish as resolution grows, so the posterior concentrates on $p_{\alpha}^{\star}$ and its free energy converges to $\Cdiv_{\alpha}(\boldsymbol{\pi})$ itself (Corollary~\ref{cor:laplace-concentration}).
This is a level-2 statement, about the distributions drawn from the posterior, and it holds in the strong form of a large-deviation principle (Corollary~\ref{cor:level2-ldp}).
Up to the usual interior and closure bounds, the posterior mass of a set of distributions decays at the exponential rate given by the smallest relative entropy from the set to $p_{\alpha}^{\star}$.
When $\sigma$ has a density that is continuous and positive at $p_{\alpha}^{\star}$, the first cost is $(|\mathcal{X}|-1)/(2t)$ to first order, $1/(2t)$ per dimension of the simplex, and the second is of order $\log t/t$ (Proposition~\ref{prop:laplace-rates}).
The coincidence divergence is therefore the limiting free energy of Bayesian inference over distributions.

\section{Relation to prior and concurrent work}\label{sec:catmarkov}

The two-distribution case ($W=2$) of the present axiomatic question was settled in~\cite{mu2021blackwell}, recalled here as Theorem~\ref{thm:mpst}: the R\'enyi family generates the bivariate divergences that are monotone under data processing and additive on independent products.
For $W=2$ and $\alpha=(t,1-t)$, the multi-way coincidence divergence $\Cdiv_{\alpha}$ recovers $(1-t)R_{t}(\pi_{1}\|\pi_{2})$, and on the $W$-prior simplex $\alpha\in\Delta_{W}$ it is the negative log of the $W$-way Bhattacharyya--Chernoff coefficient.
For more than two states,~\cite[Appendix~K]{mu2021blackwell} gave necessary conditions for large-sample dominance and conjectured that they are sufficient.
That conjecture is false for three states (Proposition~\ref{prop:mpst-counterexample}).
Large-sample dominance of $W$-tuples follows from strict inequalities between all of their atoms by~\cite[Theorem~19]{farooq2024matrix}.
The monotone homomorphisms of the matrix-majorization semiring are classified in~\cite[Propositions~13--14]{farooq2024matrix}.
The large-sample theorem is the input on which the forward direction here rests.
Composing it with the Riesz--Markov representation argument, which writes a positive linear functional as a positive integral, yields the full $W$-way characterization.

\subsection{The preordered-semiring reading}\label{ssec:semiring}

The abstract framing of Theorem~\ref{thm:correct} comes from \emph{Markov categories}~\cite{fritz2020synthetic,cho2019disintegration} and \emph{preordered semirings}~\cite{fritz2022vergleichsstellensaetze2}.
This framing presents the $W$-way representation as canonical from outside the analytic argument and connects directly to the abstract characterization of~\cite{haapasalo2025barycentric}; the main results do not depend on it.

Tuples of $W$ finite measures with a common support on a finite alphabet, modulo Blackwell equivalence, form a preordered commutative semiring $\mathcal{S}_{W}$:
addition is direct sum, multiplication is tensor product, and the preorder is matrix majorization.
By the abstract Vergleichsstellens\"atze~\cite[Theorems~7.1 and~8.6]{fritz2022vergleichsstellensaetze2}, for preordered semirings meeting the structural hypotheses stated there, which include ``power universal'' elements, the catalytic and large-sample preorders are characterized, up to the strictness of the inequalities, by the monotone homomorphisms to the real and tropical semifields and their opposites.
Monotone derivations at the degenerate homomorphisms complete the list.
The matrix-Blackwell spectrum theorems of~\cite{farooq2024matrix} compute these for $\mathcal{S}_{W}$ and find exactly the three atom families: Hellinger transforms $f_{\alpha}$, tropical maps $f_{\beta}^{T}$, and vertex derivations $\Delta_{\gamma}^{(k)}$ with $\gamma\in\R_{+}^{W}$.

In this framework Theorem~\ref{thm:correct} reads: the cone of additive monotone real-valued functionals on $\mathcal{S}_{W}$ is exactly the dual cone of the spectrum, expressed as positive integrals against extreme rays.
$\Cdiv_{\alpha}$ is the negative logarithm of the multiplicative monotone homomorphism $f_{\alpha}$ for $\alpha\in\Aplus$, and the other atoms ($\Aminus$
R\'enyi-above-one, tropical, relative-entropy) are the remaining extreme rays of the same spectrum.
The multi-way coincidence calculus is the negative logarithm of the spectrum of the $W$-prior matrix-majorization preordered semiring; the $W$-way characterization theorem is the duality statement spelling this out concretely.

\subsection{The general characterization}\label{ssec:general-characterization}
The integral representation in Theorem~\ref{thm:correct} appears as the classical-multivariate specialization of Theorem~7, Example~9 and Figure~1 of~\cite{haapasalo2025barycentric}, which establishes the analogous barycentric decomposition for general (classical and quantum) $d$-variate \emph{extensive monotone divergences} via the preordered-semiring Vergleichsstellensatz machinery of Section~\ref{ssec:semiring}.
There the test spectrum $\widehat{\mathfrak{D}}$ of monotone homomorphisms decomposes into four pieces ($\mathfrak{D}_{\mathbb{R}_{+}}$,
$\mathfrak{D}_{\mathbb{R}_{+}^{\mathrm{op}}}$,
$\mathfrak{D}_{\mathbb{T}\mathbb{R}_{+}}$,
$\mathfrak{D}_{\mathbb{T}\mathbb{R}_{+}^{\mathrm{op}}}$), together with $d$ \emph{derivation} pieces $\mathfrak{D}_{1},\dots,\mathfrak{D}_{d}$ (monotone derivations satisfying the Leibniz rule), and an inner/outer regular Borel measure $\mu$ on $\widehat{\mathfrak{D}}$ representing every monotone divergence as a barycenter $D(\vec{\rho})=\int_{\widehat{\mathfrak{D}}}\Delta(\vec{\rho})\,d\mu(\Delta)$.

The dictionary with the present paper is exact in the classical multivariate case.
The order-preserving and order-reversing real homomorphisms $\mathfrak{D}_{\mathbb{R}_{+}}\cup\mathfrak{D}_{\mathbb{R}_{+}^{\mathrm{op}}}$ are the simplex and signed-exponent atoms $\Dren_{\alpha}$ for $\alpha\in(\Aplus\cup\Aminus)\setminus E$.
The tropical-real homomorphisms $\mathfrak{D}_{\mathbb{T}\mathbb{R}_{+}}\cup\mathfrak{D}_{\mathbb{T}\mathbb{R}_{+}^{\mathrm{op}}}$ are the boundary-at-infinity tropical atoms $\Tdiv_{\beta}$.
The $d$ derivation pieces $\mathfrak{D}_{k}$ are the vertex relative-entropy atoms $\KLdiv(\pi_{k}\|\pi_{\ell})$.
The test spectrum is compact Hausdorff in the pointwise comparison topology (\cite[Proposition~8.5]{fritz2022vergleichsstellensaetze2}, relayed by~\cite{haapasalo2025barycentric}), and Proposition~\ref{prop:compactification} gives explicit coordinates for a compactification of the classical atom set with the same four strata.
The inclusion $\mathfrak{D}_{\mathrm{nd}}\subseteq\mathrm{ext}\,\mathfrak{D}$ (\cite{haapasalo2025barycentric} eq.~(4)) is the counterpart of the extreme-ray statement of Corollary~\ref{cor:necessity}.
Figure~1 of~\cite{haapasalo2025barycentric} shows the $d{=}3$ test spectrum explicitly, with the same four strata in the same positions, and the closing paragraph of Example~9 recovers the binary MPST theorem~\cite{mu2021blackwell} as the $d{=}2$ specialization.

\subsection{The contribution of this paper}\label{ssec:this-paper}
The two presentations are complementary in scope and audience.
The abstract preordered-semiring Vergleichsstellensatz route of~\cite{haapasalo2025barycentric} treats the classical and the quantum multivariate cases together, at a higher level of abstraction.
It is the proof of record for the existence of the integral representation on finite alphabets, via its Example~9.
The concurrent companion~\cite{haapasalo2026continuous} extends the classical characterization to continuous (standard Borel) observation spaces and recasts it as a criterion for large-sample and catalytic majorization of statistical experiments, with an explicit optimal-conversion-rate reading (the best rate of transforming one experiment into another is an infimum of ratios of these divergences).
Its Theorem~12 coincides with Theorem~\ref{thm:correct} on standard Borel spaces, and Proposition~\ref{prop:measurable} gives an independent proof of that extension.
That paper records the relationship from the other side as well: its Remark~13 notes that the same representation was proved independently here, with a methodology similar to that of~\cite{haapasalo2025barycentric} but different preceding proof techniques.
Priority for the classical multivariate representation nonetheless rests with~\cite{haapasalo2025barycentric}, which precedes both.
The present paper's concern is complementary to both.
Where the abstract route establishes that the representation exists and classifies what may appear in it, the questions here concern what the atoms are as probabilistic objects, and how the calculus behaves when a learning problem supplies side information.
On the first: on a finite alphabet, each atom indexed by a rational point of the simplex is the exponential rate at which independent samples, one per distribution with sizes in fixed proportions, all show the same empirical distribution (Proposition~\ref{prop:coincidence-rate}).
The family is the multivariate Laplace transform of the log-loss vector, whose exponential tilt is the geometric mixture the atom optimizes over and on which the Laplace-mixed posterior concentrates (Section~\ref{sec:readings}).
The maximum of the simplex atoms is a worst-case Kullback projection radius, and the rescaled atoms are log certainty equivalents of a risk-averse gambler~\cite{ducuara2026}.
None of this is settled by the existence theorem, which fixes the form of the atoms without assigning them probabilistic content.
On the second: Section~\ref{sec:conditional} extends the calculus to conditioned experiments, where every stratum reappears with an aggregator over the side information whose admissible range depends on the stratum.
The proof of Section~\ref{sec:proof} is given in full for that reason: its uniqueness statement, applied with a trivial conditioning variable, shows that the conditional divergences need every stratum as well (Corollary~\ref{cor:cond-forced}).

\subsection{The noncommutative analogue}\label{ssec:noncommutative}
The barycentric decomposition of~\cite[Theorem~7]{haapasalo2025barycentric} applies to tuples of finite-dimensional density operators with a common support (its Section~4).
Every divergence of such tuples that is monotone under quantum channels and additive on tensor products is a barycenter over a quantum test spectrum with the same stratification, and each point of that spectrum restricts on commuting tuples to a single classical atom.
The functionals that make up the quantum spectrum are not known.
Over a single classical atom it already contains several distinct divergences~\cite{haapasalo2025barycentric}: for two states, the $\alpha$-$z$ R\'enyi divergences throughout their data-processing range, among them the Petz and sandwiched ones, and over the relative entropy both the Umegaki and the Belavkin--Staszewski relative entropies.
A gambling-resource-theoretic treatment is given in~\cite{arcos2025resource}.

\subsection{Adjacent work}\label{ssec:adjacent}
The matrix-Blackwell spectrum is generalized to the varying-support multivariate setting in~\cite{verhagen2025matrix}.~\cite{mosonyi2024geometric} constructs new monotone quantum multivariate divergences via a variational formula complementary to the characterization route.~\cite{bunth2021equivariant} handles the equivariant submajorization setting relevant to resource-theoretic thermodynamics.
A complementary restriction is run by~\cite{cruz2026tensorizable}: keeping the Csisz\'ar $f$-divergence form and relaxing additivity to an arbitrary compositional formula under independent products forces the formula to be multi-affine, and the tensorizable $f$-divergences reduce to nonnegative combinations of the relative entropy and its reversal, or to power divergences.
The operational side of the two-distribution case, via horse-betting and certainty-equivalent reasoning, goes back to~\cite{bleuler2020}, and~\cite{ducuara2026} extends it to the multi-prior $W$-tuple setting with multi-lottery betting; a characterization of conditional entropies is~\cite[Theorem~7.1]{rubboli2026conditional}. Both lines are picked up in Section~\ref{sec:conditional}.

\subsection{Relation to the mixed partition function}\label{ssec:mixed-partition}
The mixed partition function $Z(\boldsymbol{\alpha})=H_{\alpha}=\int\prod_{k}\pi_{k}^{\alpha_{k}}\,d\nu$ that anchors the family of divergences treated here is the starting point of a related development~\cite{balsubramani2026information}, which develops a \emph{mixed coincidence calculus} for general real exponent vectors $\boldsymbol{\alpha}\in\R^{W}$ (including non-simplex and negative exponents) and unnormalized factors.
That paper proves an identity packaging four perspectives on $\log Z(\boldsymbol{\alpha})$:
a Boltzmann coincidence weight, an exponential-family normalizer, the value of an unconstrained max-entropy Lagrangian, and the optimum of a relative-entropy barycenter problem.
The present paper imposes data-processing monotonicity and product additivity on top, and asks which positive combinations of these coincidence-weight building blocks yield divergences satisfying both axioms; the multi-way coincidence divergence $\Cdiv_{\alpha}=-\log Z(\boldsymbol{\alpha})$ is the shared atom of the two papers' calculi.
The structural identification of non-simplex / mixed-sign exponents with reversed-direction (repulsive) log-loss constraints is established in~\cite{balsubramani2026information} as a feature of the general real-exponent identity; the present paper uses that observation to interpret the mixed-sign cones $\Aminus$ as the multi-distribution analogue of R\'enyi orders $>1$ (Section~\ref{sec:exotic} witnesses this identification explicitly).

\section{Discussion}\label{sec:open}

\subsection{Inputs and scope of the representation}\label{ssec:scope}

The proof of Theorem~\ref{thm:correct} draws one result from outside this paper, the large-sample theorem~\cite[Theorem~19]{farooq2024matrix}, used in Lemma~\ref{lem:profile-monotone}.
Its hypotheses are strict inequalities between the atoms of two tuples, and the lemma meets them by tensoring with a fixed tuple whose profile is bounded away from zero on the compactified atom set.
The atom set agrees with the classification in~\cite[Propositions~13--14]{farooq2024matrix} of the nondegenerate monotone homomorphisms and the monotone derivations of the matrix-majorization semiring, which holds for tuples of probability vectors with no further hypothesis.
The categorical reading of Section~\ref{ssec:semiring} is a dictionary and adds no hypothesis to the proof.

Theorem~\ref{thm:correct} requires uniformly bounded log-likelihood ratios.
That requirement cannot simply be dropped: on a larger class an additive, data-processing monotone functional can vanish on every bounded tuple and still be positive.

\begin{proposition}[Boundedness cannot be dropped]\label{prop:unbounded}
Fix $k$ and $\beta\in\Bminus$ with $\beta_{k}=\beta_{\star}=1$, let $\mathcal{U}_{\beta}$ be the class of $W$-tuples of mutually absolutely continuous distributions with $\sup_{n\in\N}\Dren_{e_{k}+n\beta}(\boldsymbol{\pi})/n<\infty$, fix a Banach limit $\mathrm{LIM}$ on bounded real sequences, and let $\Phi(\boldsymbol{\pi}):=\mathrm{LIM}_{n}\,\Dren_{e_{k}+n\beta}(\boldsymbol{\pi})/n$.
Then $\mathcal{U}_{\beta}$ contains every bounded tuple and every Gaussian location tuple $(\mathcal{N}(\mu_{j},\sigma^{2}))_{j}$ and is closed under products and Markov kernels, $\Phi$ is monotone under joint data processing, additive on products and zero on bounded tuples, and $\Phi=(\sum_{j}\beta_{j}\mu_{j})^{2}/(2\sigma^{2})$ on Gaussian location tuples.
In particular no triple of finite measures represents $\Phi$ through~\eqref{eq:correct} on $\mathcal{U}_{\beta}$.
\end{proposition}

The proof is in Appendix~\ref{app:recipe-proofs}.
For a single random variable the corresponding question has a sharper answer.
When every exponential moment is finite, the monotone additive statistics are the mixtures of normalized cumulant generating functions over compactly supported measures~\cite[Theorem~2]{mu2024monotone}.
On laws satisfying Cram\'er's condition, or with a finite $p$-th moment for some $1\le p<\infty$, they are the nonnegative multiples of the expectation~\cite[Theorem~1]{fritz2019monotone}.

\subsection{Relaxing the axioms}\label{ssec:relaxing}

The ground state follows from the other two axioms together with finiteness on bounded tuples.
A tuple of identical components and the tuple on a one-point space are each obtained from the other by a Markov kernel, so data processing gives them the same value, and the one-point tuple is its own tensor square, so additivity forces that value to vanish.

Without additivity the admissible class contains the multi-prior $f$-divergences $D_{f}(\boldsymbol{\pi})=\E_{\nu}[f(\pi_{1},\dots,\pi_{W})]$ with $f:\R_{\ge 0}^{W}\to\R$ convex, positively homogeneous and zero at $(1,\dots,1)$~\cite{csisz2008axiomatic,kumar2016}, which satisfy joint data processing by Jensen's inequality.
They are affine under direct sums: for tuples $\boldsymbol{\pi},\boldsymbol{\pi}'$ on disjoint spaces and $\lambda\in[0,1]$, the tuple $\lambda\boldsymbol{\pi}\oplus(1-\lambda)\boldsymbol{\pi}'$ on the union has $D_{f}=\lambda D_{f}(\boldsymbol{\pi})+(1-\lambda)D_{f}(\boldsymbol{\pi}')$.
For $\alpha\in(\Aplus\cup\Aminus)\setminus E$ the coincidence divergence is not affine in this sense, since $\Cdiv_{\alpha}(\lambda\boldsymbol{\pi}\oplus(1-\lambda)\boldsymbol{\pi}')=-\log\big(\lambda H_{\alpha}(\boldsymbol{\pi})+(1-\lambda)H_{\alpha}(\boldsymbol{\pi}')\big)$ lies strictly below $\lambda\Cdiv_{\alpha}(\boldsymbol{\pi})+(1-\lambda)\Cdiv_{\alpha}(\boldsymbol{\pi}')$ whenever $0<\lambda<1$ and $H_{\alpha}(\boldsymbol{\pi})\ne H_{\alpha}(\boldsymbol{\pi}')$, so it is not an $f$-divergence.
Replacing additivity by the pseudo-additivity $S(\boldsymbol{\pi}\otimes\boldsymbol{\pi}')=S(\boldsymbol{\pi})+S(\boldsymbol{\pi}')+(1-q)S(\boldsymbol{\pi})S(\boldsymbol{\pi}')$ of the nonextensive program~\cite{tsallis2009introduction}, with $q<1$, enlarges nothing.
The map $S\mapsto\log(1+(1-q)S)/(1-q)$ is increasing, fixes $0$, and turns pseudo-additivity into additivity, so it maps the pseudo-additive functionals that satisfy joint data processing and the ground state bijectively onto the $W$-way DPI--additive divergences, with inverse $D\mapsto(e^{(1-q)D}-1)/(1-q)$.
Weakening $\Sym$-invariance to invariance under a subgroup $G\le\Sym$ makes the measures $m^{\Dren}$ and $m^{\Tdiv}$ $G$-invariant, with $G$-orbits in place of $\Sym$-orbits in Section~\ref{sec:perm}.

The cost of information is a related axiomatization~\cite{pomatto2023cost}, for experiments with finitely many states whose log-likelihood ratios have finite moments.
A dilution runs an experiment with probability $\lambda$ and is otherwise uninformative.
A cost that is invariant under Blackwell equivalence, additive over independent experiments, linear under dilution and uniformly continuous with respect to a distance built from total variation and moments of the log-likelihood ratios is a nonnegative combination of the pairwise relative entropies~\cite[Theorem~1]{pomatto2023cost}.

\section{Conditional multi-way coincidence calculus}\label{sec:conditional}

The unconditional setting of Theorem~\ref{thm:correct} characterizes DPI--additive functionals on bounded $W$-tuples.
A natural extension adds \emph{side information}: the agent observes a second random variable $G$ on alphabet $\mathcal{G}$ before evaluating the $W$-tuple, and the divergence is a functional of conditional priors $\pi_{k|G}:\mathcal{X}\times\mathcal{G}\to[0,1]$ together with the marginal $p_G$ on $\mathcal{G}$.
The bivariate $W=2$ conditional R\'enyi divergence has several classical formulations; the most operationally natural one~\cite{bleuler2020} extends to general $W$ in~\cite{ducuara2026}.
For conditional \emph{entropies} the form is known~\cite[Theorem~7.1]{rubboli2026conditional}.
Every conditional entropy that is invariant under relabeling and embedding, monotone under conditionally mixing channels, additive and normalized is a barycenter of the functionals \(\mathbb{H}_{t,\tau}(X|Y)=\frac{1}{t}\log\sum_{y}P(y)\exp\!\big(t\!\int_{[0,\infty]}\Renent_{\alpha}(X|Y{=}y)\,d\tau(\alpha)\big)\) and their limiting forms, where $\tau$ is a probability measure and $\Renent_{\alpha}$ is the R\'enyi entropy.
The exact parameter sets are left open there~\cite[Conjecture~7.5]{rubboli2026conditional}.
With the aggregator in a range that depends on the stratum, the conditional coincidence atoms below satisfy data-processing inequalities for both systems and are additive, and every stratum of the unconditional atom set reappears among the conditional divergences.

\subsection{Setup and the conditional atom}\label{ssec:cond-setup}

Fix finite alphabets $\mathcal{X}$ and $\mathcal{G}$.
A \emph{conditional $W$-tuple} is a pair $(\boldsymbol{\pi}_{|G},p_G)$ where $\boldsymbol{\pi}_{|G}=(\pi_{1|G},\dots,\pi_{W|G})$ is a $W$-tuple of conditional pmfs and $p_G$ is a marginal pmf on $\mathcal{G}$.
It is \emph{bounded} when, for every $g$ with $p_G(g)>0$, the components of $\boldsymbol{\pi}_{|g}:=(\pi_{1|G}(\cdot|g),\dots,\pi_{W|G}(\cdot|g))$ have a common support.
The unconditional setting of Section~\ref{sec:proof} is the special case $|\mathcal{G}|=1$.
For $\alpha\in\Aplus\cup\Aminus$ and a parameter $\lambda\in(0,\infty]$, define the \emph{conditional multi-way coincidence atom}
\begin{equation}\label{eq:cond-atom}
\Cdiv_{\alpha,\lambda}(\boldsymbol{\pi}_{|G}\,\|\,p_G)
\;:=\; -\lambda\,\log\,\mathbb{E}_{g\sim p_G}\!\left[H_{\alpha}(\boldsymbol{\pi}_{|g})^{1/\lambda}\right]
\end{equation}
with the convention that $\lambda=\infty$ gives the conditioning average $\Cdiv_{\alpha,\infty}(\boldsymbol{\pi}_{|G}\,\|\,p_G):=\mathbb{E}_{g\sim p_G}[\Cdiv_{\alpha}(\boldsymbol{\pi}_{|g})]$, the $\lambda\to\infty$ limit of~\eqref{eq:cond-atom} and the direct conditional analogue of the conditional entropy of~\cite{rubboli2026conditional}.
For $\alpha\in\Aminus$ the integrand $\prod_{k}\pi_{k|G}(x|g)^{\alpha_{k}}$ in $H_{\alpha}(\boldsymbol{\pi}_{|g})$ is taken to be $0$ at outcomes outside the common support.
The opposite $\lambda\to 0^{+}$ endpoint is the best-case aggregator $\min_{g:p_G(g)>0}\Cdiv_{\alpha}(\boldsymbol{\pi}_{|g})$.
For $\alpha\in\Aplus\setminus E$ it is not monotone under processing of the conditioning system, so it is excluded from the data-processing-monotone family below.
In a conditioning system with two states, one with identical components and the other with pairwise distinct ones, merging the two states raises the minimum from $0$ to a positive value.
The atom recovers the known special cases:
\begin{enumerate}[leftmargin=*,label=(\arabic*)]
\item \emph{Unconditional limit.} When $|\mathcal{G}|=1$,
$\Cdiv_{\alpha,\lambda}(\boldsymbol{\pi}\,\|\,1) = -\lambda\log H_{\alpha}(\boldsymbol{\pi})^{1/\lambda} = -\log H_{\alpha}(\boldsymbol{\pi}) = \Cdiv_{\alpha}(\boldsymbol{\pi})$ for every $\lambda>0$, recovering the unconditional atom.
\item \emph{BLP recovery.} For $W=2$, $\alpha=(\alpha_0,1-\alpha_0)$,
and $\lambda=\alpha_0$, the quantity $\Cdiv_{\alpha,\lambda}/(1-\alpha_0)$ recovers the BLP conditional R\'enyi divergence $D_{\alpha_0}^{\mathrm{BLP}}(\pi_{1|G}\|\pi_{2|G}\,|\,p_G)$ of order $\alpha_0$ as defined in~\cite[eq.~(7)]{bleuler2020}. For $\alpha_0<1$ this recovery sits at an aggregator $\lambda=\alpha_0<1$, outside the range $[1,\infty]$ of Lemma~\ref{lem:cond-dpi-cond}, and the recovery is an algebraic identity.
\item \emph{Multi-prior BLP.} For general $W$ and every $\lambda>0$, the rescaled atom $\Cdiv_{\alpha,\lambda}/(\alpha_{\star}-1)$ is the multivariate conditional R\'enyi divergence $D_{\underline{\alpha},\lambda}$ of \cite[Def.~2]{ducuara2026} up to the convention sign.
\end{enumerate}

\subsection{Joint data-processing inequalities}\label{ssec:cond-dpi}

The conditional atom \eqref{eq:cond-atom} satisfies two distinct data-processing inequalities, one for each system.

\begin{lemma}[DPI w.r.t.\ the main system]\label{lem:cond-dpi-main}
For every stochastic kernel $T_{Y|XG}$ acting on $X$ and possibly depending on $G$, and every bounded conditional $W$-tuple $(\boldsymbol{\pi}_{|G},p_G)$,
\[
\Cdiv_{\alpha,\lambda}(T_{Y|XG}\circ\boldsymbol{\pi}_{|G}\,\|\,p_G)
\;\le\;\Cdiv_{\alpha,\lambda}(\boldsymbol{\pi}_{|G}\,\|\,p_G)
\]
for $\alpha\in\Aplus$ and $\lambda\in(0,\infty]$, and the reverse inequality holds for $\alpha\in\Aminus$.
\end{lemma}

The proof is in Appendix~\ref{app:proofs-conditional-multi-way-coincidence-calcul}.

\begin{lemma}[DPI w.r.t.\ the conditioning system]\label{lem:cond-dpi-cond}
Let $T_{H|G}$ be a stochastic kernel from $\mathcal{G}$ to a new alphabet $\mathcal{H}$, write $q_H = T_{H|G}(p_G)$, and let $\pi^{(k)}_{|H}(x|h) = \sum_g (p_G(g)/q_H(h))\,t_{H|G}(h|g)\,\pi_{k|G}(x|g)$ be the post-processed conditional pmfs of a bounded conditional $W$-tuple $(\boldsymbol{\pi}_{|G},p_G)$.
Then for every $\alpha\in\Aplus$ and $\lambda\in[1,\infty]$,
\[
\Cdiv_{\alpha,\lambda}(\boldsymbol{\pi}_{|H}\,\|\,q_H)
\;\le\;\Cdiv_{\alpha,\lambda}(\boldsymbol{\pi}_{|G}\,\|\,p_G)
\]
\end{lemma}

The proof is in Appendix~\ref{app:proofs-conditional-multi-way-coincidence-calcul}.
The two inequalities combine: any composition of a main-system stochastic operator and a conditioning-system stochastic operator is contractive on $\Cdiv_{\alpha,\lambda}$ when $\alpha\in\Aplus$ and $\lambda\in[1,\infty]$.

\subsection{Additivity and the conditional strata}\label{ssec:cond-rep}

The conditional atom is additive on independent extensions:
\[
\Cdiv_{\alpha,\lambda}((\boldsymbol{\pi}\otimes\boldsymbol{\pi}')_{|G\otimes G'}\,\|\,p_G\otimes p_{G'})
= \Cdiv_{\alpha,\lambda}(\boldsymbol{\pi}_{|G}\,\|\,p_G) + \Cdiv_{\alpha,\lambda}(\boldsymbol{\pi}'_{|G'}\,\|\,p_{G'})
\]
by multiplicativity of $H_{\alpha}$ and factorization of the joint expectation $\mathbb{E}_{(g,g')\sim p_G\otimes p_{G'}}$.
A \emph{conditional $W$-way divergence} is a map $\widetilde D$ from bounded conditional $W$-tuples over all pairs of finite alphabets to $[0,\infty)$ that is monotone under main-system and conditioning-system stochastic operators as in Lemmas~\ref{lem:cond-dpi-main} and~\ref{lem:cond-dpi-cond}.
It must also be additive on independent extensions as in the display above and vanish on ground states $((\pi,\dots,\pi)_{|G},p_G)$.
With $|\mathcal{G}|=1$ these are the axioms of Section~\ref{def:axioms} on finite alphabets, so $\boldsymbol{\pi}\mapsto\widetilde D(\boldsymbol{\pi}\,\|\,1)$ is a $W$-way DPI--additive divergence.

\begin{proposition}[Conditional atoms in every stratum]\label{prop:cond-strata}
The following are conditional $W$-way divergences:
\begin{enumerate}[label=(\alph*)]
\item $\Cdiv_{\alpha,\lambda}$ for $\alpha\in\Aplus\setminus E$ and $\lambda\in[1,\infty]$;
\item $\Cdiv_{\alpha,\lambda}/(1-\alpha_{\star})$ for $\alpha\in\Aminus\setminus E$ and $\lambda\in(0,1]$;
\item $(\boldsymbol{\pi}_{|G},p_G)\mapsto\max_{g:\,p_G(g)>0}\Tdiv_{\beta}(\boldsymbol{\pi}_{|g})$ for $\beta\in\Bminus\setminus\{0\}$;
\item $(\boldsymbol{\pi}_{|G},p_G)\mapsto\E_{g\sim p_G}[\KLdiv(\pi_{k|G}(\cdot|g)\,\|\,\pi_{\ell|G}(\cdot|g))]$ for $k\ne\ell$.
\end{enumerate}
With $|\mathcal{G}|=1$ they reduce to the atoms $\Cdiv_{\alpha}$, $\Dren_{\alpha}$, $\Tdiv_{\beta}$ and $\KLdiv(\pi_{k}\|\pi_{\ell})$.
\end{proposition}

\begin{corollary}[Every stratum is needed with side information]\label{cor:cond-forced}
No finite measure $\widetilde\mu$ on $(\Aplus\setminus E)\times[1,\infty]$ makes $\int\Cdiv_{\alpha,\lambda}\,d\widetilde\mu(\alpha,\lambda)$ equal to one of the divergences (b), (c) or (d) of Proposition~\ref{prop:cond-strata} on all bounded conditional $W$-tuples.
\end{corollary}

The proofs are in Appendix~\ref{app:proofs-conditional-multi-way-coincidence-calcul}; the corollary follows from Corollary~\ref{cor:necessity} applied with a trivial conditioning variable.
A characterization of all conditional $W$-way divergences, the conditional counterpart of Theorem~\ref{thm:correct}, is not known.
For conditional entropies the form is Theorem~7.1 of~\cite{rubboli2026conditional}, and the corresponding large-sample question for the conditional multivariate R\'enyi divergences is posed in~\cite[Section~IX]{ducuara2026}.

\subsection{Operational interpretation: side-information value in betting}\label{ssec:cond-op}

At the pinned aggregator $\lambda=\alpha_{\star}$, the conditional atom $\Cdiv_{\alpha,\alpha_{\star}}$ is, up to the $1/(\alpha_{\star}-1)$ rescaling and the convention sign, the multivariate conditional R\'enyi divergence $D_{\underline{\alpha},\alpha_{\star}}$ of \cite[Def.~2]{ducuara2026}.
That point lies outside the range $[1,\infty]$ of Lemma~\ref{lem:cond-dpi-cond}, since $\alpha_{\star}<1$ throughout $\Aplus\setminus E$, and the identification is algebraic.
Its operational reading in \cite[Sec.~VI]{ducuara2026}: with fair odds and a gambler betting optimally, it equals the log isoelastic certainty equivalent of a multi-lottery game played with access to $G$.
The agent is risk-averse with risk-aversion vector $R_k=1+\alpha_k/\alpha_{0}$ in the source's $0,\dots,W-1$ indexing. The increment that $G$ provides is the gap between that quantity and its unconditional counterpart.
At this aggregator the conditioning-system monotonicity is \cite[Prop.~2]{ducuara2026}, proved there by Minkowski's inequality for conditional pmfs whose non-focal priors do not depend on the conditioning variable: post-processing $G$ cannot increase the value of the side information.
Lemma~\ref{lem:cond-dpi-cond} proves the same monotonicity for $\lambda\in[1,\infty]$ without that restriction; the range is disjoint from $\lambda=\alpha_{\star}$ on $\Aplus\setminus E$, so the two results cover different parameters and neither contains the other.

\section{Summary}\label{sec:summary}

Theorem~\ref{thm:correct} characterizes the multi-distribution analogue of the classical binary representation~\cite{mu2021blackwell} in the classical-multivariate case.
The parameter space is not the simplex but a compactification of the four-stratum atom set (Proposition~\ref{prop:compactification}), whose strata are the simplex without its vertices, the signed-exponent (mixed-sign) cones, the tropical boundary at infinity, and the pairwise relative-entropy vertex edges.
The atoms span exactly the extreme rays of the cone of DPI--additive divergences, and every stratum is needed (Corollary~\ref{cor:necessity}).
The representation needs its boundedness hypothesis (Proposition~\ref{prop:unbounded}), and for three states the conditions that~\cite[Appendix~K]{mu2021blackwell} conjectured to be sufficient for large-sample dominance are not (Proposition~\ref{prop:mpst-counterexample}).
The logarithm itself is forced: additivity over independent repetitions is Cauchy's functional equation in tensor-power scale, and monotonicity under data processing leaves it only the logarithmic solution (Proposition~\ref{prop:log-forced}).
On a finite alphabet, each atom indexed by a rational point of the simplex is the exponential rate at which independent samples, one per distribution with sizes in fixed proportions, all show the same empirical distribution (Proposition~\ref{prop:coincidence-rate}).
The family is the multivariate Laplace transform of the log-loss vector, and its maximum over the simplex is a worst-case Kullback projection radius (Section~\ref{sec:readings}).
Each atom is also the exact free energy of a posterior over distributions, up to two information costs that vanish as resolution grows (Theorem~\ref{thm:laplace-mixed-concentration}).
The same characterization appears for finite alphabets, in a framework that also covers quantum divergences, in~\cite[Theorem~7 + Example~9]{haapasalo2025barycentric} via the preordered-semiring Vergleichsstellensatz; Section~\ref{ssec:semiring} documents the dictionary, and the proof of Section~\ref{sec:proof} keeps the four-stratum geometry visible through an explicit compactification.

The limit identities~\eqref{eq:KL-as-limit} and~\eqref{eq:trop-as-limit} are continuity of the atom profile at the boundary of the compactified atom set (Proposition~\ref{prop:compactification}), and the proof of the representation uses a single result from the matrix-majorization literature, the large-sample theorem of~\cite{farooq2024matrix}.

\begin{center}
\renewcommand{\arraystretch}{1.2}
\begin{tabular}{@{}p{0.26\linewidth}p{0.34\linewidth}p{0.34\linewidth}@{}}
\toprule
& \textbf{Binary ($W=2$,~\cite{mu2021blackwell})} & \textbf{Multi-way ($W>2$, here)}\\
\midrule
Hellinger transform & $H_{t}(\mu,\nu)=\int\mu^{t}\nu^{1-t}$ & $H_{\alpha}(\boldsymbol{\pi})=\int\prod_{k}\pi_{k}^{\alpha_{k}}$\\
Simplex atom & $R_{t}$, $t\in(0,1)$ & $\Dren_{\alpha}=\frac{1}{\alpha_{\star}-1}\log H_{\alpha}$, $\alpha\in\Aplus\setminus E$\\
Signed-exponent atom & $R_{t}$, $t>1$ & $\Dren_{\alpha}$, $\alpha\in\Aminus\setminus E$\\
Tropical atom & $R_{\infty}=\log\sup\mu/\nu$ & $\Tdiv_{\beta}=\frac{1}{\beta_{\star}}\log\sup\prod\pi_{k}^{\beta_{k}}$\\
Vertex atom & $R_{1}=\KLdiv$ & $\KLdiv(\pi_{k}\|\pi_{\ell})$, all pairs\\
Parameter space & $[1/2,\infty]$ compactified & compactified atom set $\overline{U}$ (Proposition~\ref{prop:compactification})\\
Spectral input & \cite[Thm.~1]{mu2021blackwell} & large-sample theorem~\cite[Thm.~19]{farooq2024matrix}\\
Symmetric form & $\int(R_{t}+R_{t}')\,dm(t)$ & $\int_{\paramset/\Sym}\Phi_{[\xi]}\,d\bar m([\xi])$\\
\bottomrule
\end{tabular}
\end{center}

\appendix
\appendixtitle

\section{Proofs of atom-level results}\label{app:atom-proofs}

This appendix collects the proofs of Proposition~\ref{prop:log-forced}, Lemma~\ref{lem:atoms-are-divergences} and Corollary~\ref{cor:converse}, grouped by the section in which each was stated.

\subsection{Proofs from Section~\ref{sec:functionaleq}: the forcing of the logarithm}

\begin{proof}[Proof of Proposition~\ref{prop:log-forced}]
Two coordinates of $\alpha$ are positive, say $\alpha_{k},\alpha_{\ell}>0$ with $k\ne\ell$.
For a bounded tuple whose log-likelihood ratios are bounded by $B$, $\prod_{j}\pi_{j}^{\alpha_{j}}=\pi_{1}\prod_{j}(\pi_{j}/\pi_{1})^{\alpha_{j}}\ge e^{-B}\pi_{1}$, so $H_{\alpha}(\boldsymbol{\pi})\ge e^{-B}>0$, and H\"older's inequality with the exponents $1/\alpha_{j}$ over the indices with $\alpha_{j}>0$ gives $H_{\alpha}(\boldsymbol{\pi})\le\prod_{j}(\int\pi_{j}\,d\nu)^{\alpha_{j}}=1$.

\emph{Every value in $(0,1]$ is attained.}
On a two-point space with counting measure and for $\epsilon\in(0,\tfrac12]$, let $\pi_{k}=(1-\epsilon,\epsilon)$, $\pi_{\ell}=(\epsilon,1-\epsilon)$ and $\pi_{j}=(\tfrac12,\tfrac12)$ for $j\notin\{k,\ell\}$.
These tuples are bounded, $\epsilon\mapsto H_{\alpha}(\boldsymbol{\pi})$ is continuous, it equals $1$ at $\epsilon=\tfrac12$, and it is at most $\epsilon^{\alpha_{\ell}}+\epsilon^{\alpha_{k}}$, which tends to $0$ with $\epsilon$.
By the intermediate value theorem every $h\in(0,1]$ is the Hellinger transform of a bounded tuple.

\emph{The function $\varphi$ is nonincreasing.}
Let $0<x\le y\le 1$, and take a bounded tuple $\boldsymbol{\pi}$ on the two-point space with $H_{\alpha}(\boldsymbol{\pi})=x$.
For $s\in[0,1]$ let $K_{s}$ be the Markov kernel that keeps each point with probability $1-s$ and otherwise replaces it by a uniform draw.
Each $K_{s}\boldsymbol{\pi}$ is bounded, $s\mapsto H_{\alpha}(K_{s}\boldsymbol{\pi})$ is continuous, and $H_{\alpha}(K_{1}\boldsymbol{\pi})=1$ because every component of $K_{1}\boldsymbol{\pi}$ is uniform.
Choosing $s$ with $H_{\alpha}(K_{s}\boldsymbol{\pi})=y$, joint data processing gives $\varphi(y)=D(K_{s}\boldsymbol{\pi})\le D(\boldsymbol{\pi})=\varphi(x)$.

\emph{Cauchy's equation.}
The Hellinger transform is multiplicative on products, so $H_{\alpha}(\boldsymbol{\pi}^{\otimes n})=H_{\alpha}(\boldsymbol{\pi})^{n}$, and additivity gives $\varphi(x^{n})=n\,\varphi(x)$ for every $x\in(0,1]$ and $n\in\N$; the ground state gives $\varphi(1)=0$.
Set $\psi(u):=\varphi(e^{-u})$ for $u\ge 0$.
Then $\psi$ is nondecreasing, $\psi(0)=0$ and $\psi(nu)=n\,\psi(u)$.
For positive integers $m,n$, $n\,\psi(\tfrac{m}{n}u)=\psi(mu)=m\,\psi(u)$, so $\psi(ru)=r\,\psi(u)$ for every rational $r>0$.
For real $r>0$ and rationals $r_{1}\le r\le r_{2}$, monotonicity gives $r_{1}\psi(1)\le\psi(r)\le r_{2}\psi(1)$, hence $\psi(r)=r\,\psi(1)$.
With $c:=\psi(1)\ge 0$, $\varphi(h)=\psi(-\log h)=-c\log h$.
\end{proof}

\subsection{Proofs from Section~\ref{sec:atoms}: multi-way coincidence atoms}

\begin{proof}[Proof of Lemma~\ref{lem:atoms-are-divergences}.]
Three cases. In each, we verify the three axioms (joint DPI, additivity on products, ground state $D(\pi,\dots,\pi)=0$) in turn.

\textbf{(1) $\xi=\alpha\in(\Aplus\cup\Aminus)\setminus E$.}

\emph{Additivity.} The Hellinger transform is multiplicative under tensor products (property H1 in Section~\ref{ssec:hellinger}):
$H_{\alpha}(\boldsymbol{\pi}\otimes\boldsymbol{\pi}')=H_{\alpha}(\boldsymbol{\pi})H_{\alpha}(\boldsymbol{\pi}')$.
Taking $\frac{1}{\alpha_{\star}-1}\log$ gives $\Dren_{\alpha}(\boldsymbol{\pi}\otimes\boldsymbol{\pi}')=\Dren_{\alpha}(\boldsymbol{\pi})+\Dren_{\alpha}(\boldsymbol{\pi}')$.

\emph{Joint DPI.}
Let $K$ be a Markov kernel, let $\bar\nu:=\frac1W\sum_{k}\pi_{k}\nu$, and let $u_{k}:=d(\pi_{k}\nu)/d\bar\nu$, so that $u=(u_{1},\dots,u_{W})$ is bounded and lies in $(0,\infty)^{W}$ $\bar\nu$-almost everywhere, since the components share their support.
Since $\sum_{k}\alpha_{k}=1$, $G_{\alpha}(u):=\prod_{k}u_{k}^{\alpha_{k}}$ is positively homogeneous of degree one, so $H_{\alpha}(\boldsymbol{\pi})=\int G_{\alpha}(u)\,d\bar\nu$, and the Hellinger transform may be computed against any dominating measure.
Under the joint law $\gamma(dx,dy):=\bar\nu(dx)K(x,dy)$ of $(X,Y)$, the densities of the output components with respect to the image $K\bar\nu$ of $\bar\nu$ are $v_{k}(Y)=\E_{\gamma}[u_{k}(X)\,|\,Y]$, so $H_{\alpha}(K\boldsymbol{\pi})=\E_{\gamma}[G_{\alpha}(\E_{\gamma}[u(X)|Y])]$.
For $\alpha\in\Aplus$ the function $G_{\alpha}$ is concave on $[0,\infty)^{W}$, and the conditional Jensen inequality gives $H_{\alpha}(K\boldsymbol{\pi})\ge\E_{\gamma}[G_{\alpha}(u(X))]=H_{\alpha}(\boldsymbol{\pi})$.
For $\alpha\in\Aminus$, with $\alpha_{j}\ge 1$ and $\alpha_{k}\le 0$ for $k\ne j$, $G_{\alpha}$ is the perspective $u_{j}\,\phi(u_{k}/u_{j}:k\ne j)$ of the function $\phi(v)=\prod_{k\ne j}v_{k}^{\alpha_{k}}$, which is convex on $(0,\infty)^{W-1}$ because $\log\phi$ is, so $G_{\alpha}$ is convex on $(0,\infty)^{W}$ and the conditional Jensen inequality gives $H_{\alpha}(K\boldsymbol{\pi})\le H_{\alpha}(\boldsymbol{\pi})$.
The signs of $\log H_{\alpha}$ follow from the one-point kernel $K$ that sends every $x$ to a single point, whose output tuple has $H_{\alpha}=1$.
Hence $\log H_{\alpha}$ and the multiplier $\frac{1}{\alpha_{\star}-1}$ have the same sign on each region:
\begin{itemize}[leftmargin=*]
\item On $\Aplus$ (with $\alpha_{\star}<1$), $H_{\alpha}\le 1$, so $\log H_{\alpha}\le 0$, and $\frac{1}{\alpha_{\star}-1}<0$; their product $\Dren_{\alpha}\ge 0$.
Garbling moves $\log H_{\alpha}$ upward toward $0$, and the negative multiplier reverses the direction: $\Dren_{\alpha}(K\boldsymbol{\pi})\le\Dren_{\alpha}(\boldsymbol{\pi})$.
\item On $\Aminus$ (with $\alpha_{\star}>1$), $H_{\alpha}\ge 1$, so $\log H_{\alpha}\ge 0$, and $\frac{1}{\alpha_{\star}-1}>0$; their product $\Dren_{\alpha}\ge 0$.
Garbling moves $\log H_{\alpha}$ downward toward $0$, and the positive multiplier preserves the direction: $\Dren_{\alpha}(K\boldsymbol{\pi})\le\Dren_{\alpha}(\boldsymbol{\pi})$.
\end{itemize}

\emph{Ground state.}
$H_{\alpha}(\pi,\dots,\pi)=\E_{\nu}[\pi^{\sum_{k}\alpha_{k}}]=\E_{\nu}[\pi]=1$ since $\sum_{k}\alpha_{k}=1$, so $\Dren_{\alpha}(\pi,\dots,\pi)=0$.

\textbf{(2) $\xi=\beta\in\Bminus\setminus\{0\}$.}

\emph{Additivity.} Since $\prod_{k}(\pi_{k}\otimes\pi'_{k})^{\beta_{k}}(x,x')=\prod_{k}\pi_{k}^{\beta_{k}}(x)\cdot\prod_{k}\pi'^{\beta_{k}}_{k}(x')$,
the supremum over the product space factorizes and $\Tdiv_{\beta}(\boldsymbol{\pi}\otimes\boldsymbol{\pi}')=\Tdiv_{\beta}(\boldsymbol{\pi})+\Tdiv_{\beta}(\boldsymbol{\pi}')$ after taking $\frac{1}{\beta_{\star}}\log$.

\emph{Joint DPI.} Normalize $\beta_{\star}=1$, which changes nothing since $\Tdiv_{c\beta}=\Tdiv_{\beta}$ for $c>0$.
On the cone of $\Bminus$ at $e_{j}$ the remaining exponents $w_{\ell}:=-\beta_{\ell}$, $\ell\ne j$, are nonnegative and sum to one.
With $\bar\nu$, $u$ and $\gamma$ as in case~(1), and since $\sum_{k}\beta_{k}=0$, the tropical atom is unchanged when the densities are taken with respect to $\bar\nu$, so $\Tdiv_{\beta}(\boldsymbol{\pi})=\log M$ with $M:=\operatorname{ess\,sup}_{\bar\nu}u_{j}/\prod_{\ell\ne j}u_{\ell}^{w_{\ell}}$, and $u_{j}\le M\prod_{\ell\ne j}u_{\ell}^{w_{\ell}}$ holds $\bar\nu$-almost everywhere.
The output densities satisfy
\[
v_{j}(Y)=\E_{\gamma}[u_{j}(X)\,|\,Y]\le M\,\E_{\gamma}\Big[\prod_{\ell\ne j}u_{\ell}(X)^{w_{\ell}}\,\Big|\,Y\Big]\le M\prod_{\ell\ne j}v_{\ell}(Y)^{w_{\ell}}
\]
where the last step is the conditional Jensen inequality for the weighted geometric mean, which is concave.
Hence $\Tdiv_{\beta}(K\boldsymbol{\pi})\le\log M=\Tdiv_{\beta}(\boldsymbol{\pi})$.
For $W=2$ this is the data-processing inequality for $R_{\infty}$ (van Erven and Harremo\"es \cite{erven2014b}, Theorem 9).

\emph{Ground state.} Since $\sum_{k}\beta_{k}=0$,
$\prod_{k}\pi^{\beta_{k}}=\pi^{\sum_{k}\beta_{k}}=\pi^{0}=1$,
so $\sup_{x}\prod_{k}\pi^{\beta_{k}}(x)=1$ and $\Tdiv_{\beta}(\pi,\dots,\pi)=\frac{1}{\beta_{\star}}\log 1=0$.

\emph{Non-negativity.} See the H\"older argument following~(\ref{eq:tropical}):
on the cone $\Bminus^{(k)}$, where $\beta_{k}\ge 0$ and $\beta_{\ell}\le 0$ for $\ell\ne k$, normalize $\beta_{\star}=1$ (free, since $\Tdiv_{c\beta}=\Tdiv_{\beta}$ for $c>0$).
The exponents $-\beta_{\ell}$ then sum to $1$, so the weighted geometric mean $g=\prod_{\ell\ne k}\pi_{\ell}^{-\beta_{\ell}}$ has $\int g\,\dnu\le 1$ while $\int\pi_{k}\,\dnu=1$, and $\prod_{m}\pi_{m}^{\beta_{m}}=\pi_{k}/g$ has supremum at least $1$.

\textbf{(3) $\xi=(k,\ell)$.} The classical relative entropy $\KLdiv(\pi_{k}\|\pi_{\ell})$ is DPI--additive in $(\pi_{k},\pi_{\ell})$ (\cite{cover1999elements}) and is well-defined as a $W$-way functional.
It depends only on two of the priors, but the joint kernel $K$ acts identically on the pair, preserving DPI.
The likelihood ratio is multiplicative on tensor products, so $\KLdiv$ is additive on independent factors.
Ground-state $\KLdiv(\pi\|\pi)=0$ is immediate. Non-negativity is Gibbs' inequality.
\end{proof}

\subsection{Proofs from Section~\ref{sec:proof}: the converse direction}

\begin{proof}[Proof of Corollary~\ref{cor:converse}.]
Each atom is DPI--additive by Lemma~\ref{lem:atoms-are-divergences}; positive linear combinations of DPI--additive divergences are DPI--additive (joint DPI is preserved under positive sums: $\sum_{i}\lambda_{i}D_{i}(K\boldsymbol{\pi})\le\sum_{i}\lambda_{i}D_{i}(\boldsymbol{\pi})$ when each $D_{i}$ is monotone under $K$).
For positive integrals against finite Borel measures, fix a bounded tuple $\boldsymbol{\pi}$.
The atoms are nonnegative and Borel in $\xi$.
The map $\alpha\mapsto\Dren_{\alpha}(\boldsymbol{\pi})$ is continuous on $(\Aplus\cup\Aminus)\setminus E$ by dominated convergence.
The map $\beta\mapsto\beta_{\star}\Tdiv_{\beta}(\boldsymbol{\pi})$ is the essential supremum of the functions $\sum_{k}\beta_{k}\log(\pi_{k}/\pi_{1})$, which are linear in $\beta$ and bounded on the common support.
It is therefore convex and finite on the hyperplane $\{\beta:\sum_{k}\beta_{k}=0\}$, and hence continuous there.
As $\beta_{\star}$ is continuous and positive on $\Bminus\setminus\{0\}$, $\beta\mapsto\Tdiv_{\beta}(\boldsymbol{\pi})$ is continuous there.
So the integrals in~\eqref{eq:correct} are well defined in $[0,\infty]$, and they are finite by hypothesis.
Joint DPI is then preserved by integration with respect to the kernel $K$: integrating the pointwise (in $\xi$) inequality $\Phi_{\xi}(K\boldsymbol{\pi})\le\Phi_{\xi}(\boldsymbol{\pi})$ against the (positive) measure $dm$ on each component gives $D(K\boldsymbol{\pi})\le D(\boldsymbol{\pi})$.
Additivity on products follows analogously: integrate the pointwise (in $\xi$) identity $\Phi_{\xi}(\boldsymbol{\pi}\otimes\boldsymbol{\pi}')=\Phi_{\xi}(\boldsymbol{\pi})+\Phi_{\xi}(\boldsymbol{\pi}')$ against $dm$.
Ground state $D(\pi,\dots,\pi)=0$ is preserved because every atom satisfies it pointwise, and the integral of zero is zero.
Finiteness on bounded tuples is exactly that integrability hypothesis.
\end{proof}

\section{Proofs for Section~\ref{sec:proof}}\label{app:recipe-proofs}

Throughout this appendix a bounded tuple on a finite alphabet is a $W$-tuple of probability vectors on a finite set with a common support, identified with its restriction to that support, on which every coordinate is positive.
A deterministic kernel in each direction relates a tuple to its restriction, so every $W$-way DPI--additive divergence and every atom takes the same value on both.
For such a tuple write $\Lambda(\boldsymbol{\pi}):=\max_{k\ne\ell}\max_{x}|\log(\pi_{k}(x)/\pi_{\ell}(x))|$ and $m_{k}(\boldsymbol{\pi}):=\min_{x}\pi_{k}(x)$, and for $w\in S_{k}$ write $\Delta^{(k)}_{w}(\boldsymbol{\pi}):=\sum_{\ell\ne k}w_{\ell}\KLdiv(\pi_{k}\|\pi_{\ell})$.
The proof of Theorem~\ref{thm:correct} uses two further results, stated first.

\begin{lemma}[Monotonicity in the profile]\label{lem:profile-monotone}
Let $D$ be a $W$-way DPI--additive divergence on bounded tuples on finite alphabets, and let $\boldsymbol{\rho}$ be such a tuple with pairwise distinct components.
Then $c_{\boldsymbol{\rho}}:=\min_{\overline{U}}\Psi_{\boldsymbol{\rho}}>0$, the inequality $\Psi_{\boldsymbol{\pi}}\ge\Psi_{\boldsymbol{\pi}'}$ on $\overline{U}$ implies $D(\boldsymbol{\pi})\ge D(\boldsymbol{\pi}')$, and for all bounded tuples $\boldsymbol{\pi},\boldsymbol{\pi}'$ on finite alphabets
\[
|D(\boldsymbol{\pi})-D(\boldsymbol{\pi}')|\le\frac{D(\boldsymbol{\rho})}{c_{\boldsymbol{\rho}}}\,\sup_{\overline{U}}|\Psi_{\boldsymbol{\pi}}-\Psi_{\boldsymbol{\pi}'}|
\]
\end{lemma}

The argument for the lemma is the one of~\cite[Lemma~6]{haapasalo2025barycentric} and~\cite[Lemma~11]{haapasalo2026continuous}, where the role of $\boldsymbol{\rho}$ is played by a power universal element.

\begin{proposition}[Beyond finite alphabets]\label{prop:measurable}
Let $D$ be a $W$-way DPI--additive divergence on bounded tuples, and suppose that a triple $(m^{\Dren},m^{\Tdiv},(c_{k\ell}))$ of finite measures and nonnegative coefficients represents $D$ through~\eqref{eq:correct} on bounded tuples on finite alphabets.
Then the same triple represents $D$ through~\eqref{eq:correct} on bounded tuples on every measurable space.
\end{proposition}

\subsection{The compactified atom set}

\begin{proof}[Proof of Proposition~\ref{prop:compactification}]
\emph{The embedding.}
The map $\phi$ is a homeomorphism of $\mathcal{A}$ onto $B$, because $r\mapsto r/(1+r)$ is a homeomorphism of $[0,\infty)$ onto $[0,1)$.
Each $s_{j}$ is continuous on $\mathcal{A}$ and vanishes only at $e_{j}$, so every $w^{(j)}$ is continuous on $U$, and so is $\iota$.
The first coordinate of $\iota$ is injective with continuous inverse $\phi^{-1}$, so $\iota$ is a homeomorphism of $U$ onto $\iota(U)$.
The set $\overline{U}$ is closed in the compact metric space $\bar{B}\times S_{1}\times\dots\times S_{W}$, hence compact and metrizable.

\emph{The boundary points.}
For $w\in S_{k}$ let $f_{k}(w)$ be the point with first coordinate $\phi(e_{k})$, $k$-th coordinate $w$, and $j$-th coordinate $e_{k}\in S_{j}$ for $j\ne k$.
Let $\tau_{k}(w)$ be the point with first coordinate $\beta^{k,w}/|\beta^{k,w}|_{2}$, $k$-th coordinate $w$, and, for $j\ne k$, $j$-th coordinate the vector of $S_{j}$ with entry $1/(2-w_{j})$ at $k$ and entry $w_{i}/(2-w_{j})$ at each $i\notin\{j,k\}$.
Both maps are continuous and injective on the compact set $S_{k}$, since their $k$-th coordinate is $w$, so they are homeomorphisms onto their images $F_{k}$ and $T_{k}$.
Points of $\iota(U)$ have first coordinate in $B\setminus\phi(E)$, points of $F_{k}$ have first coordinate $\phi(e_{k})$, and points of $T_{k}$ have a first coordinate $y$ with $|y|_{2}=1$, $y_{k}>0$ and $y_{i}\le 0$ for $i\ne k$, so these sets are pairwise disjoint.

\emph{Sequences with a boundary limit.}
Let $\alpha_{n}\in U$.
Suppose first that $\alpha_{n}\to e_{k}$ and $w^{(k)}(\alpha_{n})\to w$.
Then $\phi(\alpha_{n})\to\phi(e_{k})$, and for $j\ne k$, $w^{(j)}(\alpha_{n})\to w^{(j)}(e_{k})=e_{k}$ because $s_{j}(e_{k})=1$; so $\iota(\alpha_{n})\to f_{k}(w)$.
Suppose next that $\alpha_{n}$ lies in the cone of $\Aminus$ at $e_{k}$, with $t_{n}:=s_{k}(\alpha_{n})\to\infty$ and $w_{n}:=w^{(k)}(\alpha_{n})\to w$, so that $\alpha_{n}=e_{k}+t_{n}\beta^{k,w_{n}}$.
For vectors $a=b+d$ with $a\ne 0$ and $b\ne 0$, $|a/|a|_{2}-b/|b|_{2}|_{2}\le 2|d|_{2}/|a|_{2}$, and $|\phi(a+c)-a/|a|_{2}|_{2}=1/(1+|a|_{2})$.
With $b=t_{n}\beta^{k,w_{n}}$, $d=e_{k}-c$ and $|\beta^{k,w_{n}}|_{2}\ge 1$, this gives $\phi(\alpha_{n})\to\beta^{k,w}/|\beta^{k,w}|_{2}$.
For $j\ne k$, $s_{j}(\alpha_{n})=1+t_{n}(2-w_{n,j})$, and $w^{(j)}(\alpha_{n})$ has entry $(1+t_{n})/s_{j}(\alpha_{n})$ at $k$ and $t_{n}w_{n,i}/s_{j}(\alpha_{n})$ at $i\notin\{j,k\}$, which converge to the entries of the $j$-th coordinate of $\tau_{k}(w)$; so $\iota(\alpha_{n})\to\tau_{k}(w)$.

\emph{All limits.}
Conversely, suppose $\iota(\alpha_{n})\to z=(y,\omega)$.
If $|y|_{2}<1$, then $\alpha_{n}\to\alpha:=\phi^{-1}(y)$, which lies in the closed set $\Aplus\cup\Aminus$.
If $\alpha\in U$ then $z=\iota(\alpha)$, and if $\alpha=e_{k}$ then $w^{(k)}(\alpha_{n})\to\omega_{k}$ and $z=f_{k}(\omega_{k})$ by the first case above.
If $|y|_{2}=1$, then $|\alpha_{n}|_{2}\to\infty$, so $\alpha_{n}\in\Aminus$ for all large $n$, because $\Aplus$ is bounded and the cones of $\Aminus$ are disjoint.
A point $\alpha=e_{j}+t\beta^{j,w}$ of the cone at $e_{j}$ with $t>|e_{j}-c|_{2}$ satisfies, by the two estimates above, $|\phi(\alpha)-\beta^{j,w}/|\beta^{j,w}|_{2}|_{2}\le(1+2|e_{j}-c|_{2})/(t-|e_{j}-c|_{2})$, and $\beta^{j,w}/|\beta^{j,w}|_{2}$ has $j$-th coordinate at least $1/\sqrt{2}$, since $|\beta^{j,w}|_{2}^{2}=1+|w|_{2}^{2}\le 2$, and nonpositive other coordinates.
So a limit $y$ along points of the cone at $e_{j}$ has $y_{j}\ge 1/\sqrt{2}$ and $y_{i}\le 0$ for $i\ne j$, and $\alpha_{n}$ cannot lie in two different cones for infinitely many $n$.
Hence $\alpha_{n}$ lies in the cone at some $e_{k}$ for all large $n$, with $s_{k}(\alpha_{n})\to\infty$ and $w^{(k)}(\alpha_{n})\to\omega_{k}$, and $z=\tau_{k}(\omega_{k})$ by the second case above.
Since, as $n\to\infty$ through $n\ge 2$, $f_{k}(w)$ is the limit of $\iota(e_{k}+\tfrac1n(w-e_{k}))$ and $\tau_{k}(w)$ the limit of $\iota(e_{k}+n\beta^{k,w})$, where $e_{k}+\tfrac1n(w-e_{k})$ denotes the point of $\Aplus$ with coordinate $1-\tfrac1n$ at $k$ and $w_{\ell}/n$ at $\ell\ne k$, the closure is $\overline{U}=\iota(U)\cup\bigcup_{k}F_{k}\cup\bigcup_{k}T_{k}$.
The same analysis shows that $\iota(U)=\overline{U}\cap\{(y,\omega):y\in B\setminus\phi(E)\}$, which is open in $\overline{U}$, and it gives statements (i) and (ii).

\emph{Three estimates.}
Fix a bounded tuple $\boldsymbol{\pi}$ on a finite alphabet and write $\Lambda:=\Lambda(\boldsymbol{\pi})$.

(a) If $\alpha\in U$ has $s:=s_{k}(\alpha)\le\tfrac12$ and $w:=w^{(k)}(\alpha)$, then $|\Dren_{\alpha}(\boldsymbol{\pi})-\Delta^{(k)}_{w}(\boldsymbol{\pi})|\le\Lambda^{2}s/2$.
Indeed, a point of the cone at $e_{j}$ with $j\ne k$ has $s_{k}\ge\alpha_{j}\ge 1$, so either $\alpha\in\Aplus$, with $\alpha_{\ell}=sw_{\ell}$ for $\ell\ne k$ and $\alpha_{\star}=\alpha_{k}=1-s$, or $\alpha$ lies in the cone at $e_{k}$, with $\alpha_{\ell}=-sw_{\ell}$ and $\alpha_{\star}=1+s$.
In both cases $\alpha_{\ell}=\epsilon sw_{\ell}$ and $\alpha_{\star}-1=-\epsilon s$ with $\epsilon\in\{1,-1\}$.
Put $Z:=\sum_{\ell\ne k}\alpha_{\ell}\log(\pi_{\ell}/\pi_{k})$.
Since $\sum_{j}\alpha_{j}=1$, $\prod_{j}\pi_{j}^{\alpha_{j}}=\pi_{k}e^{Z}$, so $H_{\alpha}(\boldsymbol{\pi})=\E_{\pi_{k}}[e^{Z}]$, where $|Z|\le s\Lambda$ and $\E_{\pi_{k}}[Z]=-\epsilon s\Delta^{(k)}_{w}(\boldsymbol{\pi})$.
Jensen's inequality gives $\log H_{\alpha}\ge\E_{\pi_{k}}[Z]$, and Hoeffding's lemma, for a random variable with values in an interval of length $2s\Lambda$, gives $\log H_{\alpha}\le\E_{\pi_{k}}[Z]+s^{2}\Lambda^{2}/2$.
Dividing by $\alpha_{\star}-1=-\epsilon s$ places $\Dren_{\alpha}$ in $[\Delta^{(k)}_{w}-s\Lambda^{2}/2,\Delta^{(k)}_{w}]$ when $\epsilon=1$ and in $[\Delta^{(k)}_{w},\Delta^{(k)}_{w}+s\Lambda^{2}/2]$ when $\epsilon=-1$.

(b) For $w\in S_{k}$ and $t>0$, $0\le\Tdiv_{\beta^{k,w}}(\boldsymbol{\pi})-\Dren_{e_{k}+t\beta^{k,w}}(\boldsymbol{\pi})\le t^{-1}\log(1/m_{k}(\boldsymbol{\pi}))$.
Indeed, with $g:=\prod_{j}\pi_{j}^{\beta^{k,w}_{j}}$, the point $\alpha:=e_{k}+t\beta^{k,w}$ has $\prod_{j}\pi_{j}^{\alpha_{j}}=\pi_{k}g^{t}$ and $\alpha_{\star}-1=t$, and $m_{k}(\max g)^{t}\le H_{\alpha}=\sum_{x}\pi_{k}(x)g(x)^{t}\le(\max g)^{t}$, while $\log\max g=\Tdiv_{\beta^{k,w}}(\boldsymbol{\pi})$.

(c) For $w,w'\in S_{k}$, both $|\Tdiv_{\beta^{k,w}}(\boldsymbol{\pi})-\Tdiv_{\beta^{k,w'}}(\boldsymbol{\pi})|$ and $|\Delta^{(k)}_{w}(\boldsymbol{\pi})-\Delta^{(k)}_{w'}(\boldsymbol{\pi})|$ are at most $\Lambda|w-w'|_{1}$.
Indeed, $\Tdiv_{\beta^{k,w}}(\boldsymbol{\pi})=\max_{x}\sum_{\ell\ne k}w_{\ell}\log(\pi_{k}(x)/\pi_{\ell}(x))$ is a maximum of functions of $w$ that are $\Lambda$-Lipschitz for $|\cdot|_{1}$, and every $\KLdiv(\pi_{k}\|\pi_{\ell})$ lies in $[0,\Lambda]$.

\emph{Continuity of the profile.}
On the open set $\iota(U)$ the profile is continuous, because $H_{\alpha}(\boldsymbol{\pi})$ is a finite sum of positive continuous functions of $\alpha$ and $\alpha_{\star}-1$ vanishes on $\Aplus\cup\Aminus$ only at $E$.
Let $z_{n}\to f_{k}(w)$ in $\overline{U}$.
The sets $F_{j}$ for $j\ne k$ and $T_{j}$ for all $j$ are compact and do not contain $f_{k}(w)$, so $z_{n}\in\iota(U)\cup F_{k}$ for all large $n$.
Along terms $z_{n}=f_{k}(w_{n})$ of $F_{k}$, $w_{n}\to w$ and (c) gives $\Psi_{\boldsymbol{\pi}}(z_{n})\to\Psi_{\boldsymbol{\pi}}(f_{k}(w))$.
Along terms $z_{n}=\iota(\alpha_{n})$, statement (i) gives $s_{k}(\alpha_{n})\to 0$ and $w^{(k)}(\alpha_{n})\to w$, and (a) and (c) give $|\Psi_{\boldsymbol{\pi}}(z_{n})-\Delta^{(k)}_{w}(\boldsymbol{\pi})|\le\Lambda^{2}s_{k}(\alpha_{n})/2+\Lambda|w^{(k)}(\alpha_{n})-w|_{1}\to 0$.
Likewise, if $z_{n}\to\tau_{k}(w)$ then $z_{n}\in\iota(U)\cup T_{k}$ for all large $n$; along terms of $T_{k}$ estimate (c) applies, and along terms $z_{n}=\iota(\alpha_{n})$ statement (ii) and estimates (b) and (c) give $|\Psi_{\boldsymbol{\pi}}(z_{n})-\Tdiv_{\beta^{k,w}}(\boldsymbol{\pi})|\le\log(1/m_{k}(\boldsymbol{\pi}))/s_{k}(\alpha_{n})+\Lambda|w^{(k)}(\alpha_{n})-w|_{1}\to 0$.
\end{proof}

\subsection{Monotonicity in the profile}

\begin{proof}[Proof of Lemma~\ref{lem:profile-monotone}]
\emph{Positivity.}
For $\alpha\in\Aplus\setminus E$, H\"older's inequality with the exponents $1/\alpha_{k}$ over the indices with $\alpha_{k}>0$ gives $H_{\alpha}(\boldsymbol{\rho})\le 1$, with equality only if the $\rho_{k}$ with $\alpha_{k}>0$ coincide.
Off the vertices at least two coordinates are positive, so $H_{\alpha}(\boldsymbol{\rho})<1$ and $\Dren_{\alpha}(\boldsymbol{\rho})>0$.
For $\alpha\ne e_{k}$ in the cone at $e_{k}$, let $N:=\{\ell:\alpha_{\ell}<0\}$.
The pointwise identity $\rho_{k}=(\prod_{j}\rho_{j}^{\alpha_{j}})^{1/\alpha_{k}}\prod_{\ell\in N}\rho_{\ell}^{|\alpha_{\ell}|/\alpha_{k}}$ and H\"older's inequality with the exponents $\alpha_{k}$ and $\alpha_{k}/|\alpha_{\ell}|$, $\ell\in N$, whose reciprocals sum to $(1+\sum_{\ell\in N}|\alpha_{\ell}|)/\alpha_{k}=1$, give $1\le H_{\alpha}(\boldsymbol{\rho})^{1/\alpha_{k}}$.
Equality forces every $\rho_{\ell}$ with $\ell\in N$ to be a multiple of $\prod_{j}\rho_{j}^{\alpha_{j}}$, hence equal to a common $\bar\rho$, and then $\prod_{j}\rho_{j}^{\alpha_{j}}=\rho_{k}^{\alpha_{k}}\bar\rho^{1-\alpha_{k}}$ is a multiple of $\bar\rho$ only if $\rho_{k}=\bar\rho$.
Some $\alpha_{\ell}$ is negative, so $H_{\alpha}(\boldsymbol{\rho})>1$ and $\Dren_{\alpha}(\boldsymbol{\rho})>0$.
On $F_{k}$, $\Delta^{(k)}_{w}(\boldsymbol{\rho})\ge\min_{\ell\ne k}\KLdiv(\rho_{k}\|\rho_{\ell})>0$.
On $T_{k}$, suppose $\Tdiv_{\beta^{k,w}}(\boldsymbol{\rho})\le 0$.
Then $\rho_{k}(x)\le\prod_{\ell}\rho_{\ell}(x)^{w_{\ell}}\le\sum_{\ell}w_{\ell}\rho_{\ell}(x)$ for every $x$, and summing over $x$ forces equality throughout, so the $\rho_{\ell}$ with $w_{\ell}>0$ coincide with a common $\bar\rho$, and $\rho_{k}\le\bar\rho$ with equal sums gives $\rho_{k}=\bar\rho$, a contradiction.
By Proposition~\ref{prop:compactification}, $\Psi_{\boldsymbol{\rho}}$ is continuous on the compact space $\overline{U}$, so its minimum $c_{\boldsymbol{\rho}}$ is positive.

\emph{Monotonicity.}
Suppose $\Psi_{\boldsymbol{\pi}}\ge\Psi_{\boldsymbol{\pi}'}$ on $\overline{U}$, and fix $n\in\N$.
Every simplex, cone, tropical and relative-entropy atom is additive on products (Lemma~\ref{lem:atoms-are-divergences}), so $\Psi_{\boldsymbol{\pi}^{\otimes n}\otimes\boldsymbol{\rho}}-\Psi_{\boldsymbol{\pi}'^{\otimes n}}=n(\Psi_{\boldsymbol{\pi}}-\Psi_{\boldsymbol{\pi}'})+\Psi_{\boldsymbol{\rho}}\ge c_{\boldsymbol{\rho}}>0$ on $\overline{U}$.
Evaluated on $\iota(U)$, on the points $\tau_{k}(w)$ and on the points $f_{k}(e_{\ell})$, and using $\Tdiv_{c\beta}=\Tdiv_{\beta}$ for $c>0$, this gives the strict inequalities $\Dren_{\alpha}(\boldsymbol{\pi}^{\otimes n}\otimes\boldsymbol{\rho})>\Dren_{\alpha}(\boldsymbol{\pi}'^{\otimes n})$ for $\alpha\in U$, $\Tdiv_{\beta}(\boldsymbol{\pi}^{\otimes n}\otimes\boldsymbol{\rho})>\Tdiv_{\beta}(\boldsymbol{\pi}'^{\otimes n})$ for $\beta\in\Bminus\setminus\{0\}$, and the corresponding strict inequalities between the pairwise relative entropies.
These are the hypotheses of~\cite[Theorem~19]{farooq2024matrix} for the tuples $\boldsymbol{\pi}^{\otimes n}\otimes\boldsymbol{\rho}$ and $\boldsymbol{\pi}'^{\otimes n}$ of probability vectors, each with a common support, and that theorem gives $N\in\N$ and a stochastic matrix mapping $(\boldsymbol{\pi}^{\otimes n}\otimes\boldsymbol{\rho})^{\otimes N}$ componentwise to $\boldsymbol{\pi}'^{\otimes nN}$.
Joint data processing and additivity of $D$ give $nN\,D(\boldsymbol{\pi}')\le N(n\,D(\boldsymbol{\pi})+D(\boldsymbol{\rho}))$, so $D(\boldsymbol{\pi}')\le D(\boldsymbol{\pi})+D(\boldsymbol{\rho})/n$ for every $n$, and $D(\boldsymbol{\pi}')\le D(\boldsymbol{\pi})$ because $D(\boldsymbol{\rho})<\infty$.

\emph{The Lipschitz bound.}
Let $\delta:=\sup_{\overline{U}}|\Psi_{\boldsymbol{\pi}}-\Psi_{\boldsymbol{\pi}'}|$, which is finite, and let $a,b\in\N$ with $a/b\ge\delta/c_{\boldsymbol{\rho}}$.
Then $b\Psi_{\boldsymbol{\pi}'}\le b\Psi_{\boldsymbol{\pi}}+b\delta\le b\Psi_{\boldsymbol{\pi}}+a\Psi_{\boldsymbol{\rho}}$, that is, $\Psi_{\boldsymbol{\pi}'^{\otimes b}}\le\Psi_{\boldsymbol{\pi}^{\otimes b}\otimes\boldsymbol{\rho}^{\otimes a}}$, and monotonicity with additivity gives $b\,D(\boldsymbol{\pi}')\le b\,D(\boldsymbol{\pi})+a\,D(\boldsymbol{\rho})$.
Taking the infimum of $a/b$ over such pairs, and exchanging $\boldsymbol{\pi}$ and $\boldsymbol{\pi}'$, gives the bound.
\end{proof}

\subsection{The representing measure}

\begin{proof}[Proof of Theorem~\ref{thm:correct}]
Let $\mathcal{P}$ be the class of bounded tuples on finite alphabets, fix $\boldsymbol{\rho}\in\mathcal{P}$ with pairwise distinct components, and write $c:=c_{\boldsymbol{\rho}}$, which is positive by Lemma~\ref{lem:profile-monotone}.
By Proposition~\ref{prop:compactification}, $\Psi_{\boldsymbol{\pi}}\in C(\overline{U})$ for every $\boldsymbol{\pi}\in\mathcal{P}$, and $\Psi_{\boldsymbol{\pi}\otimes\boldsymbol{\pi}'}=\Psi_{\boldsymbol{\pi}}+\Psi_{\boldsymbol{\pi}'}$ because every atom is additive on products.

\emph{A positive functional on a rational subspace.}
The set $V_{\Q}:=\{(\Psi_{\boldsymbol{A}}-\Psi_{\boldsymbol{B}})/m:\boldsymbol{A},\boldsymbol{B}\in\mathcal{P},\ m\in\N\}$ is a vector space over $\Q$, since sums and rational multiples of its elements are again of this form after tensor powers; for instance $(\Psi_{\boldsymbol{A}}-\Psi_{\boldsymbol{B}})/m+(\Psi_{\boldsymbol{A}'}-\Psi_{\boldsymbol{B}'})/m'=(\Psi_{\boldsymbol{A}^{\otimes m'}\otimes\boldsymbol{A}'^{\otimes m}}-\Psi_{\boldsymbol{B}^{\otimes m'}\otimes\boldsymbol{B}'^{\otimes m}})/(mm')$.
Define $L((\Psi_{\boldsymbol{A}}-\Psi_{\boldsymbol{B}})/m):=(D(\boldsymbol{A})-D(\boldsymbol{B}))/m$.
This is well defined: if $(\Psi_{\boldsymbol{A}}-\Psi_{\boldsymbol{B}})/m=(\Psi_{\boldsymbol{A}'}-\Psi_{\boldsymbol{B}'})/m'$, then $\Psi_{\boldsymbol{A}^{\otimes m'}\otimes\boldsymbol{B}'^{\otimes m}}=\Psi_{\boldsymbol{A}'^{\otimes m}\otimes\boldsymbol{B}^{\otimes m'}}$, and Lemma~\ref{lem:profile-monotone} in both directions with additivity gives $m'D(\boldsymbol{A})+mD(\boldsymbol{B}')=mD(\boldsymbol{A}')+m'D(\boldsymbol{B})$.
The functional $L$ is $\Q$-linear, it is positive by Lemma~\ref{lem:profile-monotone}, and $L(\Psi_{\boldsymbol{\pi}})=D(\boldsymbol{\pi})$, taking for $\boldsymbol{B}$ the one-letter tuple, whose profile and divergence vanish.

\emph{A bound by the order unit.}
For $f\in V_{\Q}$ and rational $r\ge\|f\|_{\infty}/c$, both $r\Psi_{\boldsymbol{\rho}}-f$ and $r\Psi_{\boldsymbol{\rho}}+f$ are nonnegative elements of $V_{\Q}$, so $|L(f)|\le r\,D(\boldsymbol{\rho})$, and hence $|L(f)|\le(D(\boldsymbol{\rho})/c)\|f\|_{\infty}$.

\emph{The closed real span.}
By this bound $L$ extends uniquely to a continuous functional $\bar L$ on the uniform closure $V$ of $V_{\Q}$, a real subspace of $C(\overline{U})$.
It is real-linear: if $f_{n},g_{n}\in V_{\Q}$ converge to $f,g$ and rationals $q_{n}$ converge to $\lambda$, then $f_{n}+g_{n}\to f+g$ and $q_{n}f_{n}\to\lambda f$, so $\bar L(f+g)=\lim_{n}L(f_{n}+g_{n})=\bar L(f)+\bar L(g)$ and $\bar L(\lambda f)=\lim_{n}q_{n}L(f_{n})=\lambda\bar L(f)$.
It is positive: if $f\in V$ is nonnegative and $f_{n}\in V_{\Q}$ converge to $f$, choose rationals $r_{n}\ge\|f_{n}-f\|_{\infty}/c$ with $r_{n}\to 0$; then $f_{n}+r_{n}\Psi_{\boldsymbol{\rho}}\ge f\ge 0$ lies in $V_{\Q}$, so $\bar L(f)=\lim_{n}\big(L(f_{n}+r_{n}\Psi_{\boldsymbol{\rho}})-r_{n}D(\boldsymbol{\rho})\big)\ge 0$.

\emph{Extension to $C(\overline{U})$.}
Every $g\in C(\overline{U})$ satisfies $g\le(\|g\|_{\infty}/c)\Psi_{\boldsymbol{\rho}}$, an element of $V$.
Put $p(g):=\inf\{\bar L(v):v\in V,\ v\ge g\}$.
For $v\ge g$ we have $v+(\|g\|_{\infty}/c)\Psi_{\boldsymbol{\rho}}\ge 0$, so $\bar L(v)\ge-(\|g\|_{\infty}/c)D(\boldsymbol{\rho})$ and $p$ is finite; it is sublinear, and $p=\bar L$ on $V$ by positivity of $\bar L$.
The Hahn--Banach theorem gives a linear functional $\tilde L\le p$ on $C(\overline{U})$ extending $\bar L$, and $\tilde L$ is positive, since $g\ge 0$ gives $\tilde L(-g)\le p(-g)\le\bar L(0)=0$.

\emph{Riesz--Markov.}
On the compact metrizable space $\overline{U}$ the Riesz--Markov theorem gives a finite regular Borel measure $\mu$ with $\tilde L(g)=\int g\,d\mu$ for every $g\in C(\overline{U})$.
Hence $D(\boldsymbol{\pi})=L(\Psi_{\boldsymbol{\pi}})=\int_{\overline{U}}\Psi_{\boldsymbol{\pi}}\,d\mu$ for every $\boldsymbol{\pi}\in\mathcal{P}$, and $c\,\mu(\overline{U})\le\int\Psi_{\boldsymbol{\rho}}\,d\mu=D(\boldsymbol{\rho})$.

\emph{The three components.}
The partition $\overline{U}=\iota(U)\sqcup\bigsqcup_{k}F_{k}\sqcup\bigsqcup_{k}T_{k}$ is Borel.
Let $m^{\Dren}$ be the image of $\mu|_{\iota(U)}$ under $\iota^{-1}$, let $m^{\Tdiv}$ be the image of $\sum_{k}\mu|_{T_{k}}$ under the map $\tau_{k}(w)\mapsto\beta^{k,w}$, and set $c_{k\ell}:=\int_{F_{k}}w_{\ell}\,d\mu(f_{k}(w))$.
Since $\Psi_{\boldsymbol{\pi}}(f_{k}(w))=\sum_{\ell\ne k}w_{\ell}\KLdiv(\pi_{k}\|\pi_{\ell})$ is linear in $w$, $\int_{F_{k}}\Psi_{\boldsymbol{\pi}}\,d\mu=\sum_{\ell\ne k}c_{k\ell}\KLdiv(\pi_{k}\|\pi_{\ell})$, and~\eqref{eq:correct} holds on $\mathcal{P}$ with finite measures of total mass $\mu(\overline{U})\le D(\boldsymbol{\rho})/c_{\boldsymbol{\rho}}$.
By Proposition~\ref{prop:measurable}, the same triple represents $D$ on bounded tuples on every measurable space, in particular on the Polish spaces of Section~\ref{ssec:notation}.

\emph{Symmetry.}
For $\sigma\in\Sym$ write $(\sigma\cdot\boldsymbol{\pi})_{k}:=\pi_{\sigma^{-1}(k)}$ and $(\sigma\cdot\alpha)_{k}:=\alpha_{\sigma^{-1}(k)}$, and likewise for $\beta$.
If $D$ is symmetric, then for every $\sigma\in\Sym$, applying~\eqref{eq:correct} to $\sigma\cdot\boldsymbol{\pi}$ and using $\Dren_{\alpha}(\sigma\cdot\boldsymbol{\pi})=\Dren_{\sigma^{-1}\cdot\alpha}(\boldsymbol{\pi})$ and the analogous identities for the tropical and relative-entropy atoms gives a second representation of $D(\boldsymbol{\pi})$, against the images of $m^{\Dren}$ and $m^{\Tdiv}$ under $\alpha\mapsto\sigma^{-1}\cdot\alpha$ and $\beta\mapsto\sigma^{-1}\cdot\beta$ and the permuted coefficients $(c_{\sigma(k)\sigma(\ell)})$.
Averaging the $W!$ representations obtained this way, which is allowed because~\eqref{eq:correct} is linear in the triple, replaces each component by its orbit average: $\frac{1}{W!}\sum_{\sigma}\sigma_{*}m^{\Dren}$ and $\frac{1}{W!}\sum_{\sigma}\sigma_{*}m^{\Tdiv}$ are $\Sym$-invariant, and $\frac{1}{W!}\sum_{\sigma}c_{\sigma(k)\sigma(\ell)}$ does not depend on the pair $k\ne\ell$, since any pair is sent to any other by a permutation.
\end{proof}

\subsection{Uniqueness and necessity}

\begin{proof}[Proof of Proposition~\ref{prop:uniqueness}]
Let $\sigma_{\Dren}:=m^{\Dren}-m'^{\Dren}$, a finite signed measure on $U$; let $\sigma_{\Tdiv}$ be the difference of the images of $m^{\Tdiv}$ and $m'^{\Tdiv}$ under $\beta\mapsto\beta/\beta_{\star}$, a finite signed measure on the compact set $B_{1}:=\{\beta\in\Bminus:\beta_{\star}=1\}$; and let $\sigma_{k\ell}:=c_{k\ell}-c'_{k\ell}$.
Since $\Tdiv_{c\beta}=\Tdiv_{\beta}$ for $c>0$, the hypothesis reads
\begin{equation}\label{eq:uniq-zero}
\int_{U}\Dren_{\alpha}(\boldsymbol{\pi})\,d\sigma_{\Dren}(\alpha)+\int_{B_{1}}\Tdiv_{\beta}(\boldsymbol{\pi})\,d\sigma_{\Tdiv}(\beta)+\sum_{k\ne\ell}\sigma_{k\ell}\KLdiv(\pi_{k}\|\pi_{\ell})=0
\end{equation}
for every bounded tuple $\boldsymbol{\pi}$ on a finite alphabet.

\emph{Every atom is at most the largest log-likelihood ratio.}
For a bounded tuple $\boldsymbol{\pi}$ on a finite alphabet, $0\le\Phi_{\xi}(\boldsymbol{\pi})\le\Lambda(\boldsymbol{\pi})$ for every atom.
For the relative-entropy and tropical atoms this follows from their definitions, the tropical atom at $\beta^{k,w}$ being $\max_{x}\sum_{\ell\ne k}w_{\ell}\log(\pi_{k}(x)/\pi_{\ell}(x))$.
For $\alpha=e_{k}+t\beta^{k,w}$ in a cone, estimate (b) in the proof of Proposition~\ref{prop:compactification} gives $\Dren_{\alpha}\le\Tdiv_{\beta^{k,w}}$.
For $\alpha\in\Aplus\setminus E$, choose $k$ with $\alpha_{k}=\alpha_{\star}$, and put $s:=s_{k}(\alpha)=1-\alpha_{\star}$ and $w:=w^{(k)}(\alpha)$, so that $\alpha=e_{k}+s(w-e_{k})$.
The function $\lambda\mapsto\log H_{e_{k}+\lambda(w-e_{k})}(\boldsymbol{\pi})=\log\E_{\pi_{k}}[\exp(\lambda\sum_{\ell\ne k}w_{\ell}\log(\pi_{\ell}/\pi_{k}))]$ is convex, vanishes at $0$ and has derivative $-\Delta^{(k)}_{w}(\boldsymbol{\pi})$ there, so $\log H_{\alpha}(\boldsymbol{\pi})\ge-s\Delta^{(k)}_{w}(\boldsymbol{\pi})$ and $\Dren_{\alpha}(\boldsymbol{\pi})=\log H_{\alpha}(\boldsymbol{\pi})/(-s)\le\Delta^{(k)}_{w}(\boldsymbol{\pi})\le\Lambda(\boldsymbol{\pi})$.
Nonnegativity is part of Lemma~\ref{lem:atoms-are-divergences}.

\emph{Tilted test tuples.}
For a finite list $V=(v_{1},\dots,v_{n})$ of points of $\R^{W}$ and $\theta>0$, let $\boldsymbol{\pi}^{(\theta)}$ be the tuple on $\{1,\dots,n\}$ with $\pi^{(\theta)}_{k}(i):=e^{\theta v_{ik}}/Z_{k}(\theta)$, where $Z_{k}(\theta):=\sum_{i}e^{\theta v_{ik}}$, and put $\varphi_{V}(u):=\max_{i}\langle v_{i},u\rangle$ for $u\in\R^{W}$.
Since $\log\sum_{i}e^{\theta\langle v_{i},u\rangle}$ lies between $\theta\varphi_{V}(u)$ and $\theta\varphi_{V}(u)+\log n$,
\[
\frac{\Dren_{\alpha}(\boldsymbol{\pi}^{(\theta)})}{\theta}\to\psi_{V}(\alpha):=\frac{\varphi_{V}(\alpha)-\sum_{k}\alpha_{k}\varphi_{V}(e_{k})}{\alpha_{\star}-1},
\qquad
\frac{\Tdiv_{\beta}(\boldsymbol{\pi}^{(\theta)})}{\theta}\to\varphi_{V}(\beta)-\textstyle\sum_{k}\beta_{k}\varphi_{V}(e_{k})
\]
for $\alpha\in U$ and $\beta\in B_{1}$.
As $\theta\to\infty$, $\pi^{(\theta)}_{k}$ converges to the uniform distribution on $I_{k}:=\arg\max_{i}v_{ik}$, and $\KLdiv(\pi^{(\theta)}_{k}\|\pi^{(\theta)}_{\ell})=\theta\sum_{i}\pi^{(\theta)}_{k}(i)(v_{ik}-v_{i\ell})+\log Z_{\ell}(\theta)-\log Z_{k}(\theta)$, so
\[
\frac{\KLdiv(\pi^{(\theta)}_{k}\|\pi^{(\theta)}_{\ell})}{\theta}\to\kappa_{V}(k,\ell):=\frac{1}{|I_{k}|}\sum_{i\in I_{k}}(v_{ik}-v_{i\ell})+\varphi_{V}(e_{\ell})-\varphi_{V}(e_{k})
\]
Every log-likelihood ratio of $\boldsymbol{\pi}^{(\theta)}$ is at most $4\theta\max_{i,k}|v_{ik}|+\log n$, so the atoms divided by $\theta$ are bounded uniformly in the atom and in $\theta\ge 1$.
Dividing~\eqref{eq:uniq-zero} by $\theta$ and applying dominated convergence for the finite measures $|\sigma_{\Dren}|$ and $|\sigma_{\Tdiv}|$ gives, for every finite list $V$,
\begin{equation}\label{eq:uniq-limit}
\int_{U}\psi_{V}\,d\sigma_{\Dren}+\int_{B_{1}}\Big(\varphi_{V}(\beta)-\sum_{k}\beta_{k}\varphi_{V}(e_{k})\Big)\,d\sigma_{\Tdiv}(\beta)+\sum_{k\ne\ell}\sigma_{k\ell}\,\kappa_{V}(k,\ell)=0
\end{equation}

\emph{Pairs of lists.}
Suppose the lists $V_{1}$ and $V_{2}$ satisfy $\varphi_{V_{1}}(e_{k})=\varphi_{V_{2}}(e_{k})$ for every $k$, and that for every $k$ the points of $V_{1}$ attaining $\max_{i}v_{ik}$ form the same multiset as those of $V_{2}$.
Then $\kappa_{V_{1}}=\kappa_{V_{2}}$, and subtracting~\eqref{eq:uniq-limit} for $V_{2}$ from the identity for $V_{1}$ gives, with $h:=\varphi_{V_{1}}-\varphi_{V_{2}}$,
\begin{equation}\label{eq:uniq-pair}
\int_{U}\frac{h(\alpha)}{\alpha_{\star}-1}\,d\sigma_{\Dren}(\alpha)+\int_{B_{1}}h\,d\sigma_{\Tdiv}=0
\end{equation}

\emph{The simplex and cone component.}
Let $\alpha_{0}\in\mathcal{A}$ and $0<r<\min_{k}|e_{k}-\alpha_{0}|_{1}$.
For each sign vector $\varepsilon\in\{-1,1\}^{W}$ put $u_{\varepsilon}:=\varepsilon-(\langle\varepsilon,\alpha_{0}\rangle+r)\one$, let $V_{2}$ be the list of the $u_{\varepsilon}$, and let $V_{1}$ be $V_{2}$ followed by the origin.
On $\mathcal{A}$, where $\langle\one,\alpha\rangle=1$, $\varphi_{V_{2}}(\alpha)=\max_{\varepsilon}\langle\varepsilon,\alpha-\alpha_{0}\rangle-r=|\alpha-\alpha_{0}|_{1}-r$, so $h(\alpha)=\max(0,r-|\alpha-\alpha_{0}|_{1})=:t_{\alpha_{0},r}(\alpha)$.
Since $\varphi_{V_{2}}(e_{k})=|e_{k}-\alpha_{0}|_{1}-r>0$, the origin attains no coordinate maximum of $V_{1}$, so the pair satisfies the conditions above.
On $B_{1}$, where $\langle\one,\beta\rangle=0$, $\varphi_{V_{2}}(\beta)=|\beta|_{1}>0$ and $h=0$.
So~\eqref{eq:uniq-pair} gives $\int_{U}t_{\alpha_{0},r}(\alpha)(\alpha_{\star}-1)^{-1}\,d\sigma_{\Dren}(\alpha)=0$ for every $\alpha_{0}\in\mathcal{A}$ and every $r\in(0,\min_{k}|e_{k}-\alpha_{0}|_{1})$.
Let $g$ be continuous with compact support in $U$.
The set $U$ is closed in $\mathcal{A}\setminus E$, so by the Tietze extension theorem and a cutoff $g$ extends to a continuous function on $\mathcal{A}$, still denoted $g$, whose support $K$ is compact and disjoint from $E$.
Choose $r_{0}>0$ smaller than the $|\cdot|_{1}$-distance from $K$ to $E$, let $K'$ be the closed $r_{0}$-neighborhood of $K$, and for $0<r<r_{0}$ let $g_{r}(\alpha):=\int_{\mathcal{A}}g(\alpha_{0})\,t_{\alpha_{0},r}(\alpha)\,d\lambda(\alpha_{0})\big/\int_{\mathcal{A}}t_{\alpha_{0},r}\,d\lambda(\alpha_{0})$, where $\lambda$ is Lebesgue measure on the plane $\mathcal{A}$.
The functions $g_{r}$ are supported in $K'$ and converge uniformly to $g$ as $r\to 0$, because $g$ is uniformly continuous and $t_{\alpha_{0},r}(\alpha)$ vanishes unless $|\alpha-\alpha_{0}|_{1}<r$.
The function $(\alpha_{\star}-1)^{-1}$ is bounded on the compact set $K'\cap(\Aplus\cup\Aminus)$, which avoids $E$, so Fubini's theorem gives $\int_{U}g_{r}(\alpha)(\alpha_{\star}-1)^{-1}\,d\sigma_{\Dren}(\alpha)=0$, and letting $r\to 0$ gives $\int_{U}g(\alpha)(\alpha_{\star}-1)^{-1}\,d\sigma_{\Dren}(\alpha)=0$.
Applied to $g\cdot(\alpha_{\star}-1)$, which is again continuous with compact support in $U$, this gives $\int_{U}g\,d\sigma_{\Dren}=0$ for every such $g$.
A finite signed Borel measure on the locally compact separable metric space $U$ is determined by its integrals of continuous functions with compact support, so $\sigma_{\Dren}=0$.

\emph{The tropical component.}
Let $h_{1},h_{2}$ be maxima of finitely many linear functions on $\R^{W}$, let $h:=h_{1}-h_{2}$, and choose $R>\max_{k}|h(e_{k})|$.
Put $\varphi_{1}:=\max(h_{1}-R\langle\one,\cdot\rangle,h_{2})+h_{1}$ and $\varphi_{2}:=\max(h_{2}-R\langle\one,\cdot\rangle,h_{1})+h_{2}$.
A sum of two maxima of linear functions is the maximum of the pairwise sums, so $\varphi_{j}=\varphi_{V_{j}}$ for the list $V_{j}$ of those pairwise sums of gradients.
On $\mathcal{A}$, $\varphi_{1}-\varphi_{2}=\max(h-R,0)-\max(-h-R,0)$, which vanishes at every vertex.
At $e_{k}$ the first maximum in $\varphi_{1}$ is attained only by linear pieces of $h_{2}$, because $h_{1}(e_{k})-R<h_{2}(e_{k})$, and the first maximum in $\varphi_{2}$ only by pieces of $h_{1}$.
So the points of $V_{1}$ and of $V_{2}$ attaining the $k$-th coordinate maximum form the same multiset, of sums of a maximizing piece of $h_{1}$ and a maximizing piece of $h_{2}$, and the pair satisfies the conditions above.
On $B_{1}$, $\langle\one,\beta\rangle=0$ and $\varphi_{1}-\varphi_{2}=h$.
With $\sigma_{\Dren}=0$,~\eqref{eq:uniq-pair} gives $\int_{B_{1}}h\,d\sigma_{\Tdiv}=0$.
The restrictions to $B_{1}$ of such differences $h$ form a vector lattice of continuous functions, since $\max(h_{1}-h_{2},h'_{1}-h'_{2})=\max(h_{1}+h'_{2},h'_{1}+h_{2})-(h_{2}+h'_{2})$.
It contains the constant function $1=\max_{k}\beta_{k}$ on $B_{1}$ and separates points, since it contains the coordinate functions.
By the lattice form of the Stone--Weierstrass theorem, it is dense in $C(B_{1})$, and a finite signed measure on the compact metric space $B_{1}$ that annihilates a dense subset of $C(B_{1})$ vanishes, so $\sigma_{\Tdiv}=0$.

\emph{The relative-entropy coefficients.}
With $\sigma_{\Dren}=\sigma_{\Tdiv}=0$,~\eqref{eq:uniq-limit} reads $\sum_{k\ne\ell}\sigma_{k\ell}\,\kappa_{V}(k,\ell)=0$.
Take the list $v_{k}:=Me_{k}+u_{k}$, $k=1,\dots,W$, with $u_{1},\dots,u_{W}\in\R^{W}$ and $M>2\max_{i,j}|u_{ij}|$, so that $I_{k}=\{k\}$ and $\varphi_{V}(e_{k})=M+u_{kk}$.
Then $\kappa_{V}(k,\ell)=M+u_{\ell\ell}-u_{k\ell}$, and for $k\ne\ell$ the entry $u_{k\ell}$ occurs in no other term; since the identity holds on an open set of choices of the $u_{k}$, every $\sigma_{k\ell}$ vanishes.
\end{proof}

\begin{proof}[Proof of Corollary~\ref{cor:necessity}]
Each atom $\Phi_{\xi}$ is a $W$-way DPI--additive divergence by Lemma~\ref{lem:atoms-are-divergences}, and the unit mass at $\xi$, placed on a normalized direction when $\xi$ is tropical, represents it; by Proposition~\ref{prop:uniqueness} every triple of finite measures representing $\Phi_{\xi}$ is that unit mass, up to moving tropical mass along rays.
In particular $\Phi_{\xi}$ does not vanish identically, since the zero triple would otherwise represent it.
For $\sigma\in\Sym$ write $(\sigma\cdot\alpha)_{k}:=\alpha_{\sigma^{-1}(k)}$, and likewise for $\beta$.
The divergences $\sum_{\sigma}\Dren_{\sigma\cdot\alpha}$ for $\alpha\in\Aplus\setminus E$, the same sum for $\alpha\in\Aminus\setminus E$, $\sum_{\sigma}\Tdiv_{\sigma\cdot\beta}$ for $\beta\in\Bminus\setminus\{0\}$, and $\sum_{k\ne\ell}\KLdiv(\pi_{k}\|\pi_{\ell})$ are symmetric by permutation equivariance of the atoms.
Each of them is represented by a sum of unit masses on a single stratum, since permuting coordinates maps each stratum to itself.
By uniqueness every representation of each of them assigns positive mass to that stratum.
For the extreme rays: if $\Phi_{\xi}=D_{1}+D_{2}$ with $D_{1},D_{2}$ $W$-way DPI--additive divergences, their representing triples add to one representing $\Phi_{\xi}$, so by uniqueness both are multiples of the unit mass at $\xi$ and $D_{1},D_{2}$ are multiples of $\Phi_{\xi}$.
Conversely, regard the triple representing a divergence $D$, with its tropical part normalized, as one finite measure $\mu$ on the separable metric space $X:=U\sqcup B_{1}\sqcup\{(k,\ell):k\ne\ell\}$, and suppose $\mu$ is not a multiple of a unit mass.
If every Borel set had measure $0$ or $\mu(X)$, then for each $n$ one of countably many closed balls of radius $1/n$ covering $X$ would have full measure, and the intersection of these balls would be a single point of full measure.
So there is a Borel set $A$ with $\mu(A)>0$ and $\mu(A^{c})>0$, and $D$ splits into the divergences represented by $\mu|_{A}$ and $\mu|_{A^{c}}$, neither of which is a multiple of $D$ by uniqueness, so $D$ does not span an extreme ray.
\end{proof}

\subsection{Unbounded tuples and general sample spaces}

\begin{proof}[Proof of Proposition~\ref{prop:unbounded}]
For a tuple with a common support and $\alpha\in\Aminus$, $H_{\alpha}\ge 1$ by the H\"older argument in the proof of Lemma~\ref{lem:profile-monotone}, so $\Dren_{\alpha}\in[0,\infty]$, and $H_{\alpha}(K\boldsymbol{\pi})\le H_{\alpha}(\boldsymbol{\pi})$ for every Markov kernel $K$ (Lemma~\ref{lem:atoms-are-divergences}).
Write $a_{n}(\boldsymbol{\pi}):=\Dren_{e_{k}+n\beta}(\boldsymbol{\pi})/n$.
Additivity of each atom gives $a_{n}(\boldsymbol{\pi}\otimes\boldsymbol{\pi}')=a_{n}(\boldsymbol{\pi})+a_{n}(\boldsymbol{\pi}')$, and data processing gives $a_{n}(K\boldsymbol{\pi})\le a_{n}(\boldsymbol{\pi})$.
A Markov kernel preserves mutual absolute continuity, since $K\pi_{\ell}(A)=0$ forces $K(\cdot,A)=0$ to hold $\pi_{\ell}$-almost everywhere, hence $\pi_{k}$-almost everywhere, and then $K\pi_{k}(A)=0$.
So $\mathcal{U}_{\beta}$ is closed under products and Markov kernels, and linearity and positivity of the fixed Banach limit make $\Phi$ additive on products and monotone under joint data processing.
For a bounded tuple, $g:=\prod_{m}\pi_{m}^{\beta_{m}}$ is essentially bounded, and $\Dren_{e_{k}+n\beta}(\boldsymbol{\pi})=\frac1n\log\int\pi_{k}g^{n}\,d\nu\le\Tdiv_{\beta}(\boldsymbol{\pi})$, so $0\le a_{n}(\boldsymbol{\pi})\le\Tdiv_{\beta}(\boldsymbol{\pi})/n\to 0$; the tuple lies in $\mathcal{U}_{\beta}$, and $\Phi(\boldsymbol{\pi})=0$ because the Banach limit of a convergent sequence is its limit.
For the Gaussian location tuple with common variance $\sigma^{2}$ and $\alpha$ in the affine slice, completing the square in $\sum_{j}\alpha_{j}(x-\mu_{j})^{2}$ gives $\log H_{\alpha}=\big((\sum_{j}\alpha_{j}\mu_{j})^{2}-\sum_{j}\alpha_{j}\mu_{j}^{2}\big)/(2\sigma^{2})$.
At $\alpha=e_{k}+n\beta$, with $b:=\sum_{j}\beta_{j}\mu_{j}$ and $\sum_{j}\beta_{j}=0$, this equals $\big(n^{2}b^{2}-n\sum_{j\ne k}\beta_{j}(\mu_{j}-\mu_{k})^{2}\big)/(2\sigma^{2})$, and $\alpha_{\star}-1=n$, so
\[
a_{n}(\boldsymbol{\pi})=\frac{1}{2\sigma^{2}}\Big(b^{2}-\frac1n\sum_{j\ne k}\beta_{j}(\mu_{j}-\mu_{k})^{2}\Big)\to\frac{b^{2}}{2\sigma^{2}}
\]
The sequence is bounded, so the tuple lies in $\mathcal{U}_{\beta}$ and $\Phi(\boldsymbol{\pi})=b^{2}/(2\sigma^{2})$.
Finally, suppose a triple of finite measures represented $\Phi$ through~\eqref{eq:correct} on $\mathcal{U}_{\beta}$.
On bounded tuples on finite alphabets it would represent $\Phi=0$, which the zero triple also represents, so by Proposition~\ref{prop:uniqueness} its simplex, cone and relative-entropy parts vanish and its tropical part vanishes after normalization.
The right-hand side of~\eqref{eq:correct} would then vanish on every tuple, while $\Phi=b^{2}/(2\sigma^{2})>0$ on a Gaussian location tuple with $b\ne 0$.
\end{proof}

\begin{proof}[Proof of Proposition~\ref{prop:measurable}]
Let $\boldsymbol{\pi}$ be a bounded tuple on a measurable space with dominating measure $\nu$, with $|\log(\pi_{k}/\pi_{\ell})|\le B$ on the common support.
Let $\bar\nu:=\frac1W\sum_{\ell}\pi_{\ell}\nu$, a probability measure, and let $S(x):=(\pi_{k}(x)/\sum_{\ell}\pi_{\ell}(x))_{k}$, which is the density of $(\pi_{k}\nu)_{k}$ with respect to $W\bar\nu$ and takes values in the compact set $K_{B}:=\{y\in\simplex:y_{k}\le e^{B}y_{\ell}\ \text{for all}\ k,\ell\}$ inside the open simplex.
With $\sigma:=S_{*}\bar\nu$ and $\sum_{k}\alpha_{k}=1$, every atom of $\boldsymbol{\pi}$ is a functional of $\sigma$: $H_{\alpha}(\boldsymbol{\pi})=W\int\prod_{k}y_{k}^{\alpha_{k}}\,d\sigma(y)$, $\KLdiv(\pi_{k}\|\pi_{\ell})=W\int y_{k}\log(y_{k}/y_{\ell})\,d\sigma(y)$, and $\Tdiv_{\beta}(\boldsymbol{\pi})=\beta_{\star}^{-1}\log\max_{y\in\supp\sigma}\prod_{k}y_{k}^{\beta_{k}}$, the essential supremum.
Fix $B'>B$, let $N:=K_{B'}$, a convex polytope inside the open simplex containing $K_{B}$, and for $h>0$ fix a triangulation of $N$ with mesh at most $h$, vertex set $\mathcal{V}_{h}$ and barycentric coordinate functions $\lambda_{v}\ge 0$, so that $\sum_{v}\lambda_{v}=1$ and $\sum_{v}\lambda_{v}(y)\,v=y$.

\emph{A coarser tuple.}
Assign each $y\in N$ measurably to one simplex of the triangulation containing it, so that the points assigned to the same simplex form Borel cells partitioning $N$, and let $\boldsymbol{P}_{h}$ be the tuple on the cells $C$ with $\sigma(C)>0$ given by $P_{h}(C)_{k}:=(\pi_{k}\nu)(S^{-1}(C))=W\int_{C}y_{k}\,d\sigma(y)$.
It is the image of $\boldsymbol{\pi}$ under a deterministic kernel, so $\Phi_{\xi}(\boldsymbol{P}_{h})\le\Phi_{\xi}(\boldsymbol{\pi})$ for every atom and $D(\boldsymbol{P}_{h})\le D(\boldsymbol{\pi})$.
With $\bar y_{C}$ the $\sigma$-average of $C$, $P_{h}(C)=W\sigma(C)\,\bar y_{C}$, and $\bar y_{C}\in N$ lies within $h$ of every point of $C$.

\emph{A finer tuple.}
Let $\boldsymbol{Q}_{h}$ be the tuple on the vertices $v$ with $m_{v}:=\int\lambda_{v}\,d\sigma>0$ given by $Q_{h}(v)_{k}:=Wv_{k}m_{v}$.
Its components are probability vectors, since $\sum_{v}Wv_{k}m_{v}=W\int y_{k}\,d\sigma=1$, with a common support and log-likelihood ratios bounded by $B'$.
The kernel $K(v,dx):=\lambda_{v}(S(x))\,\bar\nu(dx)/m_{v}$ satisfies $\sum_{v}Q_{h}(v)_{k}K(v,dx)=WS_{k}(x)\,\bar\nu(dx)=\pi_{k}(x)\,\nu(dx)$, so $\boldsymbol{\pi}$ is the image of $\boldsymbol{Q}_{h}$, $\Phi_{\xi}(\boldsymbol{\pi})\le\Phi_{\xi}(\boldsymbol{Q}_{h})$ for every atom, and $D(\boldsymbol{\pi})\le D(\boldsymbol{Q}_{h})$.

\emph{Convergence of the atoms.}
In these coordinates $H_{\alpha}(\boldsymbol{P}_{h})=W\int\prod_{k}\bar y(y)_{k}^{\alpha_{k}}\,d\sigma(y)$, where $\bar y(y)$ is the $\sigma$-average of the cell containing $y$, and $H_{\alpha}(\boldsymbol{Q}_{h})=W\int\sum_{v}\lambda_{v}(y)\prod_{k}v_{k}^{\alpha_{k}}\,d\sigma(y)$; both converge to $H_{\alpha}(\boldsymbol{\pi})$ as $h\to 0$ by uniform continuity of $y\mapsto\prod_{k}y_{k}^{\alpha_{k}}$ on $N$, where it is bounded below by a positive constant, and the relative entropies converge in the same way.
A point $\bar y_{C}$ with $\sigma(C)>0$ lies within $h$ of $C\cap\supp\sigma$, and a vertex $v$ with $m_{v}>0$ lies within $h$ of a point of $\supp\sigma$ at which $\lambda_{v}>0$.
Conversely, for $y\in\supp\sigma$ the ball of radius $h$ around $y$ has positive $\sigma$-mass, so it meets a cell $C$ with $\sigma(C)>0$ in positive mass, and then $|\bar y_{C}-y|\le 2h$.
Also some vertex $v$ has $\int_{|y'-y|<h}\lambda_{v}(y')\,d\sigma(y')>0$, so that $m_{v}>0$ and $|v-y|\le 2h$.
The tropical atoms of $\boldsymbol{P}_{h}$ and $\boldsymbol{Q}_{h}$ are maxima of the continuous function $y\mapsto\beta_{\star}^{-1}\log\prod_{k}y_{k}^{\beta_{k}}$ over these points, so they converge to the tropical atoms of $\boldsymbol{\pi}$.
All atoms of $\boldsymbol{P}_{h}$, $\boldsymbol{Q}_{h}$ and $\boldsymbol{\pi}$ lie in $[0,B']$.

\emph{Conclusion.}
The map $\xi\mapsto\Phi_{\xi}(\boldsymbol{\pi})$ is Borel, as the pointwise limit along $h=1/n$ of the continuous maps $\xi\mapsto\Phi_{\xi}(\boldsymbol{P}_{h})$.
By hypothesis $D(\boldsymbol{P}_{h})$ and $D(\boldsymbol{Q}_{h})$ are the integrals of their atoms against the finite triple, and dominated convergence sends both to the integral of the atoms of $\boldsymbol{\pi}$ against it.
Since $D(\boldsymbol{P}_{h})\le D(\boldsymbol{\pi})\le D(\boldsymbol{Q}_{h})$, $D(\boldsymbol{\pi})$ equals that integral.
\end{proof}

\section{The boundary limits}\label{app:limits}

The boundary limits are Equations~\eqref{eq:KL-as-limit} and~\eqref{eq:trop-as-limit} of the main text.
Throughout, $\boldsymbol{\pi}$ is a bounded tuple on $(\mathcal{X},\nu)$, and integrals and essential suprema are taken over the common support of its components.

\begin{lemma}[Relative entropy as a vertex derivation]\label{lem:vertex-derivation}
Let $k\ne\ell$, let $Y:=\log(\pi_{\ell}/\pi_{k})$ with $|Y|\le B$ on the common support, and let $\alpha(\epsilon):=(1-\epsilon)e_{k}+\epsilon e_{\ell}$.
For $\epsilon\in(0,1]$,
\begin{equation}\label{eq:KL-derivation-explicit}
\Cdiv_{\alpha(\epsilon)}(\boldsymbol{\pi})=\epsilon\,\KLdiv(\pi_{k}\|\pi_{\ell})-\frac{\epsilon^{2}}{2}\Var_{\pi_{k}}(Y)+R(\epsilon),\qquad|R(\epsilon)|\le\tfrac43\,\epsilon^{3}B^{3}
\end{equation}
\end{lemma}

In particular $\epsilon^{-1}\Cdiv_{\alpha(\epsilon)}(\boldsymbol{\pi})\to\KLdiv(\pi_{k}\|\pi_{\ell})$, which is~\eqref{eq:KL-as-limit}.

\begin{proof}[Proof of Lemma~\ref{lem:vertex-derivation}]
Let $K(\theta):=\log\E_{\pi_{k}}[e^{\theta Y}]$, so that $\Cdiv_{\alpha(\epsilon)}(\boldsymbol{\pi})=-\log\int\pi_{k}^{1-\epsilon}\pi_{\ell}^{\epsilon}\,d\nu=-K(\epsilon)$.
Since $Y$ is bounded, $K$ is smooth, with $K(0)=0$, $K'(0)=\E_{\pi_{k}}[Y]=-\KLdiv(\pi_{k}\|\pi_{\ell})$ and $K''(0)=\Var_{\pi_{k}}(Y)$.
Its third derivative at $\theta$ is the third central moment of $Y$ under the tilted law with density proportional to $e^{\theta Y}$ with respect to $\pi_{k}$, and both $Y$ and its tilted mean lie in $[-B,B]$, so $|K'''(\theta)|\le(2B)^{3}$.
Taylor's theorem with the Lagrange remainder gives $K(\epsilon)=-\epsilon\,\KLdiv(\pi_{k}\|\pi_{\ell})+\tfrac{\epsilon^{2}}{2}\Var_{\pi_{k}}(Y)+\tfrac{\epsilon^{3}}{6}K'''(\theta)$ for some $\theta\in(0,\epsilon)$, and~\eqref{eq:KL-derivation-explicit} follows with $R(\epsilon)=-\tfrac{\epsilon^{3}}{6}K'''(\theta)$.
\end{proof}

\begin{lemma}[The tropical limit]\label{lem:tropical-limit}
Let $\beta\in\Bminus\setminus\{0\}$ with $\beta_{k}=\beta_{\star}$, let $g:=\prod_{m}\pi_{m}^{\beta_{m}}$, and let $\alpha(t):=e_{k}+t\beta$ for $t>0$.
Then
\[
\Dren_{\alpha(t)}(\boldsymbol{\pi})=\frac{1}{\beta_{\star}}\log\|g\|_{L^{t}(\pi_{k}\nu)}=-\frac{1}{t\beta_{\star}}\,\Cdiv_{\alpha(t)}(\boldsymbol{\pi})
\]
which is nondecreasing in $t$ and converges to $\Tdiv_{\beta}(\boldsymbol{\pi})$ as $t\to\infty$.
\end{lemma}

This is~\eqref{eq:trop-as-limit}, and for a finite alphabet with counting measure estimate (b) in the proof of Proposition~\ref{prop:compactification} gives the rate $t^{-1}\log(1/\min_{x}\pi_{k}(x))$ when $\beta_{\star}=1$.

\begin{proof}[Proof of Lemma~\ref{lem:tropical-limit}]
The coordinate $\alpha(t)_{k}=1+t\beta_{\star}$ is the largest, since the other coordinates $t\beta_{m}$ are nonpositive, so $\alpha(t)_{\star}-1=t\beta_{\star}$.
Since $\alpha(t)=e_{k}+t\beta$, $\prod_{m}\pi_{m}^{\alpha(t)_{m}}=\pi_{k}g^{t}$, so $H_{\alpha(t)}(\boldsymbol{\pi})=\int g^{t}\,d(\pi_{k}\nu)=\|g\|_{L^{t}(\pi_{k}\nu)}^{t}$, which gives both identities.
On the common support $\pi_{k}>0$, so $\pi_{k}\nu$ and $\nu$ have the same null sets there, and $g=\prod_{m}(\pi_{m}/\pi_{k})^{\beta_{m}}$ is essentially bounded above and below by positive constants because $\boldsymbol{\pi}$ is bounded and $\sum_{m}\beta_{m}=0$.
For the probability measure $\pi_{k}\nu$, Jensen's inequality makes $t\mapsto\|g\|_{L^{t}}$ nondecreasing, and $\|g\|_{L^{t}}\le\|g\|_{L^{\infty}}$.
For $M<\|g\|_{L^{\infty}}$ the set $\{g>M\}$ has positive $\pi_{k}\nu$-measure $\delta$, so $\|g\|_{L^{t}}\ge M\delta^{1/t}\to M$.
Hence $\|g\|_{L^{t}}\to\|g\|_{L^{\infty}}$, the $\nu$-essential supremum of $g$ on the common support, and $\Dren_{\alpha(t)}(\boldsymbol{\pi})\to\beta_{\star}^{-1}\log\|g\|_{L^{\infty}}=\Tdiv_{\beta}(\boldsymbol{\pi})$.
\end{proof}

\section{The many-state conjecture of~\cite{mu2021blackwell}}\label{app:section-k}

For $W$ states,~\cite[Appendix~K]{mu2021blackwell} compares two bounded experiments $\boldsymbol{P}$ and $\boldsymbol{Q}$ through the moment generating function of the log-likelihood-ratio vector, $M^{(j)}_{\boldsymbol{P}}(t):=\E_{P_{j}}\big[\exp\sum_{k\ne j}t_{k}\log(dP_{j}/dP_{k})\big]$ for $t\in\R^{W-1}$, and calls the pair \emph{generic} when, for every pair of states $k\ne\ell$, the maximum and the minimum of $\log(dP_{k}/dP_{\ell})$ differ from those of $\log(dQ_{k}/dQ_{\ell})$.
Proposition~7 there states that if $\boldsymbol{P}$ dominates $\boldsymbol{Q}$ in large samples and the pair is generic, then for every state $j$: (i) $M^{(j)}_{\boldsymbol{P}}(t)>M^{(j)}_{\boldsymbol{Q}}(t)$ for $t\in\R_{\ge 0}^{W-1}\setminus\{0\}$; (ii) $M^{(j)}_{\boldsymbol{P}}(t)<M^{(j)}_{\boldsymbol{Q}}(t)$ for $t\in\R_{\le 0}^{W-1}\setminus\{0\}$ with $\sum_{k}t_{k}>-1$; and (iii) $\KLdiv(P_{k}\|P_{\ell})>\KLdiv(Q_{k}\|Q_{\ell})$ for all $k\ne\ell$.
Dominance in large samples is their Definition~1: there is $n_{0}$ such that $\boldsymbol{P}^{\otimes n}$ dominates $\boldsymbol{Q}^{\otimes n}$ in the Blackwell order for every $n\ge n_{0}$.
They conjecture that the conditions of Proposition~7 are also sufficient for dominance in large samples.

The moment generating function is a Hellinger transform, $M^{(j)}_{\boldsymbol{P}}(t)=H_{\alpha}(\boldsymbol{P})$ with $\alpha_{j}=1+\sum_{k}t_{k}$ and $\alpha_{k}=-t_{k}$ for $k\ne j$.
Over all $j$, condition (i) is therefore $\Dren_{\alpha}(\boldsymbol{P})>\Dren_{\alpha}(\boldsymbol{Q})$ for $\alpha\in\Aminus\setminus E$, where $\alpha_{\star}>1$, and condition (ii) is the same inequality for $\alpha\in\Aplus\setminus E$, where $\alpha_{\star}<1$ reverses the inequality between the Hellinger transforms.
These, with (iii), are the strict inequalities that~\cite[Theorem~19]{farooq2024matrix} requires of the simplex, cone and relative-entropy atoms, and that theorem proves dominance in large samples when the strict tropical inequalities $\Tdiv_{\beta}(\boldsymbol{P})>\Tdiv_{\beta}(\boldsymbol{Q})$ for every $\beta\in\Bminus\setminus\{0\}$ hold as well.
Condition (i) implies the non-strict tropical inequalities through~\eqref{eq:trop-as-limit}, and genericity makes them strict at the directions $\beta=e_{k}-e_{\ell}$ only.
The question of~\cite[Remark~20]{farooq2024matrix} is whether some pair satisfies the assumptions of~\cite{mu2021blackwell} but not those of~\cite[Theorem~19]{farooq2024matrix}.
Such pairs exist, and for them dominance in large samples can fail.

\begin{proposition}[A three-state counterexample]\label{prop:mpst-counterexample}
Let $W=3$, let $\boldsymbol{P}$ have six outcomes whose likelihood vectors $(P_{1},P_{2},P_{3})$ are the permutations of $(\tfrac14,\tfrac15,\tfrac1{20})$, and let $\boldsymbol{Q}$ have four outcomes with likelihood vectors $(\tfrac5{18},\tfrac19,\tfrac19)$, $(\tfrac19,\tfrac5{18},\tfrac19)$, $(\tfrac19,\tfrac19,\tfrac5{18})$ and $(\tfrac12,\tfrac12,\tfrac12)$.
The pair is generic, satisfies conditions (i)--(iii), and $\boldsymbol{P}^{\otimes 2}$ dominates $\boldsymbol{Q}^{\otimes 2}$, but $\boldsymbol{P}^{\otimes n}$ does not dominate $\boldsymbol{Q}^{\otimes n}$ for any odd $n$.
In particular $\boldsymbol{P}$ does not dominate $\boldsymbol{Q}$ in large samples.
\end{proposition}

So the conjecture of~\cite[Appendix~K]{mu2021blackwell} is false for three states.
Both experiments are bounded, since they have finitely many outcomes and every likelihood is positive.

The two experiments tie at the tropical direction $\beta=(1,-\tfrac12,-\tfrac12)$, where $\Tdiv_{\beta}(\boldsymbol{P})=\Tdiv_{\beta}(\boldsymbol{Q})=\log\tfrac52$: the ratio $P_{1}/\sqrt{P_{2}P_{3}}$ takes the values $\tfrac52$, $4/\sqrt5$ and $1/\sqrt{20}$, and $Q_{1}/\sqrt{Q_{2}Q_{3}}$ takes the values $\tfrac52$, $\sqrt{2/5}$ and $1$.
Dominance at $n=2$ also shows that, for three states, dominance for some $n$ is strictly weaker than Definition~1 even for a generic pair, whereas for two states the definitions agree under genericity~\cite{mu2021blackwell}.
The example therefore refutes the conjecture under Definition~1 and leaves a version with dominance for some $n$ untouched.
Already for two states, dominance at one sample size need not persist at the next: at one parameter choice in~\cite[Example~2]{mu2021blackwell}, $\boldsymbol{P}^{\otimes 2}$ dominates $\boldsymbol{Q}^{\otimes 2}$ while $\boldsymbol{P}^{\otimes 3}$ does not dominate $\boldsymbol{Q}^{\otimes 3}$.
The pair above fails at every odd sample size while meeting the necessary conditions of~\cite[Proposition~7]{mu2021blackwell} strictly.
The main theorems of this paper do not depend on the conjecture: Theorem~\ref{thm:correct} uses~\cite[Theorem~19]{farooq2024matrix} with its strict tropical hypotheses.

\begin{proof}[Proof of Proposition~\ref{prop:mpst-counterexample}]
Both experiments are invariant under simultaneous permutation of the states and the outcomes, and the permutations of three states act transitively on ordered pairs of distinct states, so genericity and condition (iii) need to be checked for one ordered pair of states only.

\emph{Genericity and relative entropies.}
The ratio $P_{1}/P_{2}$ takes the values $\tfrac54,5,\tfrac45,4,\tfrac15,\tfrac14$, with maximum $5$ and minimum $\tfrac15$, and $Q_{1}/Q_{2}$ takes the values $\tfrac52,\tfrac25,1$, with maximum $\tfrac52$ and minimum $\tfrac25$, so the pair is generic.
Each ordered pair of distinct values among $\tfrac14,\tfrac15,\tfrac1{20}$ occurs once as $(P_{1},P_{2})$, so
\[
\KLdiv(P_{1}\|P_{2})=\tfrac14\log\tfrac{25}{4}+\tfrac15\log\tfrac{16}{5}-\tfrac1{20}\log 20,
\qquad
\KLdiv(Q_{1}\|Q_{2})=\tfrac5{18}\log\tfrac52+\tfrac19\log\tfrac25=\tfrac16\log\tfrac52
\]
and $\KLdiv(P_{1}\|P_{2})-\KLdiv(Q_{1}\|Q_{2})=\tfrac13\log\tfrac52+\tfrac15\log\tfrac{16}{5}-\tfrac1{20}\log 20>\tfrac13\cdot\tfrac9{10}+\tfrac15-\tfrac3{20}>0$, since $e^{9/10}<\tfrac52$, $e<\tfrac{16}{5}$ and $e^{3}>20$.
This is condition (iii).

\emph{A strict inequality for weighted geometric means.}
Let $\alpha\in(\Aplus\cup\Aminus)\setminus E$ and $G_{\alpha}(u):=\prod_{k}u_{k}^{\alpha_{k}}$ for $u\in(0,\infty)^{3}$, and let $c_{i}>0$ and $u_{i}\in(0,\infty)^{3}$ for finitely many $i$.
If $\alpha\in\Aplus$, H\"older's inequality with the exponents $1/\alpha_{k}$ over the states with $\alpha_{k}>0$ gives $\sum_{i}c_{i}G_{\alpha}(u_{i})=\sum_{i}\prod_{k}(c_{i}u_{i,k})^{\alpha_{k}}\le G_{\alpha}(\sum_{i}c_{i}u_{i})$.
If $\alpha\in\Aminus$ with $\alpha_{j}>1$, let $N:=\{k:\alpha_{k}<0\}$.
The identity $c_{i}u_{i,j}=(c_{i}G_{\alpha}(u_{i}))^{1/\alpha_{j}}\prod_{k\in N}(c_{i}u_{i,k})^{|\alpha_{k}|/\alpha_{j}}$ and H\"older's inequality with the exponents $\alpha_{j}$ and $\alpha_{j}/|\alpha_{k}|$, $k\in N$, give $\sum_{i}c_{i}u_{i,j}\le(\sum_{i}c_{i}G_{\alpha}(u_{i}))^{1/\alpha_{j}}\prod_{k\in N}(\sum_{i}c_{i}u_{i,k})^{|\alpha_{k}|/\alpha_{j}}$, which rearranges to $G_{\alpha}(\sum_{i}c_{i}u_{i})\le\sum_{i}c_{i}G_{\alpha}(u_{i})$.
For $\alpha\in\Aplus$, equality in H\"older's inequality forces the ratios $u_{i,k}/u_{i,\ell}$ to be independent of $i$ for all states $k,\ell$ with $\alpha_{k},\alpha_{\ell}>0$.
For $\alpha\in\Aminus$, equality forces $G_{\alpha}(u_{i})/u_{i,k}$ and the ratios among the states in $N$ to be independent of $i$, and $G_{\alpha}(u_{i})/u_{i,k}$ is then a constant multiple of $(u_{i,j}/u_{i,k})^{\alpha_{j}}$ for $k\in N$, so the ratios with the state $j$ are independent of $i$ as well.
In both cases $\alpha$ has at least two nonzero coordinates.
So if for every pair of states two of the vectors $u_{i}$ have different ratios, the inequality is strict for every $\alpha\in(\Aplus\cup\Aminus)\setminus E$.

\emph{Dominance at $n=2$, and conditions (i) and (ii).}
Write an outcome of $\boldsymbol{P}$ by its likelihood vector $p$, and an outcome of $\boldsymbol{P}^{\otimes 2}$ as a pair $(p,p')$ with likelihood vector $p\odot p'$, the entrywise product.
Let $q_{j}$ be the outcome of $\boldsymbol{Q}$ with $\tfrac5{18}$ at state $j$ and $q_{0}$ the outcome with likelihood vector $(\tfrac12,\tfrac12,\tfrac12)$.
Merge the outcomes of $\boldsymbol{Q}^{\otimes 2}$ with equal likelihood vectors into $y_{jj}:=(q_{j},q_{j})$, $y_{jk}:=\{(q_{j},q_{k}),(q_{k},q_{j})\}$ for $j\ne k$, $y_{j0}:=\{(q_{j},q_{0}),(q_{0},q_{j})\}$ and $y_{00}:=(q_{0},q_{0})$.
Splitting a merged outcome evenly between its members is a garbling, so it suffices to garble $\boldsymbol{P}^{\otimes 2}$ onto the merged experiment.
For a permutation $p$ let $s_{a}(p)$ be the state at which $p$ takes the value $a$, and send each $(p,p')$ as follows.
\begin{itemize}[leftmargin=*]
\item If $p'=p$: to $y_{j0}$ with $j=s_{1/5}(p)$, with probability $1$.
\item If $p'$ is $p$ with the values $\tfrac15,\tfrac1{20}$ exchanged: to $y_{jj}$ with $j=s_{1/4}(p)$ with probability $\tfrac{50}{81}$, and to $y_{00}$ with probability $\tfrac{31}{81}$.
\item If $p'$ is $p$ with the values $\tfrac14,\tfrac15$ exchanged: to $y_{jk}$ with $\{j,k\}$ the states other than $s_{1/20}(p)$ with probability $\tfrac{400}{729}$, and to $y_{00}$ with probability $\tfrac{329}{729}$.
\item If $p'$ is $p$ with the values $\tfrac14,\tfrac1{20}$ exchanged: to $y_{jk}$ with $\{j,k\}$ the states other than $s_{1/5}(p)$ with probability $\tfrac{200}{729}$, and to $y_{00}$ with probability $\tfrac{529}{729}$.
\item If $p'$ differs from $p$ at every state: to $y_{j0}$ with probability $\tfrac{71}{72}$ and to $y_{i0}$ with probability $\tfrac1{72}$, where $p\odot p'$ equals $\tfrac1{20}$ at state $j$ and $\tfrac1{100}$ at state $i$.
\end{itemize}
Each row sums to one.
Every rule is defined through the maps $s_{a}$, so the kernel commutes with permutations of the states, and it suffices to compute the mass sent to $y_{11}$, $y_{12}$, $y_{10}$ and $y_{00}$, where each list below records the likelihood vectors of the incoming pairs, scaled by $400$:
\begin{align*}
y_{11}:&\ \tfrac{50}{81}\cdot 2\cdot(25,4,4)=\tfrac{400}{324}(25,4,4),\\
y_{12}:&\ \tfrac{400}{729}\cdot 2\cdot(20,20,1)+\tfrac{200}{729}\cdot 2\cdot(5,5,16)=\tfrac{400}{81}(5,5,2),\\
y_{10}:&\ (16,25,1)+(16,1,25)+2\cdot\tfrac{71}{72}\big((20,4,5)+(20,5,4)\big)+2\cdot\tfrac1{72}\big((4,20,5)+(4,5,20)\big)=\tfrac{400}{18}(5,2,2),\\
y_{00}:&\ \tfrac{31}{81}\cdot 2\cdot(33,33,33)+\tfrac{329}{729}\cdot 2\cdot(41,41,41)+\tfrac{529}{729}\cdot 2\cdot(26,26,26)=400\cdot(\tfrac14,\tfrac14,\tfrac14)
\end{align*}
These are the likelihood vectors of $y_{11}$, $y_{12}$, $y_{10}$ and $y_{00}$, so the kernel is a garbling of $\boldsymbol{P}^{\otimes 2}$ onto $\boldsymbol{Q}^{\otimes 2}$.
Merging outcomes with equal likelihood vectors leaves $H_{\alpha}$ unchanged, because $G_{\alpha}$ is positively homogeneous of degree one, so applying the inequality above at each merged outcome and summing gives $H_{\alpha}(\boldsymbol{Q})^{2}=H_{\alpha}(\boldsymbol{Q}^{\otimes 2})\ge H_{\alpha}(\boldsymbol{P}^{\otimes 2})=H_{\alpha}(\boldsymbol{P})^{2}$ for $\alpha\in\Aplus\setminus E$, and the reverse for $\alpha\in\Aminus\setminus E$.
The outcome $y_{00}$ receives pairs with likelihood vectors $(25,4,4)$, $(4,25,4)$ and $(4,4,25)$, up to scale, whose ratios differ for every pair of states, so the inequalities are strict.
As $H_{\alpha}>0$, this gives $H_{\alpha}(\boldsymbol{P})<H_{\alpha}(\boldsymbol{Q})$ on $\Aplus\setminus E$ and $H_{\alpha}(\boldsymbol{P})>H_{\alpha}(\boldsymbol{Q})$ on $\Aminus\setminus E$, which are conditions (ii) and (i).

\emph{No dominance at odd $n$.}
Let $n$ be odd and suppose a garbling $T$ maps $\boldsymbol{P}^{\otimes n}$ onto $\boldsymbol{Q}^{\otimes n}$.
The function $\psi(u):=u_{1}-(\tfrac52)^{n}\sqrt{u_{2}u_{3}}$ on $(0,\infty)^{3}$ satisfies $\psi(\sum_{i}c_{i}u_{i})\le\sum_{i}c_{i}\psi(u_{i})$ for $c_{i}>0$, by the Cauchy--Schwarz inequality $\sum_{i}c_{i}\sqrt{u_{i,2}u_{i,3}}\le\sqrt{(\sum_{i}c_{i}u_{i,2})(\sum_{i}c_{i}u_{i,3})}$, with equality only if $u_{i,2}/u_{i,3}$ does not depend on $i$.
An outcome $(p^{(1)},\dots,p^{(n)})$ of $\boldsymbol{P}^{\otimes n}$ has likelihood vector $u$ with
\[
\frac{u_{1}}{\sqrt{u_{2}u_{3}}}=\prod_{i=1}^{n}\frac{p^{(i)}_{1}}{\sqrt{p^{(i)}_{2}p^{(i)}_{3}}}\le\big(\tfrac52\big)^{n}
\]
because each factor equals $\tfrac52$, $4/\sqrt5$ or $1/\sqrt{20}$ according as $p^{(i)}_{1}$ is $\tfrac14$, $\tfrac15$ or $\tfrac1{20}$; so $\psi\le 0$ on the outcomes of $\boldsymbol{P}^{\otimes n}$, with equality exactly when every $p^{(i)}_{1}$ equals $\tfrac14$.
The outcome $y^{\ast}:=(q_{1},\dots,q_{1})$ of $\boldsymbol{Q}^{\otimes n}$ has likelihood vector $((\tfrac5{18})^{n},(\tfrac19)^{n},(\tfrac19)^{n})$, at which $\psi=0$.
Hence $0=\psi(\sum_{x}T(y^{\ast}|x)u_{x})\le\sum_{x}T(y^{\ast}|x)\psi(u_{x})\le 0$, so $T(y^{\ast}|x)>0$ only for outcomes $x$ with every $p^{(i)}_{1}=\tfrac14$, and the ratio $u_{x,2}/u_{x,3}$ is the same for all of them; such outcomes exist, since the likelihood vector of $y^{\ast}$ is nonzero.
For such an outcome each $p^{(i)}_{2}/p^{(i)}_{3}$ is $4$ or $\tfrac14$, so $u_{x,2}/u_{x,3}=4^{a-b}$ with $a+b=n$, and since $n$ is odd, $a\ne b$ and the ratio differs from $1$.
But $\sum_{x}T(y^{\ast}|x)u_{x}$ is the likelihood vector of $y^{\ast}$, whose ratio between states $2$ and $3$ is $1$, a contradiction.
\end{proof}

\section{Convergent evidence for the special role of the logarithm}\label{app:lognatural}

The axioms force the logarithm of a coincidence weight; several axiomatic and operational routes that do not invoke data processing, taken below in decreasing order of structural force, lead to the same logarithm or read the same quantities.

\subsection{Kolmogorov--Nagumo plus R\'enyi's mean axiomatics}\label{ssec:kn}

A DPI-free route combines two classical theorems.

\emph{Step 1 (Kolmogorov 1930, Nagumo 1930).} A function $M:\bigsqcup_{n}\R_{>0}^{n}\to\R_{>0}$ is a \emph{quasi-arithmetic mean}
\[
M_{\varphi}(x_{1},\dots,x_{n})=\varphi^{-1}\Bigl(\tfrac{1}{n}\textstyle\sum_{i=1}^{n}\varphi(x_{i})\Bigr)
\]
for some continuous strictly monotone $\varphi$ if and only if $M$ is continuous,
permutation-symmetric, reflexive ($M(x,\dots,x)=x$), strictly increasing in each argument,
and \emph{decomposable},
\[
M(x_{1},\dots,x_{n}) = M\bigl(\underbrace{M(x_{1},\dots,x_{k}),\dots,M(x_{1},\dots,x_{k})}_{k\text{ copies}},x_{k+1},\dots,x_{n}\bigr)
\]
The generator $\varphi$ is determined up to affine transformation~\cite{kolmogorov1930sur,nagumo1930uber,acz1989}.

\emph{Step 2 (R\'enyi 1961~\cite{renyi1961measures}).} R\'enyi took the $p_{i}$-weighted quasi-arithmetic mean of the individual informations $-\log p_{i}$; writing $\varphi(t)=g(-\log t)$ turns that construction into the equivalent form used here, the negative log of the $p_{i}$-weighted quasi-arithmetic mean of $\{p_{i}\}$ themselves with generator $\varphi$:
\begin{equation}\label{eq:renyi-mean}
\Renent_{\alpha}(P) := -\log\varphi^{-1}\!\Big(\textstyle\sum_{i}p_{i}\,\varphi(p_{i})\Big)
\end{equation}
and asked which $\varphi$ make $\Renent_{\alpha}$ additive on independent products,
$\Renent_{\alpha}(P\otimes Q)=\Renent_{\alpha}(P)+\Renent_{\alpha}(Q)$.
The superscript distinguishes this entropy from the Hellinger transform $H_{\alpha}$ of \eqref{eq:hellinger}, which is multiplicative on independent products.
Additivity forces $\varphi$, up to affine transformation, to be either logarithmic, $\varphi(t)=\log t$ (giving Shannon entropy, the $\alpha\to 1$ limit), or a power, $\varphi(t)=t^{\alpha-1}$ for some real $\alpha\ne 1$ (giving R\'enyi entropy of order $\alpha$, usually taken with $\alpha>0$).

Composing Steps~1--2 characterizes the R\'enyi entropies using only continuity,
permutation symmetry, reflexivity, strict monotonicity, decomposability, and additivity on independent factors, with no DPI assumption.
R\'enyi divergences arise as the cross-entropy version of the same construction, with the same generator constraint and the same one-parameter family.

For several distributions the construction has no direct analogue, because the Hellinger transform $H_{\alpha}(\boldsymbol{\pi})=\E_{\nu}[\prod_{k}\pi_{k}^{\alpha_{k}}]$ averages a product of powers under $\nu$ and is not a quasi-arithmetic mean of a single probability vector.
The additivity requirement carries over, and Proposition~\ref{prop:log-forced} shows that, together with data processing, it forces the logarithm for every functional of a single Hellinger transform indexed by the simplex.

\subsection{Other axiomatic characterizations}

\paragraph{Classical Shannon-entropy axiomatics~\cite{khinchin1957,forte1974why,hobson1969newb,faddeev1956concept,aczel1975}.}
The post-Shannon axiomatic derivations begin from monographs in the late 1950s (continuity, additivity, monotonicity, branching) and branching axiomatizations~\cite{khinchin1957,faddeev1956concept}; later sharpenings showed that symmetry,
expansibility, additivity, and subadditivity characterize positive linear combinations of Shannon and Hartley entropy, with continuity at $n=2$ pinning Shannon entropy alone~\cite{forte1974why,aczel1975}.
The mathematical engine is again Cauchy's $L(xy)=L(x)+L(y)$.
The analogous statement for $\KLdiv$ is~\cite{hobson1969newb}.
The relative-entropy characterization is the $W=2$ case of the vertex-edge stratum of \eqref{eq:correct}, with weight $c_{k\ell}$ and with the normalized vertex derivative \eqref{eq:KL-as-limit} as its atom. The Shannon-entropy axiomatics is a parallel characterization driven by the same functional equation. It is not a special case: entropy is not a DPI--additive divergence, violating both the coincidence ground state and joint data-processing monotonicity, with no contradiction.

\paragraph{Axiomatic characterization of $f$-divergences~\cite{csisz2008axiomatic}.}
The $f$-divergences that are additive on products are the positive combinations of relative entropy and its reverse~\cite{csisz2008axiomatic}; the R\'enyi divergences of order $\alpha\ne 1$ are not $f$-divergences at all, being logarithms of the Hellinger integrals $\int\mu^{\alpha}\nu^{1-\alpha}$, which are.
The bivariate representation~\cite{mu2021blackwell} sharpens this: ``additive on products + DPI on bounded pairs'' (with no $f$-divergence assumption) already forces the R\'enyi mixture form.

\paragraph{Resource-theoretic characterization~\cite{gour2021entropy}.}
A third independent route to the same $W=2$ family runs through relative majorization of pairs $(p,q)$ of probability distributions.
There any functional monotone under classical channels, additive on independent products, and suitably normalized is sandwiched between the R\'enyi relative entropies $D_{0}$ and $D_{\infty}$.
That work leaves open whether its axioms restrict entropies to convex combinations of R\'enyi entropies; combined with~\cite{mu2021blackwell} they do, with orders in $[0,\infty]$~\cite[Proposition~2.13]{rubboli2026conditional}.
The axiomatic package is ``DPI + additivity + normalization'', differing from the $f$-divergence-cone route above and from~\cite{mu2021blackwell} (no cone assumption, bounded-pair DPI) by working through catalytic relative majorization.

\paragraph{Maximum-entropy axiomatics~\cite{johnson1979axiomatic}.}
The principle of minimum cross-entropy is derived from four axioms: uniqueness, invariance under coordinate transformations,
system independence, and subset independence~\cite{johnson1979axiomatic}.
The unique cost functional is $\KLdiv(\rho\|q)$.
The system-independence axiom is the binary case of the additivity-on-products axiom; subset-independence is closely related to joint DPI restricted to disintegrating kernels.
The four-axiom characterization selects $\KLdiv$, the vertex derivation of the $\Cdiv_{\alpha}$ family, as the unique cost functional on a simplex of priors.
The present $\paramset$-cone is a multi-distribution generalization of the same selection statement, with the four-axiom criterion replaced by the multi-prior DPI--additivity package and the unique answer enlarged from a single cross-entropy functional to the full positive integral over $\paramset$.

\paragraph{Pointwise proper-loss uniqueness~\cite{mccarthy1956measures,shuford1966admissible,bernardo1979expected,gneiting2007strictly}.}
A scoring rule $S(p,x)$ on a probabilistic forecast $p$ and observed outcome $x$ is \emph{strictly proper} if $\E_{q}[S(q,X)]\ge\E_{q}[S(p,X)]$ for every pair $(p,q)$ with equality only at $p=q$; it is \emph{local} (pointwise) if $S(p,x)$ depends on $p$ only through $p(x)$.
The classical result, going back to~\cite{mccarthy1956measures} and~\cite{shuford1966admissible} and made canonical by~\cite{bernardo1979expected}
(\cite{gneiting2007strictly} survey it),
is that, on outcome spaces with at least three points, the only strictly proper local scoring rules are positive affine transforms of the logarithmic score $S(p,x)=\log p(x)$.
The expected log-score difference is the relative entropy; this fixes the unique calibration-respecting pointwise loss on a single forecast.
Theorem~\ref{thm:correct} is also a uniqueness statement that pins down $\log$ as the canonical wrapping functional, under a different axiomatic package (joint DPI plus additivity-on-products plus the coincidence ground state on $W$-tuples,
in place of strict-properness plus locality on a single forecast).
The pointwise-proper-loss theorem operates on the single-forecast cone and selects the binary log-likelihood ratio at the vertex; the multi-prior representation operates on the $W$-tuple cone and selects the full family of rescaled atoms $\Dren_{\alpha}$ and their limits, with relative entropy re-appearing as the vertex derivation (Equation~(\ref{eq:KL-as-limit})).
Both routes single out the same generator $\log$.

\paragraph{Information radius~\cite{sibson1969information}.}
The Kullback-projection radius $\inf_{r}\sum_{k}\alpha_{k}\KLdiv(r\|\pi_{k})$,
with the variable center $r$ in the \emph{first} argument of each relative-entropy term,
is exactly $\Cdiv_{\alpha}$ on the simplex (the optimum $r=p^{\star}_{\alpha}\propto\prod_{k}\pi_{k}^{\alpha_{k}}$ is the geometric mixture, and substituting back gives the coincidence identity; see Appendix~\ref{app:radius}).
This is the reverse-orientation companion of Sibson's information radius $\inf_{r}\sum_{k}\alpha_{k}\KLdiv(\pi_{k}\|r)$,
which instead places $r$ in the second argument.
The latter is minimized at the arithmetic mixture $\sum_{k}\alpha_{k}\pi_{k}$ and reduces to the Jensen--Shannon divergence at the uniform weight, so the two radii are different functionals that coincide only in degenerate cases.
The geometric-mixture (first-argument) form equals $\Cdiv_{\alpha}$; the second-argument form is the multi-distribution generalization of mutual information~\cite{sibson1969information}, and is not in the $\Cdiv_{\alpha}$ family.

\paragraph{Coincidence reading: a generalization of R\'enyi's interpretation~\cite{balsubramani2026information}.}
The axiomatic routes above derive $-\log H_{\alpha}$ from \emph{structural} requirements (DPI, additivity, mean-style closure,
$f$-divergence-cone constraints, resource-theoretic monotonicity).
A complementary route is \emph{operational}, generalizing R\'enyi's original coincidence interpretation of his binary divergence.
On a finite alphabet with counting measure, $\sum_{x}\mu(x)^{m}\nu(x)^{n}$ is the probability that $m$ independent draws from $\mu$ and $n$ independent draws from $\nu$ all land on the same point, a Hellinger transform at an exponent vector off the affine slice when $m+n>1$.
On the slice the coincidence reading passes through empirical distributions: for rational $\boldsymbol{\alpha}$ in the simplex, $\Cdiv_{\boldsymbol{\alpha}}$ is the exponential rate at which samples of sizes in the proportions $\boldsymbol{\alpha}$ all have the same empirical distribution (Proposition~\ref{prop:coincidence-rate}).
The development in~\cite{balsubramani2026information} sets out $\log Z(\boldsymbol{\alpha}) := \log H_{\boldsymbol{\alpha}}(\boldsymbol{\pi})$ through an identity that holds at every $\boldsymbol{\alpha}$ and joins four perspectives: a Boltzmann coincidence weight, the geometric-mixture normalizer, the value of the unconstrained max-entropy Lagrangian $\max_{p}[\text{H}(p)-\sum_{k}\alpha_{k}\text{H}(p,\pi_{k})]$, and the relative-entropy barycenter optimum on the simplex.
The max-entropy reading reaches into the non-simplex regions of Theorem~\ref{thm:correct} directly.
Mixed-sign exponents are reversed-direction (repulsive) Lagrange multipliers in the max-entropy problem, the tropical boundary appears as the zero-temperature limit of the Gibbs equilibrium $p_{\boldsymbol{\alpha}}^{\star}$, and the relative-entropy vertex atoms are the derivative-style expansion of the same partition function at the simplex vertices (Section~\ref{sec:exotic}).
None of these readings invokes data-processing monotonicity; all converge on the same $\Cdiv_{\boldsymbol{\alpha}} = -\log Z(\boldsymbol{\alpha})$ that Theorem~\ref{thm:correct} selects axiomatically, after rescaling by $1/(1-\alpha_{\star})$.

\paragraph{Distributional reading.}
The coincidence reading has a distributional companion in the exact free-energy identity of Section~\ref{ssec:laplace-mixed-concentration}
(Theorem~\ref{thm:laplace-mixed-concentration}), whose concentration corollary strengthens to a level-2 large-deviation principle (Corollary~\ref{cor:level2-ldp}).
Where the coincidence identity describes $\Cdiv_{\boldsymbol{\alpha}}$ as the exponential rate at which samples in the proportions $\boldsymbol{\alpha}$ share an empirical distribution, the concentration result describes it as the limiting free energy of the Laplace-mixed posterior on $\Delta(\mathcal{X})$, which concentrates on the geometric mixture $p_{\boldsymbol{\alpha}}^{\star}\propto\prod_{k}\pi_{k}^{\alpha_{k}}$.
The rate function is $\KLdiv(\cdot \| p_{\boldsymbol{\alpha}}^{\star})$, minimized at $p = p_{\boldsymbol{\alpha}}^{\star}$, where the energy $\mathcal{F}(p)=\sum_{k}\alpha_{k}H(p,\pi_{k})-H(p)$ attains its minimum value $\Cdiv_{\boldsymbol{\alpha}}$.
Together, the coincidence and distributional readings form a complement to the structural axiomatic routes: the axiomatic routes derive $-\log H_{\boldsymbol{\alpha}}$ from what the functional must satisfy; the coincidence and concentration routes derive it from what the functional measures.

\subsection{Operational routes}

\paragraph{Multi-hypothesis testing.}
The multi-hypothesis Bayes error exponent was established as $\min_{k\ne\ell}\mathsf{C}_{\mathrm{Ch}}(\pi_{k},\pi_{\ell})$, the minimum over pairwise Chernoff informations~\cite{salikhov1973asymptotic,leang1997asymptotics}.
The quantum analogue is the quantum Chernoff exponent, whose binary converse is established in~\cite{nussbaum2009chernoff} via the Petz-type R\'enyi divergence, with multi-hypothesis upper bounds in~\cite{audenaert2014upper}.
The dual quantity
\[
\max_{\alpha\in\simplex}\Cdiv_{\alpha}\;=\;\min_{r}\max_{k}\KLdiv(r\|\pi_{k})
\]
is the worst-case Kullback projection radius and the upper envelope of the simplex-indexed family $\{\Cdiv_{\alpha}\}_{\alpha\in\simplex}$ (Section~\ref{sec:w3}).

\paragraph{Multi-lottery betting~\cite{ducuara2026}.}
A recent betting interpretation:
$\Dren_{\alpha}(\boldsymbol{p}_{X})$ equals the log of the isoelastic certainty equivalent of a betting game with $W-1$ lotteries on $X$, with risk-aversion parameters $R_{k}=1+\alpha_{k}/\alpha_{0}$, in the source's $0,\dots,W-1$ indexing, and under fair odds with optimal betting.
This is the multi-prior analogue of the classical Kelly gambling characterization of R\'enyi divergence in the binary case~\cite{bleuler2020}.

\paragraph{Quantum/resource-theoretic extensions.}
In the noncommutative case the sandwiched and Petz-type quantum R\'enyi divergences lie in the test spectrum, as two slices of the $\alpha$-$z$ family~\cite{haapasalo2025barycentric}, alongside the Kubo--Ando and matrix-mean divergences.
That spectrum is not characterized (Section~\ref{ssec:noncommutative}), and on commuting tuples each of its points restricts to a single classical atom.

\section{Proofs for Section~\ref{sec:readings}}\label{app:radius}

\subsection{The information radius}

The mixed coincidence identity gives, for $\alpha\in\simplex$ and priors on a finite alphabet with full support,
\begin{equation}\label{eq:mixed-coincidence-radius}
\Cdiv_{\alpha}(\boldsymbol{\pi}) = \min_{r\in\Delta(\mathcal{X})}\sum_{k=1}^{W}\alpha_{k}\KLdiv(r\|\pi_{k})
\end{equation}
with optimum $r=p^{\star}_{\alpha}\propto\prod_{k}\pi_{k}^{\alpha_{k}}$, the geometric mixture.
Writing $Z(\alpha)=\sum_{x}\prod_{k}\pi_{k}(x)^{\alpha_{k}}$ and $p^{\star}_{\alpha}=\frac{1}{Z(\alpha)}\prod_{k}\pi_{k}^{\alpha_{k}}$, and using $\sum_{k}\alpha_{k}=1$,
\[
\sum_{k}\alpha_{k}\KLdiv(r\|\pi_{k})
= \sum_{x} r(x)\log\frac{r(x)}{\prod_{k}\pi_{k}(x)^{\alpha_{k}}}
= \KLdiv(r\|p^{\star}_{\alpha}) - \log Z(\alpha),
\]
which is minimized over $r$ at $r=p^{\star}_{\alpha}$, where $\KLdiv(r\|p^{\star}_{\alpha})=0$, giving the minimum value $-\log Z(\alpha)=\Cdiv_{\alpha}(\boldsymbol{\pi})$.
The function $(\alpha,r)\mapsto\sum_{k}\alpha_{k}\KLdiv(r\|\pi_{k})$ is linear in $\alpha$ on the compact convex set $\simplex$, and convex and continuous in $r$ on the compact convex set $\Delta(\mathcal{X})$ because every $\pi_{k}$ has full support, so Sion's minimax theorem gives
\[
\max_{\alpha\in\simplex}\Cdiv_{\alpha}(\boldsymbol{\pi}) = \min_{r}\max_{\alpha\in\simplex}\sum_{k}\alpha_{k}\KLdiv(r\|\pi_{k})=\min_{r}\max_{k}\KLdiv(r\|\pi_{k})
\]
the \emph{information radius}, or worst-case Kullback projection radius.

The support functional $\max_{\alpha\in\simplex}\Cdiv_{\alpha}$ is sub-additive on products, and strictly so already for $W=2$ and binary alphabets.
With $(\mu,\nu)=((0.9,0.1),(0.2,0.8))$ and $(\mu',\nu')=((0.95,0.05),(0.5,0.5))$, $\max_{t}\Cdiv_{(t,1-t)}(\mu,\nu)+\max_{t}\Cdiv_{(t,1-t)}(\mu',\nu')\approx0.5162$ while $\max_{t}\Cdiv_{(t,1-t)}(\mu\otimes\mu',\nu\otimes\nu')\approx0.5152$.
The maximizing orders of the two factors are about $0.477$ and $0.434$, and that of the product is about $0.464$.
The information radius is therefore not an element of the additive cone of Theorem~\ref{thm:correct}, while each fixed-$\alpha$ atom $\Cdiv_{\alpha}$ is.

\subsection{Coincidence of empirical distributions}

\begin{proof}[Proof of Proposition~\ref{prop:coincidence-rate}]
Write $n_{k}:=\alpha_{k}N$, $\mathcal{K}:=\{k:\alpha_{k}>0\}$, and $\mathcal{P}_{N}$ for the set of distributions $p$ on $\mathcal{X}$ with $n_{k}\,p(x)\in\mathbb{Z}$ for every $k\in\mathcal{K}$ and $x\in\mathcal{X}$, the possible common empirical distributions.
For $k\in\mathcal{K}$ and $p\in\mathcal{P}_{N}$, the probability that the $k$-th sample has empirical distribution $p$ lies between $(n_{k}+1)^{-|\mathcal{X}|}e^{-n_{k}\KLdiv(p\|\pi_{k})}$ and $e^{-n_{k}\KLdiv(p\|\pi_{k})}$~\cite[Theorem~11.1.4]{cover1999elements}.
Every $p\in\mathcal{P}_{N}$ has $Np(x)=\sum_{k}n_{k}p(x)\in\mathbb{Z}$, so $\mathcal{P}_{N}$ has at most $(N+1)^{|\mathcal{X}|}$ elements~\cite[Theorem~11.1.1]{cover1999elements}.
With $G(p):=\sum_{k}\alpha_{k}\KLdiv(p\|\pi_{k})$, independence of the samples and $n_{k}\le N$ give
\[
(N+1)^{-W|\mathcal{X}|}\,e^{-N\min_{\mathcal{P}_{N}}G}\;\le\;\Pr(E_{N})=\sum_{p\in\mathcal{P}_{N}}\prod_{k\in\mathcal{K}}\Pr(\text{sample }k\text{ has empirical distribution }p)\;\le\;(N+1)^{|\mathcal{X}|}\,e^{-N\min_{\mathcal{P}_{N}}G}
\]
By~\eqref{eq:mixed-coincidence-radius}, $G\ge\Cdiv_{\alpha}(\boldsymbol{\pi})$ on $\Delta(\mathcal{X})$ with equality at $p_{\alpha}^{\star}$, so the upper bound gives $\limsup_{N}\frac{1}{N}\log\Pr(E_{N})\le-\Cdiv_{\alpha}(\boldsymbol{\pi})$.
Each $n_{k}$ is a multiple of $N/q$, so $\mathcal{P}_{N}$ contains every distribution with values in $\frac{q}{N}\mathbb{Z}$, and there are $p_{N}\in\mathcal{P}_{N}$ with $p_{N}\to p_{\alpha}^{\star}$.
The priors have full support, so $G$ is continuous on $\Delta(\mathcal{X})$, $\min_{\mathcal{P}_{N}}G\le G(p_{N})\to\Cdiv_{\alpha}(\boldsymbol{\pi})$, and the lower bound gives $\liminf_{N}\frac{1}{N}\log\Pr(E_{N})\ge-\Cdiv_{\alpha}(\boldsymbol{\pi})$.
\end{proof}

\section{The Laplace-transform normal form}\label{app:laplace}

The Laplace-transform reading of the Hellinger transform $H_{\alpha}$ complements the axiomatic and operational characterizations in the main text.
The observation is purely structural: $H_{\alpha}(\boldsymbol{\pi})$ is a multivariate Laplace transform of the joint law of the per-prior log-losses, viewed as a pushforward measure on $\R^{W}$ (more precisely on the extended box $(-\infty,+\infty]^{W}$ when some $\pi_{k}$ vanishes on positive reference mass; see the theorem statement).
This recasting makes cumulants, exponential tilts, Chernoff bounds,
saddlepoint methods and moment-uniqueness arguments available directly in the multi-way coincidence calculus, and supplies an exact free-energy identity, with a level-2 concentration corollary, alongside the representation of Theorem~\ref{thm:correct}. We use $Z(\alpha)
:= H_{\alpha}(\boldsymbol{\pi})$ as a synonym throughout this appendix to align
with the standard Laplace-transform notation.

The quantity $\Cdiv_{\alpha}=-\log H_{\alpha}=-\log Z(\alpha)$ has a long prehistory in the multi-distribution affinity literature.
When $\alpha_{k}=1/W$ for all $k$, $Z(\alpha)$ reduces to Matusita's classical multi-distribution \emph{affinity}~\cite{matusita1967affinity} $\rho(\pi_1,\dots,\pi_W) = \int \prod_{k=1}^{W} \pi_{k}^{1/W}\,d\nu$, a symmetric multi-distribution generalization of the Bhattacharyya coefficient measuring statistical closeness for several distributions at once.
It is applied to classification among several Gaussian populations in~\cite{matusita1967classification}.
Inequalities sandwiching this multi-distribution affinity by averages of pairwise Hellinger / Bhattacharyya-type affinities are derived in~\cite{toussaint1974}, foreshadowing the edge-restriction phenomenon in multi-hypothesis testing:
the uniform-weight multi-way affinity is sandwiched by symmetric functions of the pairwise affinities, but the tight MAP error exponent is governed by the hardest pair.
The multi-way coincidence divergence $\Cdiv_{\alpha}$ subsumes this classical object, extending it to non-uniform $\alpha$ (and to the signed-exponent / tropical strata of Section~\ref{sec:atoms}), with the relative-entropy barycenter / minimax-radius interpretation supplied by Appendix~\ref{app:radius}.

Earlier work~\cite{galke2024sufficiencyrenyi} observed that R\'enyi divergences in the binary case can be written as Laplace transforms.
The multi-way generalization below is the natural extension to the $W$-prior setting; because the Hellinger transform is the multivariate Laplace transform of the log-loss vector, the classical analytic methods apply.

\subsection{The binary case as a one-parameter Laplace transform}

\begin{theorem}[Two-prior Hellinger transform as a Laplace transform of a log-loss gap]
\label{thm:binary-laplace}
Let $(\mathcal{X},\mathcal{F},\nu)$ be a $\sigma$-finite measure space and let $p,q$ be probability densities with respect to $\nu$ (so $\int p\,d\nu=\int q\,d\nu=1$), with $p\ll q$ in the sense that $q(x)=0$ implies $p(x)=0$ for $\nu$-a.e.\ $x$.

Define the two-prior partition function (the $W=2$ specialization of $Z(\alpha)$ with exponents $(\alpha,1-\alpha)$) by
\[
Z_{p,q}(\alpha)
:=\E_{X\sim \nu}\!\big[p(X)^{\alpha} q(X)^{1-\alpha}\big]
=\int_{\mathcal X} p(x)^{\alpha}q(x)^{1-\alpha}\,\nu(dx)\in[0,\infty]
\qquad \alpha\in\R
\]
and for $\alpha\ne 1$ the order-$\alpha$ R\'enyi divergence by
\[
R_\alpha(p\|q):=\frac{1}{\alpha-1}\log Z_{p,q}(\alpha)\in[-\infty,\infty]
\]
with the understanding that $R_\alpha(p\|q)=+\infty$ if $Z_{p,q}(\alpha) = +\infty$.
At $\alpha=0$ the convention $0^{0}=1$ gives $R_{0}(p\|q)=0$, and the order-$0$ R\'enyi divergence $-\log\int_{\{p>0\}}q\,d\nu$ is the limit of $R_{\alpha}(p\|q)$ as $\alpha\downarrow 0$.

Let $Q$ denote the probability measure $Q(dx) = q(x)\,\nu(dx)$ and let $r:=\frac{dP}{dQ}=\frac{p}{q}$ be the Radon--Nikodym derivative of $P(dx)=p(x)\,\nu(dx)$ with respect to $Q$.
Introduce the (extended-real) \emph{log-loss gap} / \emph{information-density} statistic
\[
T:\mathcal X\to (-\infty,+\infty],\qquad
T(x):=-\log r(x) = -\log\frac{p(x)}{q(x)}
\]
(so $T(x)=+\infty$ precisely on $\{x:p(x)=0<q(x)\}$; $T$ never takes the value $-\infty$ since $p\ll q$ forces $r<\infty$ $Q$-a.e.).
Because this zero set can have positive $Q$-mass, the pushforward of $Q$ by $T$ is in general a Borel probability measure on the extended half-line, with a possible atom at $+\infty$:
\[
\tilde{q} := T_{\#}Q,\qquad
\tilde{q}(A) := Q\!\big(T\in A\big)
= \int_{\mathcal{X}} \mathbf{1}_{T(x)\in A}\, q(x)\,\nu(dx),\qquad A\in\mathcal{B}\big((-\infty,+\infty]\big)
\]
(strict positivity $p,q>0$ $\nu$-a.e.\ removes the atom and returns $\tilde{q}$ to $\R$).
Then the partition function is the (two-sided) Laplace transform of $\tilde{q}$, with the convention $e^{-\alpha\cdot(+\infty)}=0$ for $\alpha>0$, $=+\infty$ for $\alpha<0$ and $=1$ for $\alpha=0$:
\[
Z_{p,q}(\alpha)=\int_{(-\infty,+\infty]} e^{-\alpha t}\,\tilde{q}(dt)
\qquad\forall \alpha\in\R
\]
where either side may take the value $+\infty$. (For $\alpha>0$ the atom at $+\infty$ contributes $0$, so the integral equals its restriction to $\R$.)

Consequently, for $\alpha\ne 1$,
\[
R_\alpha(p\|q)
=\frac{1}{\alpha-1}\log\int_{(-\infty,+\infty]} e^{-\alpha t}\,\tilde{q}(dt)
\]
and the log-partition function $\Phi_{p,q}(\alpha):=\log Z_{p,q}(\alpha)$ is exactly the log-Laplace transform of the log-loss gap $T$ under $q$.
Moreover, defining $\tilde p(dt) := e^{-t}\,\tilde{q}(dt)$ (which assigns zero mass to the atom at $+\infty$, hence is supported on $\R$),
\begin{enumerate}[leftmargin=*]
\item $\tilde p$ is a probability measure on $\R$ and in fact $\tilde p = T_{\#}P$
(the law of $T(X)$ under $X\sim P$; note $T<\infty$ holds $P$-a.e.\ since $P(p=0)=0$).
\item The pair $(p,q)$ is interconvertible with the canonical pair $(\tilde p,\tilde{q})$ via stochastic maps between the corresponding $L^{1}$-spaces (equivalently, via Markov operators).
\end{enumerate}

Finally, if there exists a nonempty open interval $(a,b)\subset\R$ such that $Z_{p,q}(\alpha) < \infty$ for all $\alpha\in(a,b)$, then the restriction of $\tilde{q}$ to $\R$ (and hence $\tilde p$, which lives on $\R$) is uniquely determined by the function $\alpha\mapsto Z_{p,q}(\alpha)$ on $(a,b)$ (equivalently, by $\alpha\mapsto R_\alpha(p\|q)$ on $(a,b)$).
The remaining mass $\tilde{q}(\{+\infty\})=1-\tilde{q}(\R)$ is then fixed by total probability.%
\end{theorem}

\begin{proof}[Proof of Theorem~\ref{thm:binary-laplace}.]
\textbf{Step 1: Laplace representation and explicit construction of $\tilde{q}$.}
By construction, $\tilde{q}$ is the pushforward of $Q$ under $T$, i.e.\
$\tilde{q}(A)=Q(T\in A)$ for all Borel sets $A\subset(-\infty,+\infty]$.
Hence for any measurable $g:(-\infty,+\infty]\to[0,\infty]$, the defining change-of-variables identity for pushforward measures gives
\[
\int_{(-\infty,+\infty]} g(t)\,\tilde{q}(dt) = \int_{\mathcal{X}} g(T(x))\,Q(dx)
\]
Apply this with $g(t)=e^{-\alpha t}$ (using $e^{-\alpha\cdot(+\infty)}=0$ for $\alpha>0$ and $+\infty$ for $\alpha<0$, and allowing the value $+\infty$). We obtain
\[
\int_{(-\infty,+\infty]} e^{-\alpha t}\,\tilde{q}(dt)
=\int_{\mathcal{X}} e^{-\alpha T(x)}\,Q(dx)
=\int_{\mathcal{X}} \big(e^{-T(x)}\big)^{\alpha}\,Q(dx)
\]
Since $T(x)=-\log r(x)$, $e^{-T(x)} = r(x) = p(x)/q(x)$ for $Q$-a.e.\ $x$ (with $e^{-T}=0$ on the $+\infty$ atom, consistent with $r=p/q=0$ there); therefore
\[
\int_{(-\infty,+\infty]} e^{-\alpha t}\,\tilde{q}(dt)
=\int_{\mathcal{X}} r(x)^{\alpha}\,Q(dx)
=\int_{\mathcal{X}} \big(p(x)/q(x)\big)^{\alpha} q(x)\,\nu(dx)
=\int_{\mathcal{X}} p(x)^{\alpha}q(x)^{1-\alpha}\,\nu(dx)
= Z_{p,q}(\alpha)
\]
as claimed.

\textbf{Step 2: $\tilde p=e^{-t}\tilde{q}$ is a probability measure and equals $T_{\#}P$.}
Define $\tilde p(dt):=e^{-t}\,\tilde{q}(dt)$.
Then $\tilde p$ is a (finite) Borel measure with $\tilde p\ll\tilde{q}$. Its total mass is
\[
\tilde p(\R)
=\int_{\R} e^{-t}\,\tilde{q}(dt)
=\int_{\mathcal{X}} e^{-T(x)}\,Q(dx)
=\int_{\mathcal{X}} r(x)\,Q(dx)
=\int_{\mathcal{X}} dP
=1
\]
so $\tilde p$ is a probability measure.
To identify $\tilde p$ as $T_{\#}P$, take any Borel $A\subset\R$ and compute
\begin{align*}
\tilde p(A)
&= \int_A e^{-t}\,\tilde{q}(dt)
= \int_{\mathcal{X}} \mathbf{1}_{T(x)\in A}\, e^{-T(x)}\,Q(dx)
= \int_{\mathcal{X}} \mathbf{1}_{T(x)\in A}\, r(x)\,Q(dx) \\
&= \int_{\mathcal{X}} \mathbf{1}_{T(x)\in A}\, P(dx)
= P(T\in A) = (T_{\#}P)(A)
\end{align*}

\textbf{Step 3: Interconvertibility with the canonical pair.}
Consider the deterministic Markov kernel $K$ from $\mathcal{X}$ to $(-\infty,+\infty]$ induced by $T$,
namely $K(x,\cdot)=\delta_{T(x)}$. Then $QK=T_{\#}Q=\tilde{q}$ and $PK=T_{\#}P=\tilde p$.
For the reverse direction, work at the level of Markov operators on $L^{1}$-spaces:
\[
R: L^{1}((-\infty,+\infty],\tilde{q})\to L^{1}(\mathcal{X},Q),\qquad
(Rh)(x):=h(T(x))
\]
Positivity is immediate. For any $h\in L^{1}(\R,\tilde{q})$,
\[
\int_{\mathcal{X}} (Rh)(x)\,Q(dx)
=\int_{\mathcal{X}} h(T(x))\,Q(dx)
=\int_{\R} h(t)\,\tilde{q}(dt)
\]
so $R$ preserves integrals and is a stochastic map (Markov operator).
$R\,1=1$ $Q$-a.e., and $R\big(t\mapsto e^{-t}\big)(x)=e^{-T(x)}=r(x)=\frac{dP}{dQ}(x)$, so $R$ sends the canonical dichotomy $(\tilde p,\tilde{q})$ back to $(P,Q)$ (equivalently, to $(p,q)$ as densities w.r.t.\
$\nu$). This proves interconvertibility.

\textbf{Step 4: Uniqueness from an open interval of Laplace data.}
We use the following injectivity fact for the two-sided Laplace transform.

\smallskip\noindent \emph{Lemma (two-sided Laplace injectivity).}
Let $\mu_{1},\mu_{2}$ be finite Borel measures on $\R$.
If there exist $a<b$ such that $\int_{\R} e^{-\alpha t}\,\mu_{i}(dt)<\infty$ for all $\alpha\in(a,b)$ and the two integrals agree on $(a,b)$, then $\mu_{1}=\mu_{2}$.

\begin{proof}[Proof of the lemma]
Let $\mu:=\mu_{1}-\mu_{2}$ (a finite signed measure).
The map $F(\alpha):=\int_{\R} e^{-\alpha t}\,\mu(dt)$ is analytic on the vertical strip $\{\alpha\in\mathbb{C}:a<\Re(\alpha)<b\}$ and vanishes on the real interval $(a,b)$, hence vanishes identically on the strip.
Fix $c\in(a,b)$ and define the finite signed measure $\nu_{c}(dt):=e^{-ct}\,\mu(dt)$. Its Fourier transform is
\[
\widehat{\nu_{c}}(\xi)
=\int_{\R} e^{-i\xi t}\,\nu_{c}(dt)
=\int_{\R} e^{-(c+i\xi)t}\,\mu(dt)
=F(c+i\xi)=0\qquad\forall\,\xi\in\R
\]
Injectivity of the Fourier transform on finite measures gives $\nu_{c}=0$, hence $\mu=0$ and $\mu_{1}=\mu_{2}$.
\end{proof}

\smallskip
Applied to $\mu_{i}$ defined as the restriction of $\tilde{q}_{i}$ to $\R$ (and noting that any atom at $+\infty$ contributes $0$ for $\alpha>0$ and forces the integral to be $+\infty$ for $\alpha<0$), the lemma forces $\tilde{q}_{1}=\tilde{q}_{2}$ given equality of Laplace transforms on a non-empty open interval.
Since $\tilde p_{i}=e^{-t}\tilde{q}_{i}$, uniqueness of $\tilde p$ follows.
\end{proof}

Writing $\ell_{p}(x):=-\log p(x)$ and $\ell_{q}(x):=-\log q(x)$, the statistic $T(x)=-\log(p(x)/q(x))$ is exactly the difference of log-losses, $T(x)=\ell_{p}(x)-\ell_{q}(x)$, so the theorem says that the two-prior partition function $Z_{p,q}(\alpha)$ is the Laplace transform of the law of this single sufficient statistic under $q$; equivalently, $\Phi_{p,q}(\alpha)=\log Z_{p,q}(\alpha)$ is the cumulant generating function of $-T$ under $q$.
On the canonical (extended) line, the pair $(\tilde p,\tilde{q})$ has the simple density relation $d\tilde p/d\tilde{q}=e^{-t}$
(with $e^{-t}=0$ at any atom of $\tilde{q}$ at $+\infty$, so $\tilde p$ lives on $\R$),
so $\tilde p$ is an exponential tilt of $\tilde{q}$ with sufficient statistic $t$:
the one-dimensional thermodynamic normalization of the binary experiment.
If $\mathcal{X}$ is countable and $\nu$ is counting measure, then $\tilde{q} = \sum_{x\in\mathcal{X}} q(x)\,\delta_{-\log(p(x)/q(x))}$ and $\tilde p = \sum_{x\in\mathcal{X}} p(x)\,\delta_{-\log(p(x)/q(x))}$, and $Z_{p,q}(\alpha)=\sum_{x} q(x)\,(p(x)/q(x))^{\alpha}$ is the Laplace transform of this atomic measure.
Since $Z_{p,q}(1)=1$, the formula $R_\alpha=\Phi(\alpha)/(\alpha-1)$ is a $0/0$ indeterminate form at $\alpha=1$; taking the limit yields the usual relative entropy $\KLdiv(p\|q)=\int\log(p/q)\,dP=-\int t\,\tilde p(dt)$ whenever finite.

\subsection{Multi-way Laplace normal form}\label{ssec:multiway-laplace-proof}

The basic object is the partition function $Z(\alpha)=\E_{X\sim\nu}\big[\prod_{k=1}^{W}\pi_{k}(X)^{\alpha_{k}}\big]=H_{\alpha}(\boldsymbol{\pi})$,
whose negative logarithm is the multi-way coincidence divergence $\Cdiv_{\alpha}$ on the affine slice $\mathcal{A}=\{\sum_{k}\alpha_{k}=1\}$.
Paralleling Theorem~13 of \cite{galke2024sufficiencyrenyi} in the binary case, $Z(\alpha)$ is exactly a multivariate Laplace transform of the joint law of the per-prior log-losses $\ell_{k}(x):=-\log\pi_{k}(x)$ under $\nu$.
Pushing $\nu$ forward by the log-loss map gives an explicit measure $\tilde{q}$, and the proof below combines a change of variables with the pushforward and pullback maps used for Theorem~\ref{thm:binary-laplace}.

\begin{theorem}[Multi-way Laplace normal form for the Hellinger transform]
\label{thm:multiway-laplace-normal-form}
Let $(\mathcal{X},\Sigma,\nu)$ be a finite measure space.
Let $\pi_{1},\dots,\pi_{W}$ be measurable probability densities with respect to $\nu$, i.e.\
$\pi_{k}:\mathcal{X}\to[0,\infty)$ and $\int_{\mathcal{X}}\pi_{k}\,d\nu=1$ for each $k\in[W]$.
Define the \emph{log-loss vector} map $\ell:\mathcal{X}\to(-\infty,+\infty]^{W}$ by
\[
\ell(x):=(\ell_{1}(x),\dots,\ell_{W}(x)),\qquad
\ell_{k}(x):=-\log\pi_{k}(x)
\]
with the convention $-\log 0:=+\infty$ (each coordinate is $+\infty$ exactly where $\pi_{k}=0$, and never $-\infty$ since $\pi_{k}<\infty$ $\nu$-a.e.).
Write $\overline{\R}_{+\infty}^{W}:=(-\infty,+\infty]^{W}$ for the codomain.
Define the \emph{Laplace-transform measure} $\tilde{q}$ on $\overline{\R}_{+\infty}^{W}$ explicitly as the pushforward
\[
\tilde{q}:=\ell_{*}(\nu),\qquad
\tilde{q}(A) = \nu\!\left(\ell^{-1}(A)\right)\quad\forall\,\text{Borel }A\subseteq\overline{\R}_{+\infty}^{W}
\]
which can place mass on the faces $\{t_{k}=+\infty\}$ whenever $\nu(\pi_{k}=0)>0$.
Throughout, $e^{-\langle\alpha,t\rangle}$ at a point with some $t_{k}=+\infty$ is read as $0$ if $\alpha_{k}>0$ for some such $k$, as $+\infty$ if instead $\alpha_{k}<0$ for some such $k$, and through the finite coordinates alone if $\alpha_{k}=0$ for every such $k$.
This matches $\prod_{k}\pi_{k}^{\alpha_{k}}$ under $0^{\,c}=0$ for $c>0$, $0^{\,c}=+\infty$ for $c<0$, $0^{0}=1$ and $0\cdot\infty=0$.
Strict positivity $\pi_{k}>0$ $\nu$-a.e.\ for all $k$ removes every such face and returns $\tilde{q}$ to $\R^{W}$.
Then:

\smallskip\noindent \textbf{(1) Multivariate Laplace representation of $Z(\alpha)$.}
For every $\alpha\in\R^{W}$ for which the integrals are finite,
\begin{equation}\label{eq:multiway-laplace-Z}
Z(\alpha)
= H_{\alpha}(\boldsymbol{\pi})
:= \E_{X\sim\nu}\!\Big[\prod_{k=1}^{W}\pi_{k}(X)^{\alpha_{k}}\Big]
= \int_{\overline{\R}_{+\infty}^{W}} e^{-\langle\alpha,t\rangle}\,d\tilde{q}(t)
\end{equation}
where $\langle\alpha,t\rangle:=\sum_{k=1}^{W}\alpha_{k} t_{k}$ (for $\alpha$ with every $\alpha_{k}>0$ the $\{t_{k}=+\infty\}$ faces contribute $0$, so the integral equals its restriction to $\R^{W}$).

\smallskip\noindent \textbf{(2) Canonical pushed-forward priors.}
For each $k\in[W]$, define a measure $\tilde{\pi}_{k}$ on $\overline{\R}_{+\infty}^{W}$ by either of the equivalent formulas $\tilde{\pi}_{k}:=\ell_{*}(\pi_{k}\nu)$ or $d\tilde{\pi}_{k}(t)=e^{-t_{k}}\,d\tilde{q}(t)$.
Then each $\tilde{\pi}_{k}$ is a probability measure (indeed $\int e^{-t_{k}}\,d\tilde{q}(t)=1$),
assigning zero mass to its own face $\{t_{k}=+\infty\}$; it is supported on $\R^{W}$ if $\pi_{j}>0$ $\nu$-a.e.\ for the \emph{other} coordinates $j\ne k$ (otherwise it can charge a face $\{t_{j}=+\infty\}$, $j\ne k$), and $d\tilde{\pi}_{k}/d\tilde{q}(t)=e^{-t_{k}}$ holds $\tilde{q}$-a.e.

\smallskip\noindent \textbf{(3) Canonical form of the geometric mixture.}
For any $\alpha$ with $0<Z(\alpha)<\infty$, the geometric-mixture (product-of-powers) density on $\mathcal{X}$
\[
p^{\star}_{\alpha}(x) := \frac{1}{Z(\alpha)}\prod_{k=1}^{W}\pi_{k}(x)^{\alpha_{k}}
\]
pushes forward under $\ell$ to the canonical density (with respect to $\tilde{q}$ on the extended box $\overline{\R}_{+\infty}^{W}$)
\[
\tilde p^{\star}_{\alpha}(t) := \frac{e^{-\langle\alpha,t\rangle}}{Z(\alpha)}
\quad\text{with respect to }\tilde{q}
\]
This density vanishes on any face $\{t_{k}=+\infty\}$ with $\alpha_{k}>0$ (so it is carried by $\R^{W}$ whenever $\alpha$ is interior to the simplex).

\smallskip\noindent \textbf{(4) Multi-way coincidence divergence as a Laplace transform.}
For $\alpha\in\Aplus$,
\begin{equation}\label{eq:multiway-laplace-C}
\Cdiv_{\alpha}(\boldsymbol{\pi})
= -\log Z(\alpha)
= -\log\int_{\overline{\R}_{+\infty}^{W}} e^{-\langle\alpha,t\rangle}\,d\tilde{q}(t)
\end{equation}

\smallskip\noindent \textbf{(5) Interconversion with the canonical experiment.}
Define linear maps
\[
T:L^{1}(\mathcal{X},\nu)\to L^{1}(\overline{\R}_{+\infty}^{W},\tilde{q}),\qquad
Tg := \frac{d(\ell_{*}(g\nu))}{d\tilde{q}}
\]
and
\[
R:L^{1}(\overline{\R}_{+\infty}^{W},\tilde{q})\to L^{1}(\mathcal{X},\nu),\qquad
Rh := h\circ \ell
\]
Then $T$ and $R$ are stochastic maps (positive and integral-preserving), and they satisfy
\[
T(1)=1,\qquad R(1)=1,\qquad
T(\pi_{k})=e^{-t_{k}},\qquad R(e^{-t_{k}})=\pi_{k}\quad(k=1,\dots,W)
\]
so the $W$-tuple $(\pi_{1},\dots,\pi_{W})$ on $(\mathcal{X},\nu)$ is interconvertible with the canonical $W$-tuple $(e^{-t_{1}},\dots,e^{-t_{W}})$ on $(\overline{\R}_{+\infty}^{W},\tilde{q})$.%
\end{theorem}

\begin{proof}[Proof of Theorem~\ref{thm:multiway-laplace-normal-form}.]
For $\alpha\in\R^{W}$,
\[
\prod_{k=1}^{W}\pi_{k}(x)^{\alpha_{k}}
= \exp\!\Big(\sum_{k}\alpha_{k}\log\pi_{k}(x)\Big)
= \exp\!\Big(-\sum_{k}\alpha_{k}\,\ell_{k}(x)\Big)
= e^{-\langle\alpha,\ell(x)\rangle}
\]
Therefore, by the defining property of pushforwards (valid for the extended-real-valued $\ell$),
\[
Z(\alpha)
=\int_{\mathcal{X}} e^{-\langle\alpha,\ell(x)\rangle}\,d\nu(x)
=\int_{\overline{\R}_{+\infty}^{W}} e^{-\langle\alpha,t\rangle}\,d(\ell_{*}\nu)(t)
=\int_{\overline{\R}_{+\infty}^{W}} e^{-\langle\alpha,t\rangle}\,d\tilde{q}(t)
\]
which is \eqref{eq:multiway-laplace-Z}.

Fix $k$ and a Borel $A\subseteq\overline{\R}_{+\infty}^{W}$.
By definition of pushforward and $\ell_{k}=-\log\pi_{k}$,
\[
\tilde{\pi}_{k}(A)
:= (\ell_{*}(\pi_{k}\nu))(A)
= \int_{\ell^{-1}(A)} \pi_{k}(x)\,d\nu(x)
= \int_{\ell^{-1}(A)} e^{-\ell_{k}(x)}\,d\nu(x)
\]
Applying the pushforward identity again to the function $t\mapsto e^{-t_{k}}\mathbf{1}_{A}(t)$ yields
\[
\tilde{\pi}_{k}(A)
= \int_{\overline{\R}_{+\infty}^{W}} e^{-t_{k}}\mathbf{1}_{A}(t)\,d\tilde{q}(t)
= \int_{A} e^{-t_{k}}\,d\tilde{q}(t)
\]
so $d\tilde{\pi}_{k}(t)=e^{-t_{k}}\,d\tilde{q}(t)$ and $d\tilde{\pi}_{k}/d\tilde{q}=e^{-t_{k}}$
(the factor $e^{-t_{k}}$ vanishes on the face $\{t_{k}=+\infty\}$, so $\tilde{\pi}_{k}$ charges no mass there).
Taking $A=\overline{\R}_{+\infty}^{W}$ gives $\tilde{\pi}_{k}(\overline{\R}_{+\infty}^{W})=\int_{\mathcal{X}}\pi_{k}\,d\nu=1$, so $\tilde{\pi}_{k}$ is a probability measure (on the extended box; on $\R^{W}$ if $\pi_{j}>0$ $\nu$-a.e.\ for $j\ne k$).

Since $p^{\star}_{\alpha}\,\nu$ has density $\frac{1}{Z(\alpha)}\prod_{k}\pi_{k}^{\alpha_{k}}$ w.r.t.\
$\nu$, pushing it forward gives
\[
\ell_{*}(p^{\star}_{\alpha}\nu)(A)
= \int_{\ell^{-1}(A)} \frac{1}{Z(\alpha)}\prod_{k}\pi_{k}(x)^{\alpha_{k}}\,d\nu(x)
= \int_{A} \frac{e^{-\langle\alpha,t\rangle}}{Z(\alpha)}\,d\tilde{q}(t)
\]
i.e.\ $\tilde p^{\star}_{\alpha}(t)=e^{-\langle\alpha,t\rangle}/Z(\alpha)$ is the density of the pushforward w.r.t.\ $\tilde{q}$.
For $\alpha\in\Aplus$, $\Cdiv_{\alpha}=-\log Z(\alpha)$,
so \eqref{eq:multiway-laplace-C} follows from \eqref{eq:multiway-laplace-Z}.

For $T$: $\ell_{*}(g\nu)\ll\ell_{*}\nu=\tilde{q}$ for $g\in L^{1}(\mathcal{X},\nu)$, and $Tg$ is the Radon--Nikodym derivative of $\ell_{*}(g\nu)$ w.r.t.\ $\tilde{q}$. Positivity is immediate, and
\[
\int_{\overline{\R}_{+\infty}^{W}} Tg\,d\tilde{q}
= \ell_{*}(g\nu)(\overline{\R}_{+\infty}^{W})
= \int_{\mathcal{X}} g\,d\nu
\]
so $T$ is integral-preserving (a stochastic map).
Similarly $R$ is positive and, since $\tilde{q}=\ell_{*}\nu$,
\[
\int_{\mathcal{X}} Rh\,d\nu
=\int_{\mathcal{X}} h(\ell(x))\,d\nu(x)
=\int_{\overline{\R}_{+\infty}^{W}} h(t)\,d\tilde{q}(t)
\]
so $R$ is also integral-preserving. $T(\pi_{k})$ is the density of $\ell_{*}(\pi_{k}\nu)=\tilde{\pi}_{k}$ w.r.t.\ $\tilde{q}$, hence $T(\pi_{k})=e^{-t_{k}}$, and $R(e^{-t_{k}})(x)=e^{-\ell_{k}(x)}=\pi_{k}(x)$.
\end{proof}

Under $X\sim\nu$, the random vector $T:=\ell(X)=(-\log\pi_{1}(X),\dots,-\log\pi_{W}(X))$ is the vector of per-prior log-losses that drive the Hellinger transform; the measure $\tilde{q}$ is exactly the law of this log-loss vector, and $\Phi(\alpha):=\log Z(\alpha)$ is the cumulant generating function of $-T$, $\Phi(\alpha)=\log\E_{\tilde{q}}[e^{-\langle\alpha,T\rangle}]$.
Restricted to two coordinates it is the moment generating function of a pairwise log-likelihood ratio, the object of~\eqref{eq:cgf-renyi}.
The values of $Z$ on the hyperplane $\mathcal{A}$, where every atom lives, determine less than the values off it.

\begin{lemma}[Multivariate Laplace injectivity]\label{lem:laplace-injective}
Let $\mu_{1},\mu_{2}$ be finite Borel measures on $\R^{d}$ and $O\subset\R^{d}$ a nonempty open set with $\int e^{-\langle a,t\rangle}\,d\mu_{i}(t)<\infty$ for every $a\in O$ and $i=1,2$.
If the two integrals agree for every $a\in O$, then $\mu_{1}=\mu_{2}$.
\end{lemma}

\begin{proposition}[What the hyperplane determines]\label{prop:hyperplane}
Let $\boldsymbol{\pi}$ be a bounded $W$-tuple on $(\mathcal{X},\nu)$ and $L:=(\log(\pi_{k}/\pi_{1}))_{k=2}^{W}$.
\begin{enumerate}[label=(\roman*)]
\item For $\alpha\in\mathcal{A}$, $Z(\alpha)=\E_{\pi_{1}}[\exp\sum_{k\ge 2}\alpha_{k}L_{k}]$, and the values of $Z$ on any nonempty open subset of $\mathcal{A}$ determine $Z$ on $\mathcal{A}$, every atom $\Phi_{\xi}(\boldsymbol{\pi})$ with $\xi\in\paramset$, and the law of $L$ under $\pi_{1}$.
\item The values of $Z$ on $\mathcal{A}$ do not determine $\tilde q$.
Suppose $\mathcal{X}$ has two points, $\nu$ is uniform and some $\pi_{k}$ is not uniform.
Splitting one point in two, each half with half of its probability under every $\pi_{k}$, with $\nu$ uniform on the three points, leaves $Z$ unchanged on $\mathcal{A}$ and changes $\tilde q$.
\item If $\nu$ is a probability measure, $\pi_{k}>0$ $\nu$-almost everywhere for every $k$, and $Z$ is finite on a nonempty open set $O\subset\R^{W}$, then the values of $Z$ on $O$ determine $\tilde q$.
\end{enumerate}
\end{proposition}

The simplex, a subset of $\mathcal{A}$, therefore already determines every atom of a bounded tuple, including those indexed off the simplex.
The values of $Z$ on $\mathcal{A}$ depend on the tuple only through the law of $L$ under $\pi_{1}$, while $\tilde q$ depends also on the sample space, on the reference measure, and on which tuple with that law is used, and no atom sees any of these.
This holds at every $W\ge 2$; for $W=2$ the law of $L$ under $\pi_{1}$ is the one-dimensional law of Theorem~\ref{thm:binary-laplace}.

\begin{proof}[Proof of Lemma~\ref{lem:laplace-injective}]
Let $O'\subset O$ be an open ball.
On the tube $O'+i\R^{d}\subset\mathbb{C}^{d}$ let $F_{i}(z):=\int e^{-\langle z,t\rangle}\,d\mu_{i}(t)$.
For $\mathrm{Re}\,z$ in a closed cube $Q$ inside $O'$, $|e^{-\langle z,t\rangle}|=e^{-\langle\mathrm{Re}\,z,t\rangle}$ is at most the sum of $e^{-\langle v,t\rangle}$ over the vertices $v$ of $Q$, because $x\mapsto e^{-\langle x,t\rangle}$ is convex, and this sum is integrable.
So $F_{i}$ is continuous on the tube over the interior of $Q$ by dominated convergence, holomorphic in each variable by Morera's and Fubini's theorems, and hence holomorphic by Osgood's lemma.
The difference $F_{1}-F_{2}$ is holomorphic on the connected tube and vanishes on the real open set $O'$, so all its derivatives vanish at a point of $O'$ and it vanishes on the tube.
Fix $c\in O'$.
The finite measures $e^{-\langle c,t\rangle}\mu_{i}(dt)$ then have equal Fourier transforms $\omega\mapsto F_{i}(c+i\omega)$, so they are equal, and $\mu_{1}=\mu_{2}$.
\end{proof}

\begin{proof}[Proof of Proposition~\ref{prop:hyperplane}]
(i) Since $\sum_{k}\alpha_{k}=1$, $\prod_{k}\pi_{k}^{\alpha_{k}}=\pi_{1}\exp\sum_{k\ge 2}\alpha_{k}L_{k}$, which gives the formula.
As $L$ is bounded, $(\alpha_{2},\dots,\alpha_{W})\mapsto\E_{\pi_{1}}[\exp\sum_{k\ge 2}\alpha_{k}L_{k}]$ extends to an entire function on $\mathbb{C}^{W-1}$, and a real-analytic function on $\R^{W-1}$ is determined by its values on a nonempty open set.
The values of $Z$ on $\mathcal{A}$ give $\Dren_{\alpha}=\log Z(\alpha)/(\alpha_{\star}-1)$ for $\alpha\in(\Aplus\cup\Aminus)\setminus E$, the tropical atoms by Lemma~\ref{lem:tropical-limit}, and the relative-entropy atoms by Lemma~\ref{lem:vertex-derivation}.
They also give $\E_{\pi_{1}}[e^{\langle a,L\rangle}]$ for every $a\in\R^{W-1}$, which determines the law of $L$ under $\pi_{1}$ by Lemma~\ref{lem:laplace-injective}.

(ii) With $\nu$ uniform on a finite set $\mathcal{X}$ of size $m$ and probability vectors $P_{k}$, the densities are $\pi_{k}=mP_{k}$ and, for $\alpha\in\mathcal{A}$, $Z(\alpha)=\sum_{x}\frac1m\prod_{k}(mP_{k}(x))^{\alpha_{k}}=\sum_{x}\prod_{k}P_{k}(x)^{\alpha_{k}}$.
Splitting a point $x_{0}$ in two, each half with probability $P_{k}(x_{0})/2$ under $P_{k}$, replaces its term by $2\prod_{k}(P_{k}(x_{0})/2)^{\alpha_{k}}=\prod_{k}P_{k}(x_{0})^{\alpha_{k}}$, so $Z$ is unchanged on $\mathcal{A}$.
With two points, $\tilde q$ gives mass $\tfrac12$ to each of $\ell(x_{1})$ and $\ell(x_{2})$, which are distinct because some $P_{k}$ is not uniform, while after splitting $x_{2}$ every atom of $\tilde q$ has a mass that is a multiple of $\tfrac13$, so the two laws differ.

(iii) The measure $\tilde q=\ell_{*}\nu$ is a probability measure on $\R^{W}$, and $Z(a)=\int e^{-\langle a,t\rangle}\,d\tilde q(t)$ by Theorem~\ref{thm:multiway-laplace-normal-form}, so Lemma~\ref{lem:laplace-injective} applies.
\end{proof}

\subsection{An exact free-energy identity for the Laplace-mixed posterior}\label{ssec:laplace-mixed-concentration}

Fix a finite alphabet $\mathcal{X}$, priors $\pi_{1},\dots,\pi_{W}$ (not necessarily normalized), and an exponent vector $\alpha\in\R^{W}$ with $\prod_{k}\pi_{k}(x)^{\alpha_{k}}>0$ for every $x\in\mathcal{X}$, so that $0<Z(\alpha)<\infty$ and the typical distribution $p_{\alpha}^{\star}\propto\prod_{k}\pi_{k}^{\alpha_{k}}$ has full support.
Let $\sigma$ be a probability measure on the simplex $\Delta(\mathcal{X})$ with full support.
For $t>0$ the \emph{Laplace-mixed posterior} $\tilde{\mu}^{(t)}$ on $\Delta(\mathcal{X})$ has density
\begin{equation}\label{eq:laplace-measure}
\frac{d\tilde{\mu}^{(t)}}{d\sigma}(p) = \frac{1}{\mathcal{Z}_{t}}\exp\!\left(-t\,\Big[\sum_{k=1}^{W}\alpha_{k}\,H(p,\pi_{k}) - H(p)\Big]\right)
\end{equation}
with $H(p,\pi_{k})=\sum_{x}p(x)\log(1/\pi_{k}(x))$ the cross-entropy, $H(p)$ the Shannon entropy, and $\mathcal{Z}_{t}$ the normalizing constant.
It is the posterior over distributions given the priors and the log-loss constraints, and $t$ is the observation resolution.

\begin{theorem}[Exact free-energy identity]\label{thm:laplace-mixed-concentration}
Under the hypotheses above, for every $t>0$,
\begin{equation}\label{eq:free-energy-identity}
-\frac{1}{t}\log\mathcal{Z}_{t} \;=\; \Cdiv_{\alpha}(\boldsymbol{\pi}) \;+\; \E_{\tilde{\mu}^{(t)}}\!\big[\KLdiv(p\,\|\,p_{\alpha}^{\star})\big] \;+\; \frac{1}{t}\,\KLdiv\big(\tilde{\mu}^{(t)}\,\big\|\,\sigma\big),
\end{equation}
and both terms after $\Cdiv_{\alpha}(\boldsymbol{\pi})$ are nonnegative.
\end{theorem}

The identity holds exactly at every resolution.
The free energy of the posterior is the coincidence divergence plus two information costs: the expected relative entropy from the posterior's draws to the typical distribution, and the information the posterior gains over its reference, per unit of resolution.

\begin{corollary}[Concentration on the typical distribution]\label{cor:laplace-concentration}
Under the hypotheses above, as $t\to\infty$,
\[
-\frac{1}{t}\log\mathcal{Z}_{t}\;\to\;\Cdiv_{\alpha}(\boldsymbol{\pi}),
\qquad
\E_{\tilde{\mu}^{(t)}}\!\big[\KLdiv(p\,\|\,p_{\alpha}^{\star})\big]\;\to\;0,
\qquad
\tilde{\mu}^{(t)}\xrightarrow{w}\delta_{p_{\alpha}^{\star}}.
\]
\end{corollary}

Both information costs therefore vanish as resolution grows, and the free energy converges to the coincidence divergence itself.
This is a level-2 statement: it concerns the distributions drawn from the posterior, reaching beyond their means.
The next two results give the large-deviation principle with relative entropy to the typical distribution as rate function, and the rates at which the two costs vanish.

\begin{corollary}[Large deviations of the posterior]\label{cor:level2-ldp}
Under the hypotheses of Theorem~\ref{thm:laplace-mixed-concentration}, the posteriors $\tilde{\mu}^{(t)}$ satisfy the large-deviation principle on $\Delta(\mathcal{X})$ with speed $t$ and good rate function $\KLdiv(\cdot\,\|\,p_{\alpha}^{\star})$: for every Borel set $\Gamma\subset\Delta(\mathcal{X})$, with interior and closure taken in $\Delta(\mathcal{X})$,
\[
-\inf_{p\in\mathrm{int}\,\Gamma}\KLdiv(p\,\|\,p_{\alpha}^{\star})\le\liminf_{t\to\infty}\frac1t\log\tilde{\mu}^{(t)}(\Gamma)\le\limsup_{t\to\infty}\frac1t\log\tilde{\mu}^{(t)}(\Gamma)\le-\inf_{p\in\overline{\Gamma}}\KLdiv(p\,\|\,p_{\alpha}^{\star})
\]
\end{corollary}

\begin{proposition}[Rates of the two information costs]\label{prop:laplace-rates}
Under the hypotheses of Theorem~\ref{thm:laplace-mixed-concentration}, let $d:=|\mathcal{X}|-1$ and suppose moreover that $\sigma$ has a density $s$ with respect to Lebesgue measure in the coordinates $(p(x_{1}),\dots,p(x_{d}))$ that is continuous at $p_{\alpha}^{\star}$ with $s(p_{\alpha}^{\star})>0$.
Then, as $t\to\infty$,
\[
t\,\E_{\tilde{\mu}^{(t)}}\big[\KLdiv(p\,\|\,p_{\alpha}^{\star})\big]\;\to\;\frac{d}{2},
\qquad
\KLdiv\big(\tilde{\mu}^{(t)}\,\big\|\,\sigma\big)-\frac{d}{2}\log\frac{t}{2\pi e}\;\to\;-\log\Big(s(p_{\alpha}^{\star})\prod_{x\in\mathcal{X}}p_{\alpha}^{\star}(x)^{1/2}\Big)
\]
\end{proposition}

The expected relative entropy to the typical distribution is thus $d/(2t)$ to first order, $1/(2t)$ per dimension of the simplex, while the information gained over the reference, divided by $t$, is of order $(d/2)\log t/t$ and for large $t$ is the larger of the two costs.

\begin{proof}[Proof of Theorem~\ref{thm:laplace-mixed-concentration}]
Write $\mathcal{F}(p):=\sum_{k}\alpha_{k}H(p,\pi_{k})-H(p)$.
Since $\log p_{\alpha}^{\star}(x)=\sum_{k}\alpha_{k}\log\pi_{k}(x)-\log Z(\alpha)$ and $\sum_{x}p(x)=1$,
\[
\KLdiv(p\,\|\,p_{\alpha}^{\star})=-H(p)+\sum_{k}\alpha_{k}H(p,\pi_{k})+\log Z(\alpha),
\]
so $\mathcal{F}(p)=\Cdiv_{\alpha}(\boldsymbol{\pi})+\KLdiv(p\,\|\,p_{\alpha}^{\star})$, using $\Cdiv_{\alpha}=-\log Z(\alpha)$; only $p$ need be normalized here, so the priors need not be.
The Gibbs variational principle states that for a probability measure $\sigma$ and a measurable $g$ with $\int e^{g}\,d\sigma<\infty$,
$\log\int e^{g}\,d\sigma=\sup_{\nu}\{\E_{\nu}[g]-\KLdiv(\nu\,\|\,\sigma)\}$, the supremum attained at $d\nu/d\sigma\propto e^{g}$.
Take $g=-t\mathcal{F}$; the maximizer is $\tilde{\mu}^{(t)}$ by \eqref{eq:laplace-measure}, so
\[
\log\mathcal{Z}_{t}=-t\,\E_{\tilde{\mu}^{(t)}}[\mathcal{F}]-\KLdiv\big(\tilde{\mu}^{(t)}\,\big\|\,\sigma\big).
\]
Substituting $\E_{\tilde{\mu}^{(t)}}[\mathcal{F}]=\Cdiv_{\alpha}(\boldsymbol{\pi})+\E_{\tilde{\mu}^{(t)}}[\KLdiv(p\,\|\,p_{\alpha}^{\star})]$ and dividing by $-t$ gives \eqref{eq:free-energy-identity}.
Relative entropy is nonnegative, and $\sigma$ is a probability measure, so both terms after $\Cdiv_{\alpha}(\boldsymbol{\pi})$ are nonnegative.
\end{proof}

\begin{proof}[Proof of Corollary~\ref{cor:laplace-concentration}]
By the proof of Theorem~\ref{thm:laplace-mixed-concentration}, $\mathcal{Z}_{t}=Z(\alpha)^{t}\,I_{t}$ with $I_{t}:=\int\exp(-t\,\KLdiv(p\,\|\,p_{\alpha}^{\star}))\,d\sigma(p)$, so $-\frac{1}{t}\log\mathcal{Z}_{t}=\Cdiv_{\alpha}(\boldsymbol{\pi})-\frac{1}{t}\log I_{t}$.
Since the integrand is at most $1$ and $\sigma$ is a probability measure, $I_{t}\le 1$ and $-\frac{1}{t}\log I_{t}\ge 0$.
For the reverse bound, $p_{\alpha}^{\star}$ lies in the interior of $\Delta(\mathcal{X})$, where $\KLdiv(\cdot\,\|\,p_{\alpha}^{\star})$ is continuous and vanishes at $p_{\alpha}^{\star}$.
So for each $\varepsilon>0$ there is an open $U_{\varepsilon}\ni p_{\alpha}^{\star}$ on which it is below $\varepsilon$, and $\sigma(U_{\varepsilon})>0$ by full support.
Then $I_{t}\ge e^{-t\varepsilon}\sigma(U_{\varepsilon})$, so $-\frac{1}{t}\log I_{t}\le\varepsilon-\frac{1}{t}\log\sigma(U_{\varepsilon})$, whose limit superior is at most $\varepsilon$.
As $\varepsilon$ is arbitrary, $-\frac{1}{t}\log I_{t}\to 0$, which is the first limit.
By \eqref{eq:free-energy-identity} this same quantity equals $\E_{\tilde{\mu}^{(t)}}[\KLdiv(p\,\|\,p_{\alpha}^{\star})]+\frac{1}{t}\KLdiv(\tilde{\mu}^{(t)}\,\|\,\sigma)$, a sum of two nonnegative terms, so each tends to $0$; this is the second limit.
For the third, let $C\subset\Delta(\mathcal{X})$ be closed with $p_{\alpha}^{\star}\notin C$.
The simplex is compact and $\KLdiv(\cdot\,\|\,p_{\alpha}^{\star})$ is lower semicontinuous with unique zero $p_{\alpha}^{\star}$, so $\delta_{C}:=\inf_{C}\KLdiv(\cdot\,\|\,p_{\alpha}^{\star})>0$.
Markov's inequality gives $\tilde{\mu}^{(t)}(C)\le\E_{\tilde{\mu}^{(t)}}[\KLdiv(p\,\|\,p_{\alpha}^{\star})]/\delta_{C}\to 0$, and weak convergence to $\delta_{p_{\alpha}^{\star}}$ follows.
\end{proof}

\begin{proof}[Proof of Corollary~\ref{cor:level2-ldp}]
By the proof of Corollary~\ref{cor:laplace-concentration}, $\tilde{\mu}^{(t)}(\Gamma)=I_{t}^{-1}\int_{\Gamma}e^{-t\KLdiv(p\|p_{\alpha}^{\star})}\,d\sigma(p)$ with $I_{t}\le 1$ and $\frac1t\log I_{t}\to 0$.
For the upper bound, $\tilde{\mu}^{(t)}(\Gamma)\le I_{t}^{-1}\exp(-t\inf_{\Gamma}\KLdiv(\cdot\|p_{\alpha}^{\star}))$ and $\inf_{\Gamma}\ge\inf_{\overline{\Gamma}}$.
For the lower bound, let $p\in\mathrm{int}\,\Gamma$ and $\varepsilon>0$.
Since $p_{\alpha}^{\star}$ has full support, $\KLdiv(\cdot\,\|\,p_{\alpha}^{\star})$ is continuous on $\Delta(\mathcal{X})$, so there is an open set $V\ni p$ inside $\mathrm{int}\,\Gamma$ on which it is below $\KLdiv(p\|p_{\alpha}^{\star})+\varepsilon$, and $\sigma(V)>0$ because $\sigma$ has full support.
Then $\tilde{\mu}^{(t)}(\Gamma)\ge\sigma(V)\exp(-t(\KLdiv(p\|p_{\alpha}^{\star})+\varepsilon))$, so the limit inferior is at least $-\KLdiv(p\|p_{\alpha}^{\star})-\varepsilon$; taking the supremum over $p$ and letting $\varepsilon\to 0$ gives the lower bound.
The rate function is continuous on the compact set $\Delta(\mathcal{X})$, so its sublevel sets are compact.
\end{proof}

\begin{proof}[Proof of Proposition~\ref{prop:laplace-rates}]
Identify $p$ with $\theta:=(p(x_{1}),\dots,p(x_{d}))$ in the closed set $\Theta\subset\R^{d}$, write $\theta^{\star}$ for $p_{\alpha}^{\star}$, $s^{\star}:=s(\theta^{\star})$ and $f(\theta):=\KLdiv(p\,\|\,p_{\alpha}^{\star})$.
By the proof of Corollary~\ref{cor:laplace-concentration}, $d\tilde{\mu}^{(t)}/d\sigma=e^{-tf}/I_{t}$ with $I_{t}:=\int_{\Theta}e^{-tf}s\,d\theta$, so with $N_{t}:=\int_{\Theta}f\,e^{-tf}s\,d\theta$,
\begin{equation}\label{eq:rates-reduction}
t\,\E_{\tilde{\mu}^{(t)}}[f]=\frac{t\,N_{t}}{I_{t}},
\qquad
\KLdiv\big(\tilde{\mu}^{(t)}\,\big\|\,\sigma\big)=-t\,\E_{\tilde{\mu}^{(t)}}[f]-\log I_{t}
\end{equation}
Because $p_{\alpha}^{\star}$ has full support, $\theta^{\star}$ is interior to $\Theta$, $f$ is continuous on the compact set $\Theta$ with unique zero $\theta^{\star}$ and maximum $f_{\max}:=\max_{x}\log(1/p_{\alpha}^{\star}(x))$, and $f$ is smooth near $\theta^{\star}$ with $f(\theta^{\star})=0$, $\nabla f(\theta^{\star})=0$ and Hessian $A_{ij}=\delta_{ij}/p_{\alpha}^{\star}(x_{i})+1/p_{\alpha}^{\star}(x_{d+1})$.
The matrix $A$ is positive definite, and the matrix determinant lemma gives $\det A=\prod_{x}p_{\alpha}^{\star}(x)^{-1}$.
Let $q(v):=\tfrac12v^{\top}Av$, and for $b>0$ let $G_{0}(b):=\int_{\R^{d}}e^{-b\,q}\,dv=(2\pi/b)^{d/2}(\det A)^{-1/2}$ and $G_{1}(b):=\int_{\R^{d}}q\,e^{-b\,q}\,dv=\frac{d}{2b}\,G_{0}(b)$.

Fix $\eta\in(0,1)$.
By Taylor's theorem $f(\theta)-q(\theta-\theta^{\star})=o(|\theta-\theta^{\star}|^{2})$, while $q(v)\ge\tfrac{\lambda_{\min}}{2}|v|^{2}$ with $\lambda_{\min}>0$ the least eigenvalue of $A$; with continuity of $s$ at $\theta^{\star}$, there is $r>0$ such that the ball $B_{r}:=\{|\theta-\theta^{\star}|<r\}$ lies in $\Theta$ and on $B_{r}$
\[
(1-\eta)\,q(\theta-\theta^{\star})\le f(\theta)\le(1+\eta)\,q(\theta-\theta^{\star}),
\qquad
(1-\eta)\,s^{\star}\le s(\theta)\le(1+\eta)\,s^{\star}
\]
Off $B_{r}$, $c:=\min_{\Theta\setminus B_{r}}f>0$ and $\int s\,d\theta=1$, so $\Theta\setminus B_{r}$ contributes at most $e^{-tc}$ to $I_{t}$ and at most $f_{\max}e^{-tc}$ to $N_{t}$.
On $B_{r}$, substitute $v=\sqrt{t}(\theta-\theta^{\star})$, so that $t\,q(\theta-\theta^{\star})=q(v)$ and $d\theta=t^{-d/2}dv$:
\[
(1-\eta)\,s^{\star}t^{-d/2}\!\int_{|v|<r\sqrt{t}}\!e^{-(1+\eta)q}\,dv\;\le\;\int_{B_{r}}e^{-tf}s\,d\theta\;\le\;(1+\eta)\,s^{\star}t^{-d/2}\,G_{0}(1-\eta)
\]
\[
(1-\eta)^{2}s^{\star}t^{-1-d/2}\!\int_{|v|<r\sqrt{t}}\!q\,e^{-(1+\eta)q}\,dv\;\le\;\int_{B_{r}}f\,e^{-tf}s\,d\theta\;\le\;(1+\eta)^{2}s^{\star}t^{-1-d/2}\,G_{1}(1-\eta)
\]
The truncated integrals increase to $G_{0}(1+\eta)$ and $G_{1}(1+\eta)$ by monotone convergence, and $e^{-tc}=o(t^{-1-d/2})$.
Hence
\[
\frac{(1-\eta)^{2}\,G_{1}(1+\eta)}{(1+\eta)\,G_{0}(1-\eta)}\le\liminf_{t\to\infty}\frac{t\,N_{t}}{I_{t}}\le\limsup_{t\to\infty}\frac{t\,N_{t}}{I_{t}}\le\frac{(1+\eta)^{2}\,G_{1}(1-\eta)}{(1-\eta)\,G_{0}(1+\eta)}
\]
and $(1-\eta)s^{\star}G_{0}(1+\eta)\le\liminf_{t}t^{d/2}I_{t}\le\limsup_{t}t^{d/2}I_{t}\le(1+\eta)s^{\star}G_{0}(1-\eta)$.
Letting $\eta\downarrow 0$ and using the formulas for $G_{0}$ and $G_{1}$ gives $t\,N_{t}/I_{t}\to d/2$ and $t^{d/2}I_{t}\to s^{\star}G_{0}(1)=s^{\star}(2\pi)^{d/2}\prod_{x}p_{\alpha}^{\star}(x)^{1/2}$.
The first limit is the first claim, and substituting both limits into~\eqref{eq:rates-reduction} gives $\KLdiv(\tilde{\mu}^{(t)}\|\sigma)-\frac{d}{2}\log t\to-\frac{d}{2}-\log\big(s^{\star}(2\pi)^{d/2}\prod_{x}p_{\alpha}^{\star}(x)^{1/2}\big)$, which is the second claim.
\end{proof}

\subsection{Analytic consequences of the Laplace-transform normal form}\label{ssec:laplace-normal-form-explanation}

With the log-loss vector $\ell(x):=(-\log\pi_{1}(x),\dots,-\log\pi_{W}(x))$ and $\tilde{q}:=\ell_{\#}\nu$, the product $\prod_{k}\pi_{k}(x)^{\alpha_{k}}=e^{-\langle\alpha,\ell(x)\rangle}$ makes $Z(\alpha)=\int e^{-\langle\alpha,t\rangle}\,d\tilde{q}(t)$ a multivariate Laplace transform, so the classical analytic toolkit applies to the coincidence calculus directly.
For a finite alphabet with counting measure, $\tilde{q}$ is the point cloud of surprisal vectors, one atom per symbol, and $\Phi(\alpha)=\log Z(\alpha)$ is a soft minimum of the energy $E_{\alpha}(x):=\langle\alpha,\ell(x)\rangle$, increasingly dominated by the symbols that look likely under every prior at once as $\|\alpha\|$ grows along a ray.
The geometric mixture $p^{\star}_{\alpha}=e^{-E_{\alpha}}/Z(\alpha)$ is the Gibbs state of that energy, and pushing it through $\ell$ gives the exponential tilt
\[
\tilde q_{\alpha}(dt) := \ell_{\#}(p^{\star}_{\alpha}\nu)(dt) = \frac{e^{-\langle\alpha,t\rangle}}{Z(\alpha)}\,\tilde q(dt)
\]
whose free-energy reading is Theorem~\ref{thm:laplace-mixed-concentration}.

\paragraph{Cumulants.}
Where $0<Z(\alpha)<\infty$ and the moments exist, derivatives of $\Phi$ are cumulants of the log-loss vector under the tilt: $\nabla\Phi(\alpha)=-\E_{t\sim\tilde q_{\alpha}}[t]$, equivalently $\partial_{\alpha_{k}}\Phi(\alpha)=-\E_{x\sim p^{\star}_{\alpha}}[-\log\pi_{k}(x)]$, and the Hessian is a covariance,
\[
\frac{\partial^{2}\Phi}{\partial\alpha_{j}\partial\alpha_{k}}(\alpha)=\Cov_{p^{\star}_{\alpha}}\!\big(-\log\pi_{j},-\log\pi_{k}\big),
\qquad \nabla^{2}\Phi(\alpha)=\Cov_{\tilde q_{\alpha}}(t)\succeq 0
\]
So $\Phi$ is convex where finite, $\Cdiv_{\alpha}$ is concave on the simplex since $\nabla^{2}\Cdiv_{\alpha}=-\Cov_{\tilde q_{\alpha}}(t)\preceq 0$, mixed partials measure how the priors' log-losses co-vary under the geometric mixture, and higher derivatives supply the cumulant and saddlepoint expansions.
The same transform turns convolution into multiplication, which is the algebra behind tensorization ($H_{\alpha}$ multiplicative, $\log Z$ additive on product priors) and behind cascade models, where each round contributes an independent copy of the log-loss vector.

\paragraph{Complete monotonicity.}
When every $\ell_{k}\ge 0$, as for probabilities on a discrete alphabet, $\tilde q$ lives on $[0,\infty]^{W}$ and $Z$ is completely monotone on $(0,\infty)^{W}$: $(-1)^{|\kappa|}\partial^{\kappa}Z(\alpha)\ge 0$ for every multi-index $\kappa\in\bbN^{W}$.
Any approximation $\widehat Z$ built from such a $\tilde q$ must respect these signs, which makes them a consistency check.

\paragraph{Inversion.}
A Laplace transform known on a full-dimensional open set determines its measure (Lemma~\ref{lem:laplace-injective}).
Every atom of the representation lives on the hyperplane $\sum_{k}\alpha_{k}=1$, and the values of $Z$ there determine every atom and the law of the log-likelihood-ratio vector, but not $\tilde q$ (Proposition~\ref{prop:hyperplane}).
The values off the hyperplane add information about the reference measure $\nu$, which no atom sees.

\paragraph{Tail bounds and large deviations.}
Along a direction $u$, $Z(su)=\E[e^{-s\langle u,T\rangle}]$ for $T\sim\tilde q$, and Markov's inequality gives $\mathrm{Pr}(\langle u,T\rangle\le a)\le e^{sa}Z(su)$ whenever $Z(su)<\infty$; $Z$ therefore controls the mass in low-energy regions, which fixes coincidence probabilities and, after optimizing the exponent, hypothesis-testing error exponents.
Under G\"artner--Ellis conditions the rate function for empirical averages of $\ell$ is the Legendre--Fenchel conjugate of $\Phi$, which is the bridge from the exact identity $\Cdiv_{\alpha}=-\log H_{\alpha}$ to the error exponents of i.i.d.\ tensor-product experiments.

\paragraph{A language-model instance.}
Fix prompts $x^{(1)},\dots,x^{(W)}$ with next-token distributions $\pi_{k}(v)=p(v\mid x^{(k)})$ on a vocabulary $V$, so that each token is the point $\ell(v)=(-\log p(v\mid x^{(1)}),\dots,-\log p(v\mid x^{(W)}))\in\R^{W}_{\ge 0}$.
On the simplex,
\[
Z(\alpha)=\sum_{v\in V}\prod_{k=1}^{W} p(v\mid x^{(k)})^{\alpha_{k}}
\]
and $\Cdiv_{\alpha}=-\log Z(\alpha)$ is the multi-way overlap of the prompts, while $p^{\star}_{\alpha}$ reweights the vocabulary toward tokens plausible under every prompt at once.
The contrastive-decoding exponent $(1,-1)$ is the direction of the signed-exponent ray $\alpha=(1+t,-t)\subset\Aminus$, and the corresponding tropical direction lies in $\Bminus$; the simplex $\simplex$ contains neither.

\section{Proofs for Section~\ref{sec:conditional}}
\label{app:proofs-conditional-multi-way-coincidence-calcul}

Throughout, $G_{\alpha}(u):=\prod_{k}u_{k}^{\alpha_{k}}$ for $u\in[0,\infty)^{W}$, with $G_{\alpha}(0):=0$.
For $\alpha\in\Aplus$ it is concave on $[0,\infty)^{W}$ and positively homogeneous of degree one, hence superadditive: $G_{\alpha}(u+v)=2\,G_{\alpha}(\tfrac{u+v}{2})\ge G_{\alpha}(u)+G_{\alpha}(v)$.
For $\alpha\in\Aminus$ with $\alpha_{j}\ge 1$ and $\alpha_{k}\le 0$ for $k\ne j$, the function $\phi(v):=\prod_{k\ne j}v_{k}^{\alpha_{k}}$ is convex on $(0,\infty)^{W-1}$, because $\log\phi(v)=\sum_{k\ne j}(-\alpha_{k})(-\log v_{k})$ is convex, and its perspective $u\mapsto u_{j}\,\phi(u_{k}/u_{j}:k\ne j)=G_{\alpha}(u)$ is convex and positively homogeneous on $(0,\infty)^{W}$, hence subadditive there, and with $G_{\alpha}(0)=0$ subadditive on $(0,\infty)^{W}\cup\{0\}$.
For a bounded conditional tuple and $g$ with $p_{G}(g)>0$, the vectors $u_{g,x}:=(\pi_{1|G}(x|g),\dots,\pi_{W|G}(x|g))$ lie in $(0,\infty)^{W}\cup\{0\}$, because the components share their support.

\begin{proof}[Proof of Lemma~\ref{lem:cond-dpi-main}]
Fix $g$ with $p_{G}(g)>0$, write $K(y|x):=T_{Y|XG}(y|x,g)$ and $u_{x}:=u_{g,x}$.
For every output $y$, $\sum_{x}K(y|x)u_{x}$ lies in $(0,\infty)^{W}\cup\{0\}$, and positive homogeneity with superadditivity, for $\alpha\in\Aplus$, gives
\[
G_{\alpha}\Big(\sum_{x}K(y|x)\,u_{x}\Big)\ge\sum_{x}K(y|x)\,G_{\alpha}(u_{x})
\]
while subadditivity gives the reverse inequality for $\alpha\in\Aminus$.
Summing over $y$ and using $\sum_{y}K(y|x)=1$ gives $H_{\alpha}((T_{Y|XG}\circ\boldsymbol{\pi})_{|g})\ge H_{\alpha}(\boldsymbol{\pi}_{|g})$ for $\alpha\in\Aplus$ and $\le$ for $\alpha\in\Aminus$.
For $\lambda\in(0,\infty)$ the map $h\mapsto h^{1/\lambda}$ is increasing and $h\mapsto-\lambda\log h$ is decreasing, so averaging over $g\sim p_{G}$ gives the claimed inequality on $\Aplus$ and its reverse on $\Aminus$.
For $\lambda=\infty$ the same holds because $\Cdiv_{\alpha,\infty}(\cdot\,\|\,p_{G})=-\E_{g\sim p_{G}}[\log H_{\alpha}(\cdot_{|g})]$ and $\log$ is increasing.
\end{proof}

\begin{proof}[Proof of Lemma~\ref{lem:cond-dpi-cond}]
Fix $h$ with $q_{H}(h)>0$ and let $w(g|h):=p_{G}(g)\,t_{H|G}(h|g)/q_{H}(h)$, a probability vector in $g$, so that $\pi^{(k)}_{|H}(x|h)=\sum_{g}w(g|h)\,\pi_{k|G}(x|g)$.
For $\alpha\in\Aplus$, concavity of $G_{\alpha}$ at each $x$, summed over $x$, gives
\begin{equation}\label{eq:cond-dpi-mixture}
H_{\alpha}(\boldsymbol{\pi}_{|h})\;\ge\;\sum_{g}w(g|h)\,H_{\alpha}(\boldsymbol{\pi}_{|g})
\end{equation}
Let $\lambda\in[1,\infty)$.
The map $s\mapsto s^{1/\lambda}$ is increasing and concave on $[0,\infty)$, so~\eqref{eq:cond-dpi-mixture} and Jensen's inequality give $H_{\alpha}(\boldsymbol{\pi}_{|h})^{1/\lambda}\ge\sum_{g}w(g|h)\,H_{\alpha}(\boldsymbol{\pi}_{|g})^{1/\lambda}$.
Averaging over $h\sim q_{H}$, with $q_{H}(h)\,w(g|h)=p_{G}(g)\,t_{H|G}(h|g)$ and $\sum_{h}t_{H|G}(h|g)=1$,
\[
\E_{h\sim q_{H}}\big[H_{\alpha}(\boldsymbol{\pi}_{|h})^{1/\lambda}\big]\;\ge\;\E_{g\sim p_{G}}\big[H_{\alpha}(\boldsymbol{\pi}_{|g})^{1/\lambda}\big]
\]
and the decreasing map $s\mapsto-\lambda\log s$ gives the claim.
For $\lambda=\infty$, concavity of $\log$ and~\eqref{eq:cond-dpi-mixture} give $\log H_{\alpha}(\boldsymbol{\pi}_{|h})\ge\sum_{g}w(g|h)\log H_{\alpha}(\boldsymbol{\pi}_{|g})$, and the same averaging gives the claim.
\end{proof}

\begin{proof}[Proof of Proposition~\ref{prop:cond-strata}]
In each case the value is nonnegative and finite on bounded conditional tuples, and it vanishes on ground states because every unconditional atom does.
Write $w(g|h)$ as in the proof of Lemma~\ref{lem:cond-dpi-cond}.

(a) Lemmas~\ref{lem:cond-dpi-main} and~\ref{lem:cond-dpi-cond} give monotonicity under both kinds of processing, and additivity is the display preceding Proposition~\ref{prop:cond-strata}.

(b) Let $\alpha\in\Aminus\setminus E$ and $\lambda\in(0,1]$.
Subadditivity and positive homogeneity of $G_{\alpha}$ give $H_{\alpha}(\boldsymbol{\pi}_{|g})=\sum_{x}G_{\alpha}(u_{g,x})\ge G_{\alpha}(\sum_{x}u_{g,x})=G_{\alpha}(\one)=1$, so $\Cdiv_{\alpha,\lambda}\le 0$, and $1-\alpha_{\star}<0$, so the functional is nonnegative.
Lemma~\ref{lem:cond-dpi-main} shows that main-system processing raises $\Cdiv_{\alpha,\lambda}$, and dividing by $1-\alpha_{\star}<0$ turns this into monotonicity.
For conditioning-system processing, subadditivity and positive homogeneity of $G_{\alpha}$ at each $x$, summed over $x$, give $H_{\alpha}(\boldsymbol{\pi}_{|h})\le\sum_{g}w(g|h)\,H_{\alpha}(\boldsymbol{\pi}_{|g})$.
The map $s\mapsto s^{1/\lambda}$ is increasing and convex on $[0,\infty)$ for $\lambda\in(0,1]$, so $H_{\alpha}(\boldsymbol{\pi}_{|h})^{1/\lambda}\le\sum_{g}w(g|h)\,H_{\alpha}(\boldsymbol{\pi}_{|g})^{1/\lambda}$, and averaging over $h\sim q_{H}$ gives $\E_{h\sim q_{H}}[H_{\alpha}(\boldsymbol{\pi}_{|h})^{1/\lambda}]\le\E_{g\sim p_{G}}[H_{\alpha}(\boldsymbol{\pi}_{|g})^{1/\lambda}]$.
So $\Cdiv_{\alpha,\lambda}$ does not decrease under conditioning-system processing, and dividing by $1-\alpha_{\star}<0$ gives monotonicity.
Additivity holds as in case (a), since $H_{\alpha}$ is multiplicative on products and the expectation over $p_{G}\otimes p_{G'}$ factorizes.

(c) Normalize $\beta$ so that $\beta_{k}=\beta_{\star}=1$ for some $k$, which changes nothing, and let $w_{\ell}:=-\beta_{\ell}\ge 0$ for $\ell\ne k$, so that $\sum_{\ell\ne k}w_{\ell}=1$.
Under main-system processing each $\Tdiv_{\beta}(\boldsymbol{\pi}_{|g})$ decreases by the data-processing inequality for the tropical atom (Lemma~\ref{lem:atoms-are-divergences}), so their maximum decreases.
For conditioning-system processing let $M:=\exp\max_{g:p_{G}(g)>0}\Tdiv_{\beta}(\boldsymbol{\pi}_{|g})$, so that $\pi_{k|G}(x|g)\le M\prod_{\ell\ne k}\pi_{\ell|G}(x|g)^{w_{\ell}}$ for every $x$ and every $g$ with $p_{G}(g)>0$; the tropical atoms are nonnegative (Lemma~\ref{lem:atoms-are-divergences}), so the maximum is too.
For $h$ with $q_{H}(h)>0$, concavity of the weighted geometric mean gives
\[
\pi^{(k)}_{|H}(x|h)=\sum_{g}w(g|h)\,\pi_{k|G}(x|g)\le M\sum_{g}w(g|h)\prod_{\ell\ne k}\pi_{\ell|G}(x|g)^{w_{\ell}}\le M\prod_{\ell\ne k}\pi^{(\ell)}_{|H}(x|h)^{w_{\ell}}
\]
so $\Tdiv_{\beta}(\boldsymbol{\pi}_{|h})\le\log M$ for every such $h$.
On independent extensions the tropical atom of $(\boldsymbol{\pi}\otimes\boldsymbol{\pi}')_{|(g,g')}$ is $\Tdiv_{\beta}(\boldsymbol{\pi}_{|g})+\Tdiv_{\beta}(\boldsymbol{\pi}'_{|g'})$, and the maximum over pairs $(g,g')$ with $p_{G}(g)p_{G'}(g')>0$ of this sum is the sum of the two maxima.

(d) Main-system processing decreases each $\KLdiv(\pi_{k|G}(\cdot|g)\|\pi_{\ell|G}(\cdot|g))$ by the data-processing inequality.
For conditioning-system processing, joint convexity of the relative entropy gives $\KLdiv(\pi^{(k)}_{|H}(\cdot|h)\,\|\,\pi^{(\ell)}_{|H}(\cdot|h))\le\sum_{g}w(g|h)\,\KLdiv(\pi_{k|G}(\cdot|g)\,\|\,\pi_{\ell|G}(\cdot|g))$, and averaging over $h\sim q_{H}$ gives the average over $g\sim p_{G}$.
The relative entropy is additive on products, and the expectation over $p_{G}\otimes p_{G'}$ of a sum is the sum of the expectations.

With $|\mathcal{G}|=1$ the four functionals are $\Cdiv_{\alpha}$, $\Cdiv_{\alpha}/(1-\alpha_{\star})=\Dren_{\alpha}$, $\Tdiv_{\beta}$ and $\KLdiv(\pi_{k}\|\pi_{\ell})$.
\end{proof}

\begin{proof}[Proof of Corollary~\ref{cor:cond-forced}]
Suppose $\int\Cdiv_{\alpha,\lambda}\,d\widetilde\mu$ agreed with one of the divergences (b), (c) or (d) on all bounded conditional tuples.
With $|\mathcal{G}|=1$, $\Cdiv_{\alpha,\lambda}(\boldsymbol{\pi}\,\|\,1)=\Cdiv_{\alpha}(\boldsymbol{\pi})=(1-\alpha_{\star})\Dren_{\alpha}(\boldsymbol{\pi})$, so on bounded tuples on finite alphabets the functional $\int\Cdiv_{\alpha,\lambda}\,d\widetilde\mu$ is represented through~\eqref{eq:correct} by the finite measure $m^{\Dren}$ on $\Aplus\setminus E$ obtained by weighting the $\alpha$-marginal of $\widetilde\mu$ with $1-\alpha_{\star}\in(0,1)$, with no tropical or relative-entropy mass.
The restriction of the divergence (b), (c) or (d) is an atom of $\Aminus\setminus E$, of $\Bminus\setminus\{0\}$ or of the edges, whose only representation is the unit mass at that atom by Corollary~\ref{cor:necessity}; the two representations differ, contradicting Proposition~\ref{prop:uniqueness}.
\end{proof}

\section*{Acknowledgments}

We acknowledge the use of large language models in preparing this work.

\bibliographystyle{plainnat}
\bibliography{multiway_axiomatic}

\end{document}